\documentclass[aps,pra,10pt,twocolumn,superscriptaddress,notitlepage,nofootinbib,longbibliography]{revtex4-2}

\usepackage[utf8]{inputenc}
\usepackage[english]{babel}

\usepackage{amsmath,amssymb,amsfonts}
\usepackage{mathtools}
\usepackage{bm}
\usepackage{dsfont}
\usepackage{amsthm}
\usepackage{physics}

\usepackage{graphicx}
\usepackage[dvipsnames]{xcolor}

\usepackage{siunitx}

\usepackage[autostyle]{csquotes}
\MakeOuterQuote{"}
\usepackage[shortlabels]{enumitem}
\usepackage{comment}
\usepackage{verbatim}

\usepackage{xpatch}
\usepackage{etoolbox}
\usepackage{tcolorbox}
\usepackage{soul}

\usepackage[hidelinks]{hyperref}
\usepackage[capitalize]{cleveref}

\AtBeginDocument{\RenewCommandCopy\qty\SI}

\date{\today}

\definecolor{cadetgrey}{rgb}{0.57, 0.64, 0.69}

\newtheorem{lemma}{Lemma}
\newtheorem{theorem}{Theorem}
\theoremstyle{definition}
\newtheorem{remark}{Remark}

\newcommand{\set}[1]{\{#1\}}

\newcommand{\identity}{\mathds{1}}

\newcommand{\rv}[1]{\bm{#1}}
\newcommand{\chiT}[1]{\rv{\chi}_{\rm W}^{#1}}
\newcommand{\Tsind}{\rm W}

\newcommand{\vos}[1]{\vec{#1}}

\newcommand{\epsleak}{\varepsilon}
\newcommand{\epsproB}{\epsilon_{\rm pB}}
\newcommand{\epsproK}{\epsilon_{\rm pK}}
\newcommand{\epsproS}{\epsilon_{\rm pS}}

\newcommand{\ptrash}{\alpha}
\newcommand{\pB}[1]{p_{#1}^{\rm B}}
\newcommand{\piA}[1]{p_{#1}^{\rm A}}
\newcommand{\Tobs}[1]{\hat{W}_{#1}}
\newcommand{\Tobsn}[1]{\hat{W}_{#1}^{n}}

\newcommand{\Iav}[1]{I_{u}^{\rm av}(#1)}

\newcommand{\loss}{\perp}
\newcommand{\LC}{LWM}

\tcbuselibrary{theorems}

\tcbset{
  myboxcolor/.style={colframe=#1!25!black}
}

\newtcbtheorem[]{mybox}{Box}%
  {colback=black!8!white, colframe=blue!25!black,
   fontupper=\small, fonttitle=\bfseries}{box}

\newcommand{\nstates}{n_A}

\crefname{appendix}{Appendix}{Appendices}

\makeatletter
\AddToHook{cmd/appendix/before}{\def\cref@section@alias{appendix}\def\cref@subsection@alias{appendix}}
\makeatother

\begin{document}

\title{Numerical security analysis for practical quantum key distribution}

\author{Álvaro Navarrete}
\email{anavarrete@vqcc.uvigo.es}
\affiliation{Vigo Quantum Communication Center, University of Vigo, Vigo E-36310, Spain}
\affiliation{School of Telecommunication Engineering, Department of Signal Theory and Communications, University of Vigo, Vigo E-36310, Spain}
\affiliation{atlanTTic Research Center, University of Vigo, Vigo E-36310, Spain}

\author{Guillermo Currás-Lorenzo}
\affiliation{Vigo Quantum Communication Center, University of Vigo, Vigo E-36310, Spain}
\affiliation{School of Telecommunication Engineering, Department of Signal Theory and Communications, University of Vigo, Vigo E-36310, Spain}
\affiliation{atlanTTic Research Center, University of Vigo, Vigo E-36310, Spain}

\author{Margarida Pereira}
\affiliation{Vigo Quantum Communication Center, University of Vigo, Vigo E-36310, Spain}
\affiliation{School of Telecommunication Engineering, Department of Signal Theory and Communications, University of Vigo, Vigo E-36310, Spain}
\affiliation{atlanTTic Research Center, University of Vigo, Vigo E-36310, Spain}

\author{Marcos Curty}
\affiliation{Vigo Quantum Communication Center, University of Vigo, Vigo E-36310, Spain}
\affiliation{School of Telecommunication Engineering, Department of Signal Theory and Communications, University of Vigo, Vigo E-36310, Spain}
\affiliation{atlanTTic Research Center, University of Vigo, Vigo E-36310, Spain}

\begin{abstract}
Quantum key distribution (QKD) promises information-theoretic security based on quantum mechanics and idealized device models. Practical implementations, however, deviate from these models due to unavoidable device imperfections, and existing security proofs fall short of capturing the complexity of real-world systems. Here we introduce a versatile numerical finite-key security framework valid against general coherent attacks and applicable to a broad class of practical QKD setups. It accommodates most relevant imperfections at both transmitter and receiver, including non-independent-and-identically-distributed (non-IID) signals arising in high-speed QKD systems due to the limited bandwidth of optical modulators, while requiring only partial characterization of the apparatuses. We demonstrate the power of our framework by proving the security of a realistic decoy-state QKD implementation with laser sources, providing a practical route towards rigorous security certification of real-world QKD setups.
\end{abstract}

\maketitle

\section{Introduction}\label{sec:Introduction}
In principle, quantum key distribution (QKD)~\cite{lo2014secure,pirandola2020advances,xu2020secure} enables two distant parties to establish a shared secret key whose security is guaranteed by quantum mechanics. Its central premise is that any eavesdropping attempt on non-orthogonal quantum states necessarily induces disturbances, which leave observable signatures in the measurement record. In practice, however, device imperfections can undermine this guarantee and make the security of real QKD implementations uncertain~\cite{NSA_QKD,NCSC_QKD,ANSSI_BSI_NLNCSA_QKD}.

Over the past two decades, substantial progress has narrowed the gap between theory and practice. Modern analytical security proofs now accommodate a broad range of device imperfections in both sources and detectors~\cite{tamaki2014loss,pereira,pereira2020quantum,sixto2025quantum,tupkary2024phase,curras2025security,wangPhaseError2025,pereira2024quantum,zapatero2021security,xoel,curras2023security,curras2026securitydecoy,curras-lorenzoRigorousPhaseerrorestimation2026}. The most advanced such analysis~\cite{curras2023security2} can guarantee finite-key security against coherent attacks while requiring only partial device characterization. However, although this proof is tight for BB84 \cite{curras2025numerical,pereira2025optimal}, its techniques become loose for other protocols with non-qubit encoding \cite{navarrete2021practical}, as shown in \cite{curras2025numerical}. Moreover, its BB84 application assumes single-photon sources, and even though an extension to the decoy-state method~\cite{hwang2003quantum,lo2005decoy,wang2005beating} is possible in principle, it remains unclear how to incorporate decoy-state-specific imperfections within this framework. More fundamentally, analytical proofs are tailored to specific protocols and are difficult to extend to highly unstructured schemes with many transmitted states; even minor protocol changes typically require substantial modifications to the security analysis. This rigidity severely limits their adaptability to the diversity of QKD implementations.

Numerical security proofs have recently emerged as flexible alternatives to analytical methods~\cite{coles2016numerical,winick2018reliable,wang2019characterising,george2021numerical,bunandar2020numerical,kamin2025improved,kamin2024finite,zhou2022numerical,metger2023security,lorente2024quantum,navarro2025finite,tupkary2026rigorous}. They allow rapid evaluation across a broad range of QKD scenarios and often yield optimal or near-optimal asymptotic key rates, even for protocols lacking the symmetries on which analytical proofs usually rely. This flexibility makes them well suited to address device imperfections, which often break these symmetries. Along these lines, recent works have incorporated some device imperfections into numerical frameworks~\cite{kamin2025r,naharImperfectDetectors2026}, while others have extended numerical techniques to settings with only partial state characterization~\cite{curras2025numerical,pereira2025optimal}. The latter, however, remain confined to the collective-attack scenario. More critically, no existing numerical proof to date can deal with pulse correlations, i.e., scenarios in which the state emitted in one round may depend on previous setting choices. Such correlations typically arise from the finite bandwidth of optical modulators in high-speed systems~\cite{grunenfelder2020performance,agulleiro2025modeling,roberts2018patterning,yoshino2018quantum,trefilov2025intensity} and break the round-by-round independence of Alice's signal preparation on which existing numerical techniques largely rely.

In this work, we overcome these limitations by introducing a versatile numerical finite-key security framework that remains valid in the presence of non-independent-and-identically-distributed (non-IID) signals and is applicable to a broad range of practical QKD implementations. Our approach establishes security against coherent attacks while incorporating the most relevant device imperfections at both the transmitter and receiver, requiring only partial characterization of the apparatuses. The resulting bounds are tight and surpass state-of-the-art analytical security proofs~\cite{curras2023security2} across all considered regimes where the latter apply. As a concrete demonstration, we establish the security of a decoy-state QKD implementation with laser sources in the presence of state-preparation flaws (SPFs), side-channel information leakage, source correlations affecting the bit, basis, and intensity settings, and basis-dependent detection efficiency mismatch. These results provide a clear route towards rigorous security certification of real-world QKD systems.

\section{Security analysis}\label{sec:SecurityAnalysisNew}

\subsection{Protocol, notation and security condition}

\begin{figure*}
\begin{mybox}{Actual Protocol (AP)}{AP}
\begin{enumerate}
    \item In each round\footnote{In this box we omit the round index $u$ from the quantities whenever it is clear from the context.} $u\in\set{1,\dots,N}$,
    \begin{enumerate}
        \item  With probability $p_i^{\rm A}$, $i\in\set{0,1,\dots,n_A-1}$, 
        where $n_A$ is the number of setting choices, Alice prepares a system $C$ in the state $\sigma_{C}^{i^u_{u-L}}$---which in general may depend not only on the current setting $i \equiv i_u$, but also on the previous settings $i^{u-1}_{u-L}\equiv i_{u-1},i_{u-2},\dots,i_{u-L}$, where $L$ is the correlation length---and sends it to Bob through the quantum channel. The settings $i\in\set{0,1}$ correspond to Alice's $Z$ basis.
        \item Bob selects a measurement basis $\beta\in\set{Z,X_1,\dots,X_{n_M}}$ with probability $\pB{\beta}$, where $n_M$ is the number of test bases, and measures the received system $B$---which may differ from system $C$ due to Eve's intervention---with a positive-operator-valued measure (POVM) $\set{\hat{\Gamma}^{b|\beta}_B}_{b}$, obtaining the measurement outcome $b$. Here we shall consider that Bob's $Z$-basis measurement has possible outcomes $b\in\set{0,1,\loss}$, where $\loss$ is the non-detection outcome.
    \end{enumerate}
    \item Alice and Bob publicly exchange the classical information required for sifting via an authenticated channel. In particular, they reveal the basis information and whether a detection occurred in each round. They classify detected $Z$-basis rounds as \textit{key} rounds and use them to generate their sifted bits. All remaining rounds are classified as \textit{test} rounds, for which the parameters $i$, $\beta$, and $b$ are publicly revealed.
    \item Alice and Bob perform error correction, error verification, and privacy amplification to obtain a secret key.
\end{enumerate}
\end{mybox}
\end{figure*}

For illustration purposes, we shall focus on BB84-type prepare-and-measure (P\&M) QKD protocols, though the analysis is valid for more general schemes, including entanglement-based, interference-based, and decoy-state QKD setups (see Appendix~\ref{app:applications_to_QKD_Protocols}). Specifically, the type of P\&M protocol considered is described in Box~\ref{box:AP}. 
We assume that Alice's transmitted states $\sigma_{C_u}^{i^u_{u-L}}$ are partially characterized and may depend not only on the current setting choice $i_u\in\set{0,1,\dots,\nstates-1}$, where $\nstates$ is the number of settings, but also on the previous choices $i_{u-L}^{u-1}$, with $L$ being the correlation length. We adopt the notation $x_{n}^{m}\equiv x_{m},x_{m-1},\dots,x_{n+1},x_{n}$, where $n\leq m$, and the sequence must be understood to be truncated to rounds within $[1,N]$, with $N$ being the number of transmission rounds. To be precise, we assume that these states satisfy the fidelity constraint
\begin{equation}\label{eq:assumption1thm_main}
    F\left(\sigma_{C_u}^{i^u_{u-L}},\ket{\phi_{i_u}}_{C_u}\right)\geq 1-\epsleak^{i_u},\quad \forall i^{u}_{u-L}
\end{equation}
for certain parameters $\epsleak^{i_u}\geq0$, where $F(\sigma,\ket{\phi}):=\ev{\sigma}{\phi}$ denotes the fidelity and the states $\ket{\phi_{i_u}}_{C_u}$ are some predefined reference states. From~\cref{eq:assumption1thm_main}, it follows (see \cref{lemma:density_op_correlations} in Appendix~\ref{app:correlated_sources}) that, for the purpose of proving the protocol's security, without loss of generality, the transmitted states can be considered to take the alternative pure-state form
\begin{equation}\label{eq:assumption1extra_main}
    \ket*{\psi_{i_{u-L}^u}^{\epsleak}}_{C_u}=\sqrt{1-\epsleak}\ket*{\phi_{i_u}}_{C_u}+\sqrt{\epsleak}\ket*{\phi^{\perp}_{i_{u-L}^{u}}}_{C_u},
\end{equation}
where the states $\ket*{\phi^{\perp}_{i_{u-L}^{u}}}_{C_u}$ are unknown but satisfy $\braket*{\phi^{\perp}_{i_{u-L}^{u}}}{\phi_{i_u}}_{C_u}=0$, and for simplicity we fix $\epsleak^{i_u}=\epsleak$ for all $i_u$. In this latter scenario, the joint state of all transmitted systems $C_u$ and ancilla systems $A_u$ in an entanglement-based view of the protocol can be written as
\begin{equation}\label{eq:global_correlated_state_main}
\ket{\Psi^{\epsleak}}_{A_1^NC_1^N}=\sum_{i_1^N}\bigotimes_{u=1}^{N}\sqrt{p_{i_u}^{\rm A}} \ket{i_u}_{A_u}\ket*{\psi_{i_{u-L}^u}^{\epsleak}}_{C_u},
\end{equation}
where $\piA{i_u}$ is the probability that Alice selects the setting $i_u$ and $\set{\ket{i_u}}_{i_u}$ is an orthonormal basis. For notational simplicity, from now on we drop the round index $u$ when it refers to the current round and this is clear from the context (e.g., $A \equiv A_u$).

In the channel, Eve applies a general completely positive trace-preserving (CPTP) map $\mathcal{C}:\mathcal{L}(\mathcal{H}_{C_1^N})\to \mathcal{L}(\mathcal{H}_{B_1^N})$.
Subsequently, Alice and Bob perform a joint positive-operator-valued measure (POVM) $\set{\hat{S}_{AB}^{i,\beta,b}}_{i,\beta,b}$ on each pair of systems $AB$, where
\begin{equation}\label{eq:SABoperators_main}
\begin{split}
    \hat{S}_{AB}^{i,\beta,b} &= 
    \pB{\beta} \dyad{i}_A\otimes\hat{\Gamma}^{b|\beta}_B,
\end{split}
\end{equation}
$\pB{\beta}$ is the probability that Bob selects the measurement basis $\beta$ and measures the incoming signal with the POVM $\set{\hat{\Gamma}^{b|\beta}_B}_{b\in\set{0,1,\loss}}$, and $\loss$ is the non-detection outcome. We shall consider that Alice and Bob construct their sifted keys from the detected $Z$-basis rounds, namely, the set of rounds with $(i,\beta,b)\in\mathcal{K}$, where $\mathcal{K}:=\set{(i,\beta,b)|\beta=Z\text{ and }i,b\in\set{0,1}}$.

To prove the security of the scheme, we consider a virtual protocol (VP) in which Alice and Bob perform the POVM $\set{\hat{V}_{AB}^{i,b}}_{i,b\in\set{0,1}}\cup\set{\hat{S}_{AB}^{i,\beta,b}}_{(i,\beta,b)\notin \mathcal{K}}$, where
\begin{equation}\label{eq:VABoperators_main}
\begin{split}
    \hat{V}_{AB}^{i,b} &=
    \pB{Z} \dyad{i_X}_A\otimes\hat{\Gamma}^{b|V}_B,
\end{split}
\end{equation}
with $\ket{0_X}=(\ket{0}+\ket{1})/\sqrt{2}$ and $\ket{1_X}=(\ket{0}-\ket{1})/\sqrt{2}$ being the Hadamard states. 
Here, $\hat{\Gamma}^{0|V}_B$ and $\hat{\Gamma}^{1|V}_B$ denote the POVM elements associated with the detection outcomes 0 and 1 in an arbitrary virtual measurement
$\set{\hat{\Gamma}^{b|V}_B}_{b\in\set{0,1,\perp}}$ performed by Bob, subject only to the constraint
$\hat{\Gamma}^{\loss|V}_B=\hat{\Gamma}^{\loss|Z}_B$.
For efficient BB84-type protocols (see Appendix~\ref{app:applications_to_QKD_Protocols}) in which Bob's POVMs satisfy the basis-independent detection efficiency (BIDE) condition---that is, $\hat{\Gamma}^{\loss|Z}_B=\hat{\Gamma}^{\loss|X}_B$, where $X$ is the test basis---this virtual measurement can be simply set to $\hat{\Gamma}^{b|V}_B=\hat{\Gamma}^{b|X}_B$. In fact, for such protocols, our analysis can even be directly applied without any characterization at all of Bob's POVMs $\set{\hat{\Gamma}^{b|\beta}_B}_{b}$ beyond the BIDE assumption (see Appendix~\ref{app:applications_to_QKD_Protocols} for further details). 
For those scenarios in which Bob's measurement setup does not satisfy the BIDE assumption, one can apply the lifting theorem in~\cite{curras2025security}, thus covering also imperfect and partially characterized detection setups whose POVMs may even vary from round to round, within certain allowed deviation limits (see Appendix~\ref{app:finite-key}).

In the VP, we consider that a phase error occurs when Alice's and Bob's bits differ in a virtual key round, i.e., one of the rounds associated with the operators $\set{\hat{V}_{AB}^{i,b}}_{i,b\in\set{0,1}}$. That is, the phase-error operator for any round can be defined as
\begin{equation}\label{eq:phase-error-operator}
\begin{split}
    \hat{E}^{\rm ph}_{AB} &= \hat{V}_{AB}^{0,1}+\hat{V}_{AB}^{1,0}.
\end{split}
\end{equation}
Note, however, that the phase-error definition could be chosen differently.

Now, let $\rv{N}_{\rm key}$ ($\rv{N}_{\rm ph}$) denote the number of sifted-key bits (phase errors) observed in the VP, and let $\rv{e}_{\rm ph}=\rv{N}_{\rm ph}/\rv{N}_{\rm key}$ be the phase-error rate. Moreover, let us assume that the \textit{a priori} probabilistic statement
\begin{equation}\label{eq:probabilistic_statement}
    \Pr[\rv{e}_{\rm ph}> e_{\rm ph}^{\rm U}(\rv{N}_{\rm key},\set{\rv{N}_{\rm test}^{i,\beta,b}}_{i,\beta,b})]\leq \epsilon_{\rm pro},
\end{equation}
holds for certain $\epsilon_{\rm pro}$, where $e_{\rm ph}^{\rm U}$ is a function of the observed random variables (RVs) in an execution of the protocol, and $\rv{N}_{\rm test}^{i,\beta,b}$ is the number of test rounds associated with the tuple $(i,\beta,b)$. Then, there exists a classical postprocessing (i.e., step 3 in Box~\ref{box:AP}) that returns a secret key of length
\begin{equation}\label{eq:secret-key-length}
\begin{split}
    \rv{l}_{\rm key} ={}&  \rv{N}_{\rm key}\left[1-h\left(e_{\rm ph}^{\rm U}(\rv{N}_{\rm key},\set{\rv{N}_{\rm test}^{i,\beta,b}}_{i,\beta,b})\right)\right]-\rv{\lambda}_{\rm EC}
    \\
    & -2\log(1/2\epsilon_{\rm PA})-\log(2/\epsilon_{\rm EV}),
\end{split}
\end{equation}
that guarantees that the protocol is $\epsilon_{\rm tot}$-secure \cite[Theorem~1]{curras-lorenzoRigorousPhaseerrorestimation2026} (see also \cite[Theorem~1]{tupkary2024phase}). Here, $\epsilon_{\rm tot}=\epsilon_{\rm EV}+\epsilon_{\rm sec}$, with $\epsilon_{\rm EV}$ being the failure probability of the error-verification step, $\epsilon_{\rm sec}=2\sqrt{\epsilon_{\rm pro}}+\epsilon_{\rm PA}$ is the secrecy parameter, $\epsilon_{\rm PA}>0$ is freely chosen, and $\rv{\lambda}_{\rm EC}$ is the information revealed during the error correction routine.

The main contribution of this work is to derive a probabilistic bound like~\cref{eq:probabilistic_statement} for practical QKD setups in the presence of multiple device imperfections.

\subsection{Single-round operator inequality}
The phase-error probability of any given single round of the VP can be upper bounded with the following semidefinite program (SDP),
\begin{equation}\label{eq:primal1}
    \begin{split}
        \max_{\rho_{AB}} \quad &  \Tr\set{\hat{E}^{\rm ph}_{AB}\rho_{AB}}\\
        \text{s.t.} \quad & \rho_{AB}\succeq 0 \\
        \quad & \Tr\set{\hat{T}_{A}^{k} \rho_{AB}} = t_k, \quad (k=1,\dots,n_{\rm T})\\
        \quad & \Tr\set{\hat{Q}_{AB}^{l} \rho_{AB}} = q_l, \quad (l=1,\dots,n_{\rm Q}),
    \end{split}
\end{equation}
where $\hat{E}^{\rm ph}_{AB}$ is the phase-error operator defined in~\cref{eq:phase-error-operator}, $\rho_{AB}$ is the (unknown) state of systems $AB$, the $n_{\rm T}$ measurement operators $\hat{T}_{A}^{k}$ serve to characterize---either tomographically or only partially---the marginal state of system $A$, with the quantities $t_k$ being their associated probabilities, and the $n_{\rm Q}$ operators $\hat{Q}_{AB}^{l}$ together with the quantities $q_l$ are defined as
\begin{equation}\label{eq:linear_comb_c}
\begin{split}
    \hat{Q}_{AB}^{l}&:= \sum_{i,\beta,b}c_{i,\beta,b}^l\dyad{i}_A\otimes\hat{\Gamma}^{b|\beta}_B,\\
    q_l&:= \sum_{i,\beta,b}c_{i,\beta,b}^l \piA{i}p_{b|i,\beta}.
\end{split}
\end{equation}
That is, $\hat{Q}_{AB}^{l}$ and $q_l$ are linear combinations---determined by the real coefficients $c_{i,\beta,b}^l$---of the operators and statistics, respectively, that are associated with the test rounds of the protocol. 
These coefficients must be chosen such that the quantities $q_l$ can be related to the statistics publicly revealed in the actual protocol. This means, for instance, that $c_{0,Z,0}^l=c_{0,Z,1}^l=c_{1,Z,0}^l=c_{1,Z,1}^l$, since the bit-wise information of the $Z$-basis rounds is never revealed by Alice and Bob.
The quantities $q_l$ depend on the unknown state $\rho_{AB}$ through the probabilities $p_{b|i,\beta}$, and therefore they are in principle unknown. On the other hand, the probabilities $t_k$ only depend on Alice's marginal state $\rho_A$. This state is perfectly known in certain scenarios---i.e., if one assumes perfect state characterization---but in general it may be unknown. Moreover, note that the quantities $q_l$ and $t_k$ are generally round dependent.

Importantly, the SDP introduced in~\cref{eq:primal1} can be used to derive an operator inequality involving the phase-error operator and all the aforementioned measurement operators~\cite{zhou2022numerical}. For this, before running the protocol, one needs to make a guess on the values of $q_1,\dots,q_{n_{\rm Q}}$ and $t_1,\dots,t_{n_{\rm T}}$. We denote these estimates as the ``guesses'' $q^{\rm gs}_l$ and $t^{\rm gs}_k$, respectively. For example, the $t^{\rm gs}_k$ can be derived from the reference states $\ket{\phi_i}$ introduced in~\cref{eq:assumption1thm_main}---which typically correspond to the states transmitted in an ideal scenario without any imperfection---while $q^{\rm gs}_l$ can be obtained by considering these reference states together with a realistic model of the quantum channel and Bob's measurement device. In doing so, one can construct the following SDP (which is indeed the dual of~\cref{eq:primal1} particularized for these guesses)
\begin{equation}\label{eq:dualSDP2}
\begin{split}
    \min_{\vos{\Lambda}} \quad &  \sum_l\eta_lq^{\rm gs}_l  + \sum_k\lambda_kt^{\rm gs}_k\\
    \text{s.t.} \quad & \hat{E}^{\rm ph}_{AB}  \preceq \sum_l\eta_l\hat{Q}_{AB}^{l}  +\sum_k\lambda_k\hat{T}_{A}^{k}\otimes\identity_B,
\end{split}
\end{equation}
where the vector $\vos{\Lambda}:=(\eta_1,\dots,\eta_{n_{\rm Q}},\lambda_1,\dots,\lambda_{n_{\rm T}})$ contains all the optimization variables.
Note that, if $\vos{\Lambda}^{*}_{\rm gs}:=(\eta_1^*,\dots,\eta_{n_{\rm Q}}^*,\lambda_1^*,\dots,\lambda_{n_{\rm T}}^*)$ is the optimal solution of the previous SDP, then the operator inequality
\begin{equation}\label{eq:operator_ineq_2}
 \hat{E}^{\rm ph}_{AB}  \preceq \sum_l\eta_l^*\hat{Q}_{AB}^{l} +\Tobs{A}\otimes\identity_B
\end{equation}
holds---indeed, it remains true even if $\vos{\Lambda}^{*}_{\rm gs}$ is not necessarily optimal, so long as it is a feasible point---where we conveniently define the observable 
\begin{equation}\label{eq:Tk_obs_definition}
\Tobs{A}:=\sum_k\lambda_k^*\hat{T}_{A}^{k}=\sum_{\omega}\omega\dyad{\omega}_{A},
\end{equation}
with $\set{\omega}_{\omega}$ ($\set{\ket{\omega}}_{\omega}$) being its set of eigenvalues (eigenvectors). Importantly, \cref{eq:operator_ineq_2} implies that the quantity $\Tr\{(\sum_l\eta_l^*\hat{Q}_{AB}^{l} +\Tobs{A}\otimes\identity_B)\rho_{AB}\}$ is a valid upper bound for the phase-error probability $ \Tr\set{\hat{E}^{\rm ph}_{AB}\rho_{AB}}$, which is expected to be tight if the unknown state $\rho_{AB}$ confirms the guesses, i.e., if $\Tr\set{\hat{T}_{A}^{k} \rho_{AB}} \approx t_k^{\rm gs}$ and $\Tr\set{\hat{Q}_{AB}^{l} \rho_{AB}} \approx q_l^{\rm gs}$ for all $k$ and $l$, respectively. In fact, the primal SDP in~\cref{eq:primal1} is introduced only for illustration, since the construction of~\cref{eq:dualSDP2} is direct: one inserts a linear combination of the available measurement operators into the target operator inequality constraint, and takes its expected value under normal operation of the protocol as the objective function.

\subsection{Finite-key security with imperfect devices}\label{subsec:finite_key_imperfect}
\begin{figure*}
\begin{mybox}[myboxcolor=green]{Virtual Estimation Protocol (VEP)}{VEP}
\begin{enumerate}
    \item Alice prepares the entangled state $\ket{\Psi^{\epsleak}}_{A_1^NC_1^N}$ defined in~\cref{eq:global_correlated_state_main}---which depends on some $\epsleak$ such that~\cref{eq:assumption1thm_main} holds with $\epsleak^{i_u}=\epsleak$---and sends all systems $C_1^N$ through the channel to Bob, who receives the systems $B_1^N$. 
    \item Then, for each round $u$, Alice and Bob perform the POVM
    \begin{equation}
        \set{(1-\ptrash)\hat{V}_{AB}^{i,b}}_{i,b\in\set{0,1}}\cup\set{(1-\ptrash)\hat{S}_{AB}^{i,\beta,b}}_{(i,\beta,b)\notin\mathcal{K}}\cup\set{\ptrash\dyad{\omega}_A\otimes\identity_{B}}_{\omega},\nonumber
    \end{equation}
    on their respective round systems $A$ and $B$, where $\ptrash$ is the probability that a round is designated as a ``local $W$-measurement'' (\LC{}) round, the operators $\hat{V}_{AB}^{i,b}$ and $\hat{S}_{AB}^{i,\beta,b}$ are defined in~\cref{eq:SABoperators_main,eq:VABoperators_main}, and the states $\ket{\omega}$ are the eigenvectors of the observable $\Tobs{A}$ given by~\cref{eq:Tk_obs_definition}. Then, they obtain the values of the following random variables:
    \begin{itemize}\label{it:RVs_main}
        \item[\small $\bullet$] $\rv{\chi}_{\rm key}^{u}$: it takes the value 1 if a key round is observed, which is associated with any of the POVM elements $(1-\ptrash)\hat{V}_{AB}^{i,b}$. Otherwise, it takes the value 0.
        \item[\small $\bullet$] $\rv{\chi}_{\rm ph}^{u}$: it takes the value 1 if a phase error occurs, which is associated with the POVM elements $(1-\ptrash)\hat{V}_{AB}^{1,0}$ and $(1-\ptrash)\hat{V}_{AB}^{0,1}$ (see~\cref{eq:phase-error-operator}). Otherwise, it takes the value 0.
        \item[\small $\bullet$] $\rv{\chi}_{{\rm Q},l}^{u}$: it takes the value $\frac{c_{i,\beta,b}^l}{p^{\rm B}_{\beta}}$ if Alice and Bob observe the measurement outcome associated with the POVM element $(1-\ptrash)\hat{S}_{AB}^{i,\beta,b}$. Besides, it takes the value\footnote{Note that the coefficients $c_{i,Z,b}^l$ must be identical for all $i,b\in\set{0,1}$.} $\frac{c_{0,Z,0}^l}{\pB{Z}}$ if $\rv{\chi}_{\rm key}^{u}=1$. Otherwise, it takes the value 0.
        \item[\small $\bullet$]  $\chiT{u}$: it takes the value $\omega$ if Alice and Bob observe the measurement outcome associated with the POVM element $\ptrash\dyad{\omega}_A\otimes\identity_{B}$. Otherwise, it takes the value 0.
    \end{itemize}
    \item Alice and Bob compute $\rv{M}_{\rm key}=\sum_u\rv{\chi}_{\rm key}^{u}$,  $\rv{M}_{\rm ph}=\sum_u\rv{\chi}_{\rm ph}^{u}$ and $\rv{M}_{{\rm Q},l}=\sum_u\rv{\chi}_{{\rm Q},l}^{u}$.
\end{enumerate}
\end{mybox}
\end{figure*}
To prove the security of the protocol in the finite-key regime, we consider the Virtual Estimation Protocol (VEP) \cite{zhou2022numerical} introduced in Box~\ref{box:VEP}. The quantum communication phase of this protocol is similar to that of the VP. 
The main difference is that in the VEP each round is designated as a ``local $W$-measurement'' (\LC{}) round with probability $\alpha$. In these rounds, Alice measures the observable $\Tobs{A}$ on her ancilla system, rather than applying the measurements used in the VP.
At the end of the VEP, Alice observes the quantities $\rv{M}_{\rm key}=\sum_u\rv{\chi}_{\rm key}^{u}$,  $\rv{M}_{\rm ph}=\sum_u\rv{\chi}_{\rm ph}^{u}$, $\rv{M}_{{\rm Q},l}=\sum_u\rv{\chi}_{{\rm Q},l}^{u}$ and $\rv{M}_{\Tsind}=\sum_u\chiT{u}$, which depend on some virtual random variables defined in Box~\ref{box:VEP}.

Now, let $\mathcal{F}_{u}$ be a filtration identifying the RVs $\rv{\chi}_{\rm ph}^{u'}$, $\rv{\chi}_{{\rm Q},l}^{u'}$ (for all $l$) and $\chiT{u'}$, for $u'\in\set{1,\dots,u}$, and let $\rho_{AB}^{\mathcal{F}_{u-1}}$ be the state of Alice's ancilla system $A_u$ and Bob's received system $B_u$ in the VEP after conditioning on $\mathcal{F}_{u-1}$. Then, it is clear that
\begin{align}\label{eq:Expectations1}
    \mathbb{E}\left[\rv{\chi}_{\rm ph}^{u}|\mathcal{F}_{u-1}\right]
    &=(1-\ptrash)\Tr\set{\hat{E}_{AB}^{\rm ph}\rho_{AB}^{\mathcal{F}_{u-1}}},\nonumber
    \\
    \mathbb{E}\left[\rv{\chi}_{{\rm Q},l}^{u}|\mathcal{F}_{u-1}\right]
    &=(1-\ptrash)\Tr\set{\hat{Q}_{AB}^{l}\rho_{AB}^{\mathcal{F}_{u-1}}},\quad\forall l,
    \\
    \mathbb{E}\left[\chiT{u}|\mathcal{F}_{u-1}\right]
    &=\ptrash\Tr\set{\Tobs{A}\otimes\identity_B\rho_{AB}^{\mathcal{F}_{u-1}}}.\nonumber
\end{align}
Note that, whereas in the absence of source correlations $\set{\chiT{u}}_u$ defines a sequence of independent RVs, $\set{\mathbb{E}\left[\chiT{u}|\mathcal{F}_{u-1}\right]}_u$ is generally a sequence of dependent RVs, as these conditional expectations are RVs that depend on Eve's potential coherent attack. From~\cref{eq:operator_ineq_2,eq:Expectations1}, it follows that
\begin{equation}\label{eq:ineq_conditional_expect_Main}
\begin{split}
\sum_{u=1}^N\mathbb{E}\left[\rv{\chi}_{\rm ph}^{u}|\mathcal{F}_{u-1}\right] 
\leq
&
\sum_l\eta_l^*\sum_{u=1}^N\mathbb{E}\left[\rv{\chi}_{{\rm Q},l}^{u}|\mathcal{F}_{u-1}\right]
\\
&+  \frac{1-\ptrash}{\ptrash}\sum_{u=1}^N\mathbb{E}\left[\chiT{u}|\mathcal{F}_{u-1}\right].
\end{split}
\end{equation}
Next, one can invoke Azuma's~\cite{azuma} or Kato's~\cite{kato} inequalities to link the previous sums of conditional expectations with the sums of the actual RVs, namely $\rv{M}_{\rm ph}$, $\rv{M}_{{\rm Q},l}$ and $\rv{M}_{\Tsind}$. In particular, by employing Kato's inequality, we prove the following result in Appendix~\ref{app:finite-key} (see~\cref{eq:Kfinal}): let us assume that, prior to the protocol execution, one can guarantee that the probabilistic statement
\begin{equation}\label{eq:priori_bound_MT}
    \Pr[\rv{M}_{\Tsind}> M_{\Tsind}^{\rm U}]\leq \epsproB
\end{equation}
holds for certain $M_{\Tsind}^{\rm U}$ and $\epsproB$. Then, the inequality
\begin{widetext}
\begin{equation}\label{eq:Kfinal_main}
\begin{split}
    \rv{M}_{\rm ph} 
    \leq 
    \frac{N}{\sqrt{N}-2\tilde{a}_{\rm U}}\left[\frac{1}{\sqrt{N}}
    \left( \sum_l\eta_l^*K_{\epsilon_l,N,x_{\min}^l,x_{\max}^l}^{\text{sign}(\eta_l^*)}(\rv{M}_{{\rm Q},l})
    +
    \frac{1-\ptrash}{\ptrash} K_{\epsproK,N,\omega_{\min},\omega_{\max}}^{1}(M_{\Tsind}^{\rm U})
    \right)
    +\tilde{b}_{\rm U}-\tilde{a}_{\rm U}\right]
\end{split}
\end{equation}
\end{widetext}
holds except with probability $3\epsproK+\epsproB$, where the quantities $\tilde{a}_{\rm U}$ and $\tilde{b}_{\rm U}$ are computed prior to the protocol execution according to~\cref{eq:optimal_ab_tilde}; $x_{\min}^l$ and $x_{\max}^l$ ($\omega_{\min}$ and $\omega_{\max}$) denote the minimum and maximum possible outcomes of the RVs $\rv{\chi}_{{\rm Q},l}^u$ ($\rv{\chi}_{{\Tsind}}^u$), respectively; the quantities $\epsilon_l$ represent some arbitrary failure probabilities satisfying $\sum_{l} \epsilon_l=\epsproK$; the functions $K_{\epsilon,N,x_{\min},x_{\max}}^{\pm 1}$ are defined in~\cref{eq:Kfunctions_unnormalized}; and $\text{sign}(\cdot)$ denotes the sign function. We refer the reader to Appendix~\ref{app:finite-key} for further details.

The quantities $\rv{M}_{{\rm Q},l}$ in~\cref{eq:Kfinal_main} can be readily related to observable statistics in the actual protocol by randomly sampling them with probability $1-\alpha$ (see Appendix~\ref{app:finite-key}). Therefore, to use~\cref{eq:Kfinal_main} one only needs to obtain a bound $M_{\Tsind}^{\rm U}$ satisfying~\cref{eq:priori_bound_MT}. 
This bound is derived in Appendix~\ref{app:correlated_sources}, where we prove the following result:

\noindent {\bf Theorem \ref{thm:correlations}. }(Probabilistic upper bound on $\rv{M}_{\Tsind}$) {\itshape 
Consider the virtual estimation protocol (VEP) introduced in Box~\ref{box:VEP}, and let $\ket{\Psi^{\epsleak}}_{A_1^NC_1^N}$, $N$, $\chiT{u}$, $\rv{M}_{\Tsind}$, $\alpha$, and $\Tobs{A}$ be as defined there, where $\ket{\Psi^{\epsleak}}_{A_1^NC_1^N}$ is given in~\cref{eq:global_correlated_state_main} and depends on a certain set of reference states $\set{\ket{\phi_i}}_{i\in\set{0,1,\dots,\nstates-1}}$, setting probabilities $\set{\piA{i}}_{i\in\set{0,1,\dots,\nstates-1}}$, correlation length $L$, and on the parameter $\epsleak\geq0$.
Let the quantities $E_{\chi_{\Tsind}}^{\rm L}$, $E_{\chi_{\Tsind}}^{\rm U}$ and $E_{\chi_{\Tsind}^2}^{\rm U}$ satisfy
\begin{equation}\label{eq:optimizations_theorem_main}
\begin{split}
E_{\chi_{\Tsind}}^{\rm L}
&\leq \ptrash\min_{G\in\mathcal{S}_{\rm const}}\left[\Tr\set{\Tobs{A} \rho_{A}^{\varepsilon'}}\right],
\\
E_{\chi_{\Tsind}}^{\rm U}
&\geq\ptrash\max_{G\in\mathcal{S}_{\rm const}}\left[\Tr\set{\Tobs{A} \rho_{A}^{\varepsilon'}}\right],
\\
E_{\chi_{\Tsind}^2}^{\rm U}
&\geq\ptrash\max_{G\in\mathcal{S}_{\rm const}}\left[\Tr\set{(\Tobs{A})^2 \rho_{A}^{\varepsilon'}}\right],
\end{split}
\end{equation}
where
\begin{equation}
G:=
\begin{bmatrix}
G_{\phi} & G_{\rm c}     \\
G_{\rm c}^{\dagger}   & G_{\perp}   \\
\end{bmatrix}
\end{equation}
is a $2\nstates \times 2\nstates$ Hermitian block matrix built from the $\nstates\times\nstates$ matrices $G_{\phi}$, $G_{\rm c}$, and $G_{\perp}$; $\rho_{A}^{\varepsilon'}$ is a density matrix whose entries depend linearly on those of $G$, with $[\rho_{A}^{\varepsilon'}]_{ij} := \sqrt{\piA{i}\piA{j}}\big[(1-\varepsilon')[G_{\phi}]_{ji} +\varepsilon'[G_{\perp}]_{ji} + \sqrt{(1-\varepsilon')\varepsilon'} ([G_{\rm c}]_{ij}^{*} + [G_{\rm c}]_{ji})\big]$ and $\epsleak':=1-(1-\epsleak)^{L+1}$; and $\mathcal{S}_{\rm const}$ is a convex set defined by the constraints
\begin{equation}\label{eq:SDP_gram_main}
    \begin{split}
        \quad & G \succeq 0,\\
        \quad & [G_{\phi}]_{ij}=\braket{\phi_{i}}{\phi_{j}}, \quad i,j\in\set{0,1,\dots,\nstates-1},\\
        \quad & [G_{\rm c}]_{ii}=0, \quad i\in\set{0,1,\dots,\nstates-1},\\
        \quad & [G_{\perp}]_{ii}=1,\quad i\in\set{0,1,\dots,\nstates-1}.
    \end{split}
\end{equation}
Then, for any $\epsproB>0$, we have that
\begin{align}
\Pr[\rv{M}_{\Tsind}\geq M_{\Tsind}^{\rm U}] &\leq \epsproB
\end{align}
holds for
\begin{align}
M_{\Tsind}^{\rm U}:=& N E_{\chi_{\Tsind}}^{\rm U}+(L+1)\bar{N}_L\tilde{\Delta}_{\rm B},
\end{align}
where

\begin{align}
\tilde{\Delta}_{\rm B} := &\frac{c}{3\bar{N}_L} \ln\frac{L+1}{\epsproB} \nonumber
\\
&+ \sqrt{ \left(\frac{c}{3\bar{N}_L} \ln\frac{L+1}{\epsproB}\right)^2  + 2 \frac{V_{\chi_{\Tsind}}^{\rm U}}{\bar{N}_L} \ln\frac{L+1}{\epsproB}}, \nonumber
\\
c := & \max\set{\omega_{\rm max}-E_{\chi_{\Tsind}}^{\rm L},E_{\chi_{\Tsind}}^{\rm U}-\omega_{\rm min}},\nonumber
\\
V_{\chi_{\Tsind}}^{\rm U}
:=&
E_{\chi_{\Tsind}^2}^{\rm U} - \left[ \max(0, E_{\chi_{\Tsind}}^{\rm L}) + \min(0, E_{\chi_{\Tsind}}^{\rm U}) \right]^2,\nonumber
\end{align}
$\bar{N}_{L}:=\lceil N/(L+1) \rceil$, and $\omega_{\min}$ ($\omega_{\max}$) denotes the minimum (maximum) possible value of the RVs $\chiT{u}$.
}

\cref{thm:correlations} provides the bound required to use~\cref{eq:Kfinal_main}. 
Importantly, the optimization problems in~\cref{eq:optimizations_theorem_main} are SDPs, meaning that rigorous bounds on the relevant quantities can be obtained by numerically solving their duals, as any dual-feasible solution constitutes a valid bound.
Leaving the technicalities aside, the idea behind~\cref{thm:correlations} is to partition the protocol rounds into $L+1$ groups according to the value of $u \bmod (L+1)$. Conditioned on Alice's setting choices in all the remaining rounds, the outcomes of the fictitious local measurements within a group are independent, and their worst-case mean and variance over all the states compatible with~\cref{eq:assumption1extra_main} can be bounded with an SDP. The bound $M_{\Tsind}^{\rm U}$ then follows by applying Bernstein's inequality to each group, together with a union bound. An analogous result for the decoy-state setting, which additionally accounts for setting-dependent photon-number statistics, is provided in~\cref{thm:correlations_decoy} (see Appendix~\ref{app:correlated_sources}).
The tightness of these optimizations is affected by the correlation length $L$ via the effective fidelity-type parameter $\varepsilon'=1-(1-\epsleak)^{L+1}\approx (L+1)\varepsilon$, which grows approximately linearly with $L$. Remarkably, our results can be adapted to other correlation models, including models in which the correlations have an unbounded length~\cite{pereira2024quantum,agulleiro2025modeling,curras-lorenzoRigorousPhaseerrorestimation2026}.

In Appendix~\ref{app:finite-key} we extend \cref{eq:Kfinal_main} to cover also the case of imperfect receivers suffering a detection efficiency mismatch. Moreover, in that Appendix we derive a probabilistic upper bound $N_{\rm ph}^{\rm U}\left(\rv{N}_{\rm key},\rv{M}_{\rm key},\set{\rv{M}_{{\rm Q},l}}_l\right)$ (see~\cref{eq:final_bound_Nph}) on the number of phase errors $\rv{N}_{\rm ph}$ in the VP via random sampling arguments. In doing so, we obtain a bound in the form of~\cref{eq:probabilistic_statement}, which allows one to compute the secret-key length (via~\cref{eq:secret-key-length}) of general QKD schemes with flawed, leaky, correlated, and partially characterized transmitted states at Alice, and basis-dependent detection efficiency at Bob.

Importantly, we note that this work builds on the ideas introduced in Ref.~\cite{zhou2022numerical}, from which we adopt the SDP-based derivation of an operator inequality and the idea of a virtual estimation protocol with fictitious local measurements (see Box~\ref{box:VEP}), and we significantly extend them to accommodate all the aforementioned device imperfections, as well as to cover the decoy-state scenario. To be precise, our analysis is, to our knowledge, (i) the first numerical finite-key security proof against coherent attacks that is applicable to partially characterized and correlated transmitted states, and (ii) the first finite-key security framework for decoy-state QKD that simultaneously incorporates intensity and bit/basis correlations. Moreover, unlike Ref.~\cite{zhou2022numerical}, (iii) for BB84-type protocols it does not require a specific description of Bob's measurement operators, as it suffices to assume that his detection efficiency is basis independent; (iv) this latter assumption can in turn be relaxed by means of the lifting theorem in Ref.~\cite{curras2025security}, thus covering detection-efficiency mismatch; and (v) the actual protocol does not require Alice and Bob to discard the data of any fictitious rounds.

\section{Results}\label{sec:Results}

\cref{fig:SKR_BB84_main} shows the secret-key rate of a single-photon BB84 scheme obtained with our numerical approach in the presence of multiple device imperfections. Specifically, we model SPFs through the particular choice of reference states $\set{\ket{\phi_i}_C}_i$ introduced in~\cref{eq:ideal_states_BB84}, in which the encoding angles---representing, for instance, the signal polarizations---deviate from their ideal values $\theta_i$ according to $\theta_i\to(1+\Delta\theta/\pi)\theta_i$, and we fix $\Delta\theta=0.063$ \cite{xu2015experimental,honjo2004differential} (see Appendix~\ref{app:applications_to_QKD_Protocols}). Moreover, we consider that the transmitted states are partially characterized and correlated, which is modeled with the assumption given in~\cref{eq:assumption1extra_main}. For illustration purposes, we consider two values for the parameter $\epsleak\in\set{10^{-5},0}$, and three values for the maximum correlation length $L\in\set{0,1,2}$. As for the receiver, we assume a pessimistic dark-count probability $p_d=10^{-6}$, a detection efficiency $\eta_{\rm det}=0.73$ \cite{pittaluga2021600}, and we fix the detector's tolerances to $\Delta_{\eta}=\Delta_{\rm dc}\in\set{0,0.05}$ (see Appendix~\ref{app:finite-key}). In addition, we fix the number of transmitted signals to $N=10^{10}$, the security parameters to $\epsilon_{\rm sec}=\epsilon_{\rm EV}=10^{-10}$, and we set for simplicity $\epsilon_{\rm PA}=\epsilon_{\rm sec}/2$, $\epsilon_{\rm pro}=(\epsilon_{\rm sec}/4)^2$, $\epsproK=\epsproB=\epsproS=\epsilon_{\rm pro}/6$ and $\epsilon_{\rm dep-1}^2=\epsilon_{\rm dep-2}^2=\epsilon_{\rm pro}/12$ (see Appendix~\ref{app:finite-key} for the definitions of these security parameters). As for the quantum channel, we consider a lossy but noiseless fiber-based channel model with attenuation coefficient $\alpha_{\rm dB}=0.2$ dB/km. For each distance, we optimize over the probability to select the $Z$-basis (which for simplicity we consider equal for Alice and Bob), and the parameters $\ptrash$ and $\gamma_{\Tsind}$ (see Appendix~\ref{app:applications_to_QKD_Protocols}).

\begin{figure}
    \centering
    \includegraphics[width=0.99\columnwidth]{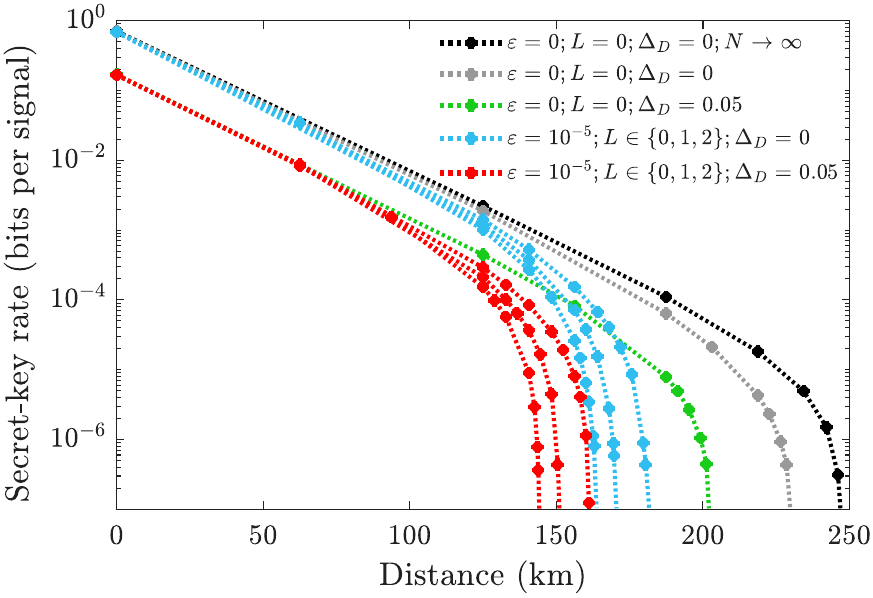}
    \caption{Optimized secret-key rate of a single-photon BB84 scheme in the presence of multiple device imperfections. For all lines, we consider SPFs with $\Delta\theta=0.063$ \cite{xu2015experimental,honjo2004differential}, and we set the number of protocol rounds to $N=10^{10}$ (except for the black line, which assumes the ideal asymptotic case in the absence of any imperfection). Besides, we consider that the quantum states emitted by Alice are partially characterized and correlated, which is modeled via the parameter $\epsleak$ and the correlation length $L$. For the receiver, we fix the detector tolerances to the same value, i.e.,  $\Delta_{\eta}=\Delta_{\rm dc}=\Delta_{D}$.\label{fig:SKR_BB84_main}}
\end{figure}

The results shown in \cref{fig:SKR_BB84_main} demonstrate that our numerical method is robust to multiple realistic device imperfections---even when all of them are taken into account simultaneously---while guaranteeing security against coherent attacks. In particular, the red curves jointly incorporate the most relevant imperfections present in a BB84 setup, yet achieve a maximum transmission distance of $\sim 145$~km, exceeding half of the maximum distance attainable under ideal conditions. For further details about the application of the security proof to this particular scenario, we refer the reader to Appendix~\ref{subsec:bb84_example}. There we also compare our numerical method with the state-of-the-art analytical security proof~\cite{curras2023security2}, which accommodates the same device imperfections considered in this scenario. We find that both approaches yield comparable performance, with our numerical method slightly outperforming the analytical one across all regimes considered. This is not unexpected, as BB84 is a well-studied protocol for which the analysis in~\cite{curras2023security2} is known to provide tight bounds. One of the main advantages of the numerical approach lies in its versatility, as it delivers tight security bounds for protocols where no analytical proof is available, or where the existing analytical proofs are suboptimal. A representative example is the coherent-light-based measurement-device-independent (MDI) protocol introduced in~\cite{navarrete2021practical}, whose non-qubit encoding differs markedly from that of BB84 and for which our method achieves a significant performance improvement with respect to the available analytical proofs~\cite{navarrete2021practical,curras2023security2}, as shown in Appendix~\ref{subseq:MDI}.

\begin{figure}
    \centering
    \includegraphics[width=0.99\columnwidth]{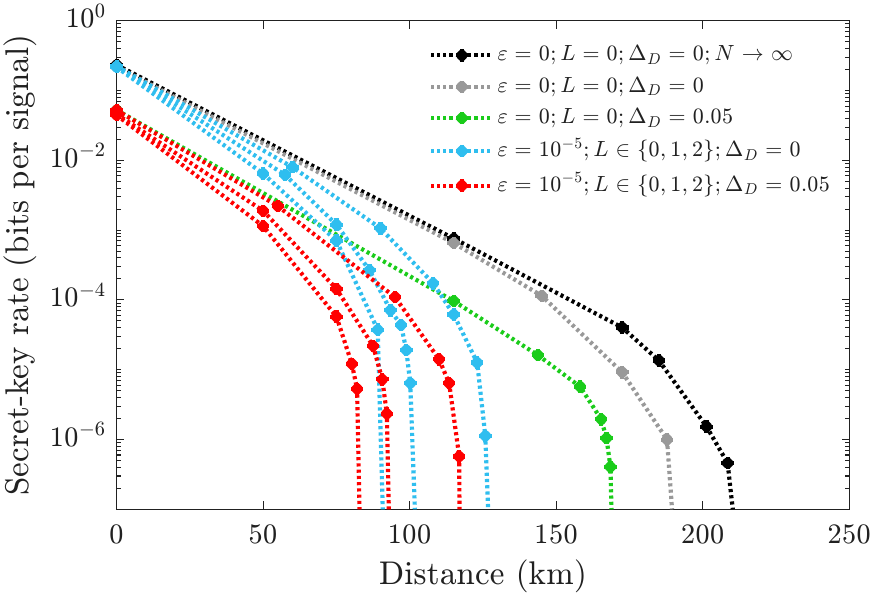}
    \caption{Optimized secret-key rate of a decoy-state BB84 QKD scheme with a laser source and three intensity settings in the presence of multiple device imperfections. For all lines, we consider SPFs with $\Delta\theta=0.063$ \cite{xu2015experimental,honjo2004differential}, and we set the number of protocol rounds to $N=10^{11}$ (except for the black line, which considers the ideal asymptotic case in the absence of any imperfection). Besides, we consider that the quantum states emitted by Alice are partially characterized and correlated, which is modeled via the parameter $\epsleak$ and the correlation length $L$. For the receiver, we fix the detector tolerances to the same value, i.e.,  $\Delta_{\eta}=\Delta_{\rm dc}=\Delta_{D}$.\label{fig:SKR_decoy_main}}
\end{figure}

This versatility is especially relevant in the context of decoy-state schemes, where, to date, no security proof---analytical or numerical---can accommodate all device imperfections considered in this work. In particular, \cref{fig:SKR_decoy_main} shows the secret-key rate of a decoy-state QKD scheme with coherent light generated by a laser in the presence of SPFs, information leakage, bit/basis and intensity correlations, and detection efficiency mismatch.

We model SPFs as in the BB84 setting by considering that Alice's encoding angles deviate from their ideal values $\theta_i$ according to $\theta_i\to(1+\Delta\theta/\pi)\theta_i$ (see~\cref{eq:ideal_states_DS}). Besides, we assume that Alice's emitted states are only partially characterized and correlated, as described by the assumption in~\cref{eq:assumption1thm_decoy}. For the channel model and device imperfections, we use the same parameter values as in the previous example, namely $\Delta\theta=0.063$, $\epsleak\in\set{10^{-5},0}$, $L\in\set{0,1,2}$, $\Delta_{\eta}=\Delta_{\rm dc}=\Delta_{D}\in\set{0,0.05}$, $p_d=10^{-6}$, and $\eta_{\rm det}=0.73$, and we fix the number of transmitted signals to $N=10^{11}$. 
To characterize intensity correlations, we use the model provided in Appendix~\ref{app:model_correlations}, with a nearest-neighbor correlation factor $\epsilon_{\rm ic}^1=0.03$, and a decay parameter $\zeta=6$---which aligns with the results in~\cite{agulleiro2025modeling}. In addition, we fix $\epsilon_{\rm sec}=\epsilon_{\rm EV}=10^{-10}$, and we set for simplicity $\epsilon_{\rm PA}=\epsilon_{\rm sec}/2$, $\epsilon_{\rm pro}=(\epsilon_{\rm sec}/4)^2$, $\epsproK=\epsproB=\epsilon_{\rm dep-1}^2=\epsilon_{\rm dep-2}^2=\epsilon_{\rm pro}/10$ (see Appendix~\ref{app:general_protocols_decoy_state}), and use $n_{\rm cut}=3$.

For each distance, we optimize the parameter $\alpha$, the basis selection probability $p_Z$---for simplicity, we select the same probability for Alice and Bob---the parameter $\gamma_{\Tsind}$ (see Appendix~\ref{app:general_protocols_decoy_state} for further details), the average values of the two higher intensities, and the probabilities of selecting each intensity setting. The weakest intensity is fixed to $10^{-5}$ to account for the finite extinction ratio of practical intensity modulators.  

The simulations plotted in \cref{fig:SKR_decoy_main} demonstrate that the numerical security proof is robust against the aforementioned device imperfections in the decoy-state scenario, reaching up to $\sim 83$ km even in the presence of second-order intensity and bit/basis correlations with $\epsleak=10^{-5}$, and a basis-dependent detection efficiency mismatch with $\Delta_{D}=0.05$. For concreteness, these simulations assume perfect continuous phase randomization, but the analysis does not rely on this assumption: it also covers sources with discrete phase randomization~\cite{caoDiscretephaserandomizedCoherent2015} and with imperfect continuous phase randomization~\cite{naharImperfectPhase2023} (see Appendix~\ref{app:imperfect_PR}).

\section{Conclusions}
We have introduced a numerical finite-key security framework for QKD protocols against coherent attacks in the presence of realistic device imperfections. 
The framework accommodates general state-preparation flaws and side-channel information leakage, and remains applicable even when the emitted states are only partially characterized, requiring only lower bounds on their fidelity to suitable reference states.
Crucially, it also incorporates pulse correlations---one of the principal outstanding challenges for previous numerical security proofs---as well as detector imperfections such as basis-dependent detection efficiency. The approach is broadly applicable to prepare-and-measure, measurement-device-independent, and entanglement-based QKD protocols, while outperforming state-of-the-art analytical proofs in regimes where a comparison is possible.

Remarkably, the framework also applies to decoy-state protocols in the presence of general imperfections, including cross correlations between bit, basis, and intensity settings. In particular, it captures scenarios in which both the emitted Fock states and their emission probabilities depend not only on the current settings, but also on previous intensity and bit/basis choices---correlations that naturally arise in high-speed QKD implementations. Our results bridge an important gap between theoretical security analyses and the behaviour of practical QKD systems, constituting a significant step towards the deployment and certification of secure quantum communication technologies.

\section{Acknowledgments}
This work was supported by the Galician Regional Government (consolidation of Research Units: AtlantTIC); 
the Spanish Ministry of Economy and Competitiveness (MINECO); 
the Fondo Europeo de Desarrollo Regional (FEDER) through the grant No. PID2024-162270OB-I00; 
the ``Hub Nacional de Excelencia en Comunicaciones Cuánticas'' funded by the Ministerio para la Transformación Digital y de la Función Pública and the European Union NextGenerationEU; 
the European Union's Horizon Europe Framework Programme under the Marie Sklodowska-Curie Grant No. 101072637 (Project QSI); 
the project ``Quantum Secure Networks Partnership'' (QSNP, grant agreement No. 101114043); 
the European Union under the Project IberianQCI (grant 101249593), and the Programa de Cooperación Interreg VI-A España--Portugal (POCTEP) 2021--2027 through the project QUANTUM\_IBER\_IA. 
A.N. acknowledges financial support from the Xunta de Galicia (Consellería de Educación, Ciencia, Universidades e Formación Profesional) through a Xunta de Galicia Postdoctoral Fellowship (No. ED481B-2025/113). 
G.C.-L. acknowledges funding from the European Union's Horizon Europe research and innovation programme under the Marie Skłodowska-Curie Postdoctoral Fellowship grant agreement No.\ 101149523.

\clearpage
\onecolumngrid

\appendix
\section{Application of the method to specific QKD protocols}\label{app:applications_to_QKD_Protocols}
The security analysis introduced in~\cref{sec:SecurityAnalysisNew} of the main text is versatile and can be readily applied to many different QKD schemes based on single photons or practical signals such as weak coherent pulses that suffer from different types of device imperfections. In this Appendix, we illustrate this with various examples.

\subsection{Single-photon BB84 protocol}\label{subsec:bb84_example}
Here we assume that Alice has a single-photon source and her transmitted quantum states satisfy~\cref{eq:assumption1thm_main} for a certain parameter $\epsleak \ge \epsleak^{i_u}\quad \forall\, i_u$, a given correlation length $L$, and a set of reference states. Specifically, we consider the following reference states
\begin{equation}\label{eq:ideal_states_BB84}
\begin{split}
    \ket{\phi_i}_{C}&=\cos(\kappa_i\theta_i)\ket{0}_{C}+\sin(\kappa_i\theta_i)\ket{1}_{C},
\end{split}
\end{equation}
where $i\in\set{0,1,2,3}$, the phase parameters are given by $(\theta_0,\theta_1,\theta_2,\theta_3)=(0,\frac{\pi}{2},\frac{\pi}{4},\frac{3\pi}{4})$, and the coefficients $\kappa_i$ quantify the characterized deviations inside the qubit space spanned by $\set{\ket{0}_C,\ket{1}_C}$ due to a flawed state preparation. In particular, for illustration purposes we consider here that $\kappa_i=(1+\Delta\theta/\pi)$ for a certain deviation term $\Delta\theta$, which is a model that has been used in previous works~\cite{pereira2020quantum,curras2025numerical,curras2023security2}. That is, the description of the states given in~\cref{eq:ideal_states_BB84} already accounts for SPFs, as this type of imperfection is typically easy to characterize experimentally by the users. We emphasize, however, that this selection of reference states is just a particular example, as these states can be chosen freely. As usual, the $Z$- and $X$-basis states are those associated with $i\in\set{0,1}$ and $i\in\set{2,3}$, respectively.

As for Bob, in the BB84 protocol he uses a single measurement basis for testing, namely the $X$ basis. Besides, let us assume only for the moment that he receives either a single-photon state or a vacuum state---this assumption is in fact not necessary, as we show below, but simplifies the discussion. In the absence of imperfections at Bob's side, a convenient matrix representation of the receiver's measurement POVM elements in the basis $\set{\ket{\perp},\ket{0},\ket{1}}$ is simply given by
\begin{equation}
    \hat{\Gamma}_B^{\perp|\beta}=\dyad{\perp_{\beta}} \equiv{\dyad{\perp}}=
\begin{bmatrix}
1 & 0 & 0 \\
0 & 0 & 0 \\
0 & 0 & 0
\end{bmatrix},
\quad
(\beta\in\set{Z,X}),
\end{equation}
\text{and}
\begin{equation}
\hat{\Gamma}_B^{b|\beta}=\dyad{b_{\beta}}=
\begin{bmatrix}
0 & 0 \\
0 & \hat{\Gamma}_B^{b|\beta,\text{det}} 
\end{bmatrix}
\quad
(b\in\set{0,1}, \beta\in\set{Z,X}),
\end{equation}
with 
\begin{equation}\label{eq:four-matrices-POVM}
\begin{array}{cc}
\hat{\Gamma}_B^{0|Z,\text{det}}=\dyad{0_Z}=
\begin{bmatrix}
1 & 0 \\
0 & 0 
\end{bmatrix},
&\quad
\hat{\Gamma}_B^{1|Z,\text{det}}=\dyad{1_Z}=
\begin{bmatrix}
0 & 0 \\
0 & 1 
\end{bmatrix},
\\[1em]
\hat{\Gamma}_B^{0|X,\text{det}}=\dyad{0_X}=
\frac{1}{2}
\begin{bmatrix}
1 & 1 \\
1 & 1 
\end{bmatrix},
&\quad
\hat{\Gamma}_B^{1|X,\text{det}}=\dyad{1_X}=
\frac{1}{2}
\begin{bmatrix}
1 & -1 \\
-1 & 1 
\end{bmatrix}.
\end{array}
\end{equation}
The phase-error operator in this simple scenario is given by
\begin{equation}
\hat{E}^{\rm ph}_{AB}=\pB{Z}\dyad{0_X}_A\otimes \hat{\Gamma}_B^{1|X} + \pB{Z}\dyad{1_X}_A\otimes \hat{\Gamma}_B^{0|X}.
\end{equation}
Note that here we used the BIDE condition, i.e. $\hat{\Gamma}^{\loss|Z}_B=\hat{\Gamma}^{\loss|X}_B$, which allows us to choose $\hat{\Gamma}_B^{b|V}=\hat{\Gamma}_B^{b|X}$.

\begin{remark}\label{rem:remark_detectors}
Importantly, even though a specific representation of Bob's $X$-basis POVM elements is required to run the SDP of~\cref{eq:dualSDP2} and obtain an operator inequality in the form of~\cref{eq:operator_ineq_2}, a precise characterization of Bob's POVM is not required in practice. To see this, we first note that the statistics associated with Bob's $Z$-basis measurement are irrelevant to estimate the phase-error rate, and so only Bob's $X$-basis measurement operators are employed in~\cref{eq:dualSDP2}. Now, if one plugs the ideal $X$-basis operators introduced in~\cref{eq:four-matrices-POVM} to solve~\cref{eq:dualSDP2}, the resulting operator inequality---i.e., the one given in~\cref{eq:operator_ineq_2}---can be rewritten as
\begin{equation}\label{eq:operators-separated}
    \hat{O}_{A}^{(0_X)}\otimes \dyad{0_{X}}_B + \hat{O}_{A}^{(1_X)}\otimes \dyad{1_{X}}_B  +\hat{O}_{A}^{(\perp_X)}\otimes \dyad{\perp}_B \succeq 0,
\end{equation}
where the operator $\hat{O}_{A}^{(b_X)}$ is constructed from Alice's operators that are associated with Bob's POVM element $\dyad{b_X}_B$ with $b\in\set{0,1,\perp}$. Due to the orthogonality of the states $\dyad{b_X}_B$, \cref{eq:operators-separated} can be separated into three independent operator inequalities of the form $\hat{O}_{A}^{(b_X)} \succeq 0$, $b\in\set{0,1,\perp}$, that involve only Alice's operators. This means that, for any uncharacterized operators $\tilde{\Gamma}_B^{b|X}\succeq 0$, the inequalities $\hat{O}_{A}^{(b_X)}\otimes\tilde{\Gamma}_B^{b|X} \succeq 0$ must also hold. And since the sum of positive semi-definite operators is also positive semi-definite, we have that~\cref{eq:operators-separated}---and by extension~\cref{eq:operator_ineq_2}---also holds after making the substitutions $\dyad{b_X}_B\to\tilde{\Gamma}_B^{b|X}$. 

The previous argumentation implicitly assumes that the BIDE condition is satisfied. This guarantees that $\tilde{\Gamma}^{\loss|Z}_B=\tilde{\Gamma}^{\loss|X}_B$, which allows us to select $\tilde{\Gamma}^{b|V}_B=\tilde{\Gamma}^{b|X}_B$. Nonetheless, as explained in Appendix~\ref{app:finite-key}, the techniques introduced in~\cite{curras2025security} can be incorporated into our security proof for it to apply to scenarios with a basis-dependent detection efficiency mismatch.
\end{remark}

Now, regarding the coefficients of the linear combinations introduced in~\cref{eq:linear_comb_c}, we note that if $\Delta\theta=0$---i.e., when there are no SPFs---the $X$-basis bit error probability is typically sufficient to obtain a good estimation of the phase-error probability. Specifically, in this particular scenario one could consider a single value of the index $l$, namely $l=1$, and set
\begin{equation}
c^{1}_{i,\beta,b} = 
\begin{cases}
1 &\text{ if } (i,b,\beta)\in\set{(2,1,X),(3,0,X)},\\
0 &\text{ otherwise. }
\end{cases}
\end{equation}
In general, however, additional statistics are required to tightly estimate the phase-error probability~\cite{tamaki2014loss}. In particular, one can directly use all the statistics associated with detected events in which Bob selects the $X$ basis---which in this protocol correspond exactly to eight possible events. For instance, one could simply consider some binary coefficients $c^{l}_{i,\beta,b}$ with $l=1,\dots,8$, defined as
\begin{equation}
c^{1+i'+4b'}_{i,\beta,b} = 
\begin{cases}
1 &\text{ if } i=i',b=b', \beta=X, \\
0 &\text{ otherwise, }
\end{cases}
\end{equation}
where $i'\in\set{0,1,2,3}$ and $b'\in\set{0,1}$. Given these coefficients, one can determine the guesses $q^{\rm gs}_l$ and the operators $\hat{Q}^l_{AB}$. In any case, note that this is just a simple example that employs all the available detection statistics.

As for the fictitious characterization measurement, one can use the following set of operators~\cite{zhou2022numerical}
\begin{equation}\label{eq:Tk_BB84}
\hat{T}_A^{i,j} = 
\begin{cases}
\frac{1}{2}\ketbra{i}{j} + \frac{1}{2}\ketbra{j}{i} &\text{ if } i\geq j, \\
\frac{1}{2\text{i}}\ketbra{i}{j} - \frac{1}{2\text{i}}\ketbra{j}{i} &\text{ if } i < j,
\end{cases}
\end{equation}
where $i,j\in\set{0,1,2,3}$. Given~\cref{eq:ideal_states_BB84,eq:Tk_BB84}, one can obtain good guesses by computing $t_{i,j}^{\rm gs}=\Tr{\hat{T}^{i,j}_A\rho_A^{\epsleak'}}$ for the marginal state
\begin{equation}\label{eq:rhoA_eps_BB84}
    \rho_{A}^{\varepsilon'} = \sum_{i,j}\sqrt{p_i^Ap_j^A}\left[(1-\varepsilon')\braket{\phi_j}{\phi_i} + \varepsilon'\braket*{\phi_j^{\perp}}{\phi_i^{\perp}} + \sqrt{(1-\varepsilon')\varepsilon'} (\braket*{\phi_j}{\phi_i^{\perp}} + \braket*{\phi_j^{\perp}}{\phi_i})\right] \ketbra{i}{j}_A,
\end{equation}
where $\varepsilon'=1-(1-\epsleak)^{L+1}$.
Importantly, since $\rho_{A}^{\epsleak'}$ depends on the unknown states $\ket*{\phi^{\perp}_{i}}$, here we shall consider a rather pessimistic---although not necessarily the true worst-case---scenario in which $\braket*{\phi^{\perp}_{i}}{\phi^{\perp}_{j}}=\braket*{\phi^{\perp}_{i}}{\phi_{j}}=0$ for $i\neq j$. We remark that this is not an assumption required in the security proof, but only a supposition to obtain some suitable unconfirmed guesses $t_{i,j}^{\rm gs}$. Naturally, each tuple $(i,j)$ can be associated with a different value of the index $k\in\set{1,\dots,16}$, such that we can rewrite $\hat{T}^{i,j}_A$ as $\hat{T}^{k}_A$ and $t_{i,j}^{\rm gs}$ as $t_{k}^{\rm gs}$. Once all the operators and guesses are specified, one can run the SDP given in~\cref{eq:dualSDP2} and follow the subsequent steps described in the security analysis to obtain the secret-key rate of the protocol. 

We note that~\cref{eq:dualSDP2} can be re-normalized to improve numerical stability and to make its input statistics independent of Alice's settings probabilities $\piA{i}$ (see Appendix~\ref{app:normalizationSDP}).
Moreover, in certain scenarios, \cref{eq:dualSDP2} may return solutions $\vos{\Lambda}_{\rm gs}^*$ that, while close to optimal in the asymptotic scenario, lead to suboptimal phase-error rate estimations in the finite-key regime. This occurs because the concentration bounds used in Appendix~\ref{app:finite-key} introduce deviations that are not accounted for in~\cref{eq:dualSDP2}. Such situations may arise, for instance, when one incorporates redundant statistics into the SDP, such as non-detected events. Moreover, the quantities $\vos{\eta}$ and $\vos{\lambda}$ that minimize the single-photon phase-error probability in~\cref{eq:dualSDP2} may be very large in some cases, which can introduce a significant penalty in the finite-key concentration bounds. For example, the bound provided in~\cref{thm:correlations} depends on the minimum and maximum possible outcomes of the RVs $\chiT{u}$, which are given by the eigenvalues of $\Tobs{A}$---together with the value zero---and, by extension, depend on the solutions $\set{\lambda_k^*}_k$. To account for this, one can always introduce additional constraints in the SDP that limit the maximum magnitude of these eigenvalues. For instance, the SDP in~\cref{eq:dualSDP2} can be modified as
\begin{equation}\label{eq:dualSDP3}
\begin{split}
    \min_{\vos{\Lambda}} \quad &  \sum_l\eta_lq^{\rm gs}_l  + \sum_k\lambda_kt^{\rm gs}_k\\
    \text{s.t.} \quad & \hat{E}^{\rm ph}_{AB}  \preceq \sum_l\eta_l\hat{Q}^l_{AB}  +\sum_k\lambda_k\hat{T}^k_{A}\otimes\identity_B,
    \\
    &-\gamma_{\Tsind}\identity_A \preceq \sum_k\lambda_k\hat{T}_A^k \preceq \gamma_{\Tsind}\identity_A,
\end{split}
\end{equation}
where the parameter $\gamma_{\Tsind}$ can be optimized for each considered scenario.

\subsubsection*{Simulations}
The secret-key rate attainable in the presence of multiple device imperfections has already been presented in~\cref{sec:Results} of the main text. To benchmark these results, \cref{fig:SKR_BB84_N_and_eps} compares the secret-key rate of the single-photon BB84 protocol achievable in the presence of SPFs and information leakage using both our numerical security proof and a state-of-the-art analytical security proof that accommodates the same device imperfections~\cite{curras2023security2}, and which is known to be nearly optimal in the asymptotic regime~\cite{curras2025numerical,pereira2025optimal}. Here we do not consider bit/basis correlations because the simulations in~\cite{curras2023security2} do not include such a scenario.

As before, we fix the security parameters to $\epsilon_{\rm sec}=\epsilon_{\rm EV}=10^{-10}$, and, for simplicity, set $\epsilon_{\rm PA}=\epsilon_{\rm sec}/2$, $\epsilon_{\rm pro}=(\epsilon_{\rm sec}/4)^2$, and $\epsilon_{\rm proK}=\epsilon_{\rm proB}=\epsilon_{\rm proS}=\epsilon_{\rm pro}/6$ (see Appendix~\ref{app:finite-key} for further details). We consider a noiseless fiber-based quantum channel model with attenuation coefficient $\alpha_{\rm dB}=0.2$~dB/km, and assume that Bob's detectors have a detection efficiency $\eta_{\rm det}=0.73$~\cite{pittaluga2021600} and a dark-count probability $p_d=10^{-6}$. The SPF parameter is fixed to $\Delta\theta=0.063$ \cite{xu2015experimental,honjo2004differential}. For our method, we use \texttt{patternsearch} and/or \texttt{surrogateopt} routines from Matlab's Global Optimization Toolbox~\cite{mathworks_global_optimization_toolbox_2025} to optimize, for each distance, the probability that Alice selects the $Z$ basis (which, for simplicity, we take to be equal for Alice and Bob), as well as the parameters $\ptrash$ and $\gamma_{\Tsind}$. For the analytical approach, we similarly optimize the probability that Alice selects the $Z$ basis (again taken equal for Bob) and the probability $p_{Z_C}$ (see~\cite{curras2023security2} for further details). In \cref{fig:SKR_BB84_N_and_eps}a, we fix $N=10^{10}$ and consider different values of $\epsleak$, while in \cref{fig:SKR_BB84_N_and_eps}b we fix $\epsleak=10^{-5}$ and consider different values of $N$. The simulations show that our numerical security proof achieves better results than the analytical approach across all the considered regimes.

\begin{figure}
    \centering
    \includegraphics[width=0.95\columnwidth]{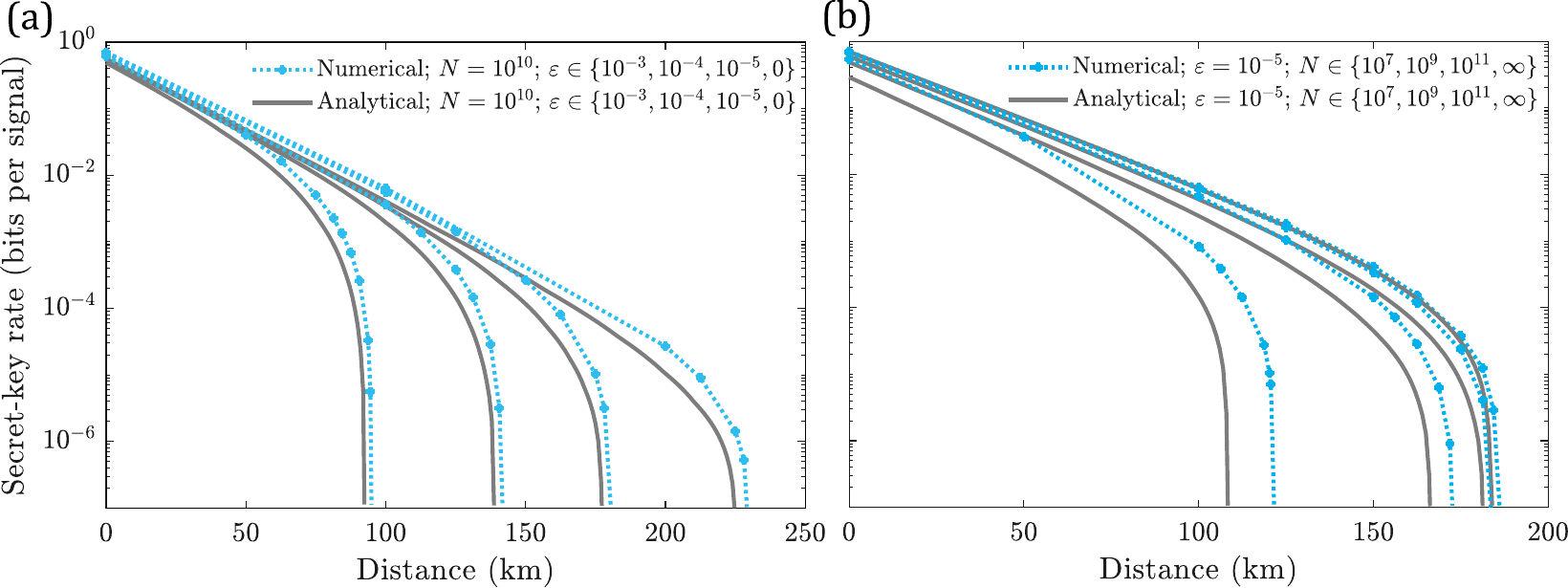}
    \caption{Optimized secret-key rate of a single-photon BB84 scheme obtained with our numerical security proof (blue dotted lines) and the analytical security proof of~\cite{curras2023security2} (gray solid lines). We consider SPFs with parameter $\Delta\theta=0.063$ \cite{xu2015experimental,honjo2004differential}, a pessimistic dark-count probability $p_d=10^{-6}$, and a detection efficiency of Bob's detectors equal to $\eta_{\rm det}=0.73$ \cite{pittaluga2021600}. In \cref{fig:SKR_BB84_N_and_eps}a, we fix $N=10^{10}$ and run the simulations for different values of $\epsleak$, while in \cref{fig:SKR_BB84_N_and_eps}b we fix $\epsleak=10^{-5}$ and run the simulations for different values of $N$.\label{fig:SKR_BB84_N_and_eps}}
\end{figure}

\subsection{MDI-QKD type protocols}\label{subseq:MDI}
\subsubsection*{General setting}
\begin{figure*}
\begin{mybox}{Actual MDI-type Protocol}{APMDI}
\begin{enumerate}
    \item In each round $u\in\set{1,\dots,N}$, 
    \begin{enumerate}
        \item With probability $\piA{a}$ ($\pB{b}$), $a,b\in\set{0,1,\dots,n}$, Alice (Bob) prepares a system $\tilde{A}$ ($\tilde{B}$) in the state $\sigma_{\tilde{A}_u}^{a^u_{u-L}}$ ($\sigma_{\tilde{B}_u}^{b^u_{u-L}}$) and sends this system to an intermediate station controlled by an untrusted party. We consider here that the settings $\set{0,1}$ belong to the key-generation basis, which we denote by the $Z$ basis.
        \item The untrusted party jointly measures the incoming signals and announces a measurement outcome $\Omega\in\set{\Omega_0,\Omega_1,\dots,\Omega_{m-1}}$ via a classical register $C$.
    \end{enumerate}
    \item 
    Alice and Bob publicly exchange some classical information and decide, for each round, if it is considered a \textit{key} or a \textit{test} round. The key rounds are used by Alice and Bob to generate their sifted bits, while the test rounds are used for parameter estimation and all the parameters $a,b$ are revealed.
    Alice and Bob record the number of key rounds ($\rv{N}_{\rm key}^{\Omega}$) associated with each $\Omega$, as well as the number of test rounds for each combination of $a,b,\Omega$ ($\rv{N}_{\rm test}^{a,b,\Omega}$).
    \item Alice and Bob perform error correction, error verification, and privacy amplification.
\end{enumerate}
\end{mybox}
\end{figure*}

A generic description of a measurement-device-independent (MDI) type QKD protocol is given in Box~\ref{box:APMDI}. We assume that~\cref{eq:assumption1thm_main} holds for both Alice's and Bob's transmitted states for a certain $\epsleak\geq \epsleak^{i_u}\quad \forall i_u$, and so we can apply~\cref{lemma:density_op_correlations} (see \cref{app:correlated_sources}) and consider that the emitted quantum states are pure. In particular, let $\ket{\Psi}_{(A\tilde{A}B\tilde{B})_1^N}$ be the joint state of all Alice's and Bob's ancilla systems $A$ and $B$, and the transmitted systems $\tilde{A}$ and $\tilde{B}$, in an entanglement-based view of the protocol. This state can be written as
\begin{equation}\label{eq:rhoAAprimaBBprima_MDI}
\begin{split}
     \ket{\Psi}_{(A\tilde{A}B\tilde{B})_1^N}
     &=
     \sum_{a_1^N}\sum_{b_1^N}\bigotimes_{u=1}^{N}\sqrt{\piA{a_u}\pB{b_u}} \ket{a_u}_{A_u}\ket{b_u}_{B_u}\ket*{\psi_{a_{u-L}^u}}_{\tilde{A}_u}\ket*{\psi_{b_{u-L}^u}}_{\tilde{B}_u}.
\end{split}
\end{equation}
In the channel, Eve performs any arbitrary joint measurement on all systems $\tilde{A}$ and $\tilde{B}$ and announces an outcome $\Omega_u$ for each round through a classical register $C_u$. As before, we shall drop the round index $u$ in the notation below whenever it is clear from the context that we refer to a single round. In each round, Alice and Bob perform a POVM $\set{\hat{S}_{ABC}^{a,b,\Omega,{\rm key}}}_{a,b,\Omega}\cup\set{\hat{S}_{ABC}^{a,b,\Omega}}_{a,b,\Omega}$ on the systems $ABC$, where
\begin{equation}
\begin{split}
    \hat{S}_{ABC}^{a,b,\Omega,{\rm key}} = 
    p_{\text{key}|a,b,\Omega} \dyad{a}_A\otimes\dyad{b}_B\otimes\dyad{\Omega}_C,
    \qquad
    \hat{S}_{ABC}^{a,b,\Omega} = 
    p_{\text{test}|a,b,\Omega} \dyad{a}_A\otimes\dyad{b}_B\otimes\dyad{\Omega}_C,
\end{split}
\end{equation}
with $p_{\text{tag}|a,b,\Omega}$ being the conditional probability that a round is assigned as $\text{tag}\in\set{\text{key},\text{test}}$ given $a,b,\Omega$, which satisfies $p_{\text{key}|a,b,\Omega}+p_{\text{test}|a,b,\Omega}=1$. Besides, since we assume that Alice and Bob extract key from the rounds in which $a,b\in\set{0,1}$, we have that $p_{\text{key}|a,b,\Omega}=0$ if $a\notin\set{0,1}$ or $b\notin\set{0,1}$. On the other hand, in the VP, Alice and Bob perform in each round the alternative POVM $\set{\hat{V}_{ABC}^{a,b,\Omega}}_{a,b,\Omega}\cup\set{\hat{S}_{ABC}^{a,b,\Omega}}_{a,b,\Omega}$ on $ABC$, where
\begin{equation}
\begin{split}
    \hat{V}_{ABC}^{a,b,\Omega} &=
    p_{\text{key}|a,b,\Omega} \dyad{a_X}_A\otimes\dyad{b_X}_B\otimes\dyad{\Omega}_C.
\end{split}
\end{equation}
As in the BB84-type scenario, the phase-error operator $\hat{E}^{\rm ph}_{ABC}$ can be written as a sum of operators $\hat{V}_{ABC}^{a,b,\Omega}$, with its specific form being dependent on the QKD scheme considered. 

With the previous ingredients, we shall consider, for illustration purposes, the following modified version of the single-round primal SDP introduced in~\cref{eq:primal1}
\begin{equation}\label{eq:SDPprimalMDI}
    \begin{split}
        \max_{\rho_{ABC}} \quad &  \Tr\set{\hat{E}^{\rm ph}_{ABC}\rho_{ABC}}\\
        \text{s.t.} \quad & \rho_{ABC}\succeq 0, \\
        \quad & \Tr\set{\identity_{AB}\otimes\ketbra{\Omega}{\Omega'}_C \rho_{ABC} } = 0, \quad \Omega\neq\Omega',\\
        \quad & \Tr\set{\hat{T}^k_{AB}\otimes \identity_{C} \rho_{ABC}} = t_k,\quad \forall k,\\
        \quad & \Tr\set{\hat{Q}^l_{ABC} \rho_{ABC}} = q_l,\quad \forall l,
    \end{split}
\end{equation}
where
\begin{equation}\label{eq:Ql_and_ql_MDI}
\begin{split}
    \hat{Q}^l_{ABC}&:= \sum_{\substack{a,b,\Omega\\p_{\text{test}|a,b,\Omega}>0}}c_{a,b,\Omega}^l\dyad{a}_A\otimes\dyad{b}_B\otimes\dyad{\Omega}_C,\\
    q_l&:= \sum_{\substack{a,b,\Omega\\p_{\text{test}|a,b,\Omega}>0}}c_{a,b,\Omega}^l \piA{a}\pB{b}p_{\Omega|a,b},
\end{split}
\end{equation}
and where the operators $\hat{T}^k_{AB}$ now serve to characterize both systems $A$ and $B$ (see the example in the next subsection about a coherent-light-based MDI-type protocol). 

The SDP provided in~\cref{eq:SDPprimalMDI} is essentially identical to that in~\cref{eq:primal1} except for its second constraint, which enforces a block-diagonal structure on the density operator with respect to the announcement basis $\set{\ket{\Omega}_C}$, since the announcement system $C$ is classical. In any case, since all the operators in~\cref{eq:SDPprimalMDI} have a block-diagonal structure for system $C$, such a second constraint can be neglected without affecting the optimal solution~\cite{zhou2022numerical}. A suitable dual SDP for this MDI scenario is simply given by
\begin{equation}\label{eq:dualSDP-MDI}
\begin{split}
    \min_{\vos{\Lambda}} \quad &  \sum_l\eta_lq^{\rm gs}_l  + \sum_k\lambda_kt^{\rm gs}_k\\
    \text{s.t.} \quad & \hat{E}^{\rm ph}_{ABC}  \preceq \sum_l\eta_l\hat{Q}^l_{ABC}  +\sum_k\lambda_k\hat{T}^k_{AB}\otimes\identity_C,
    \\
    & -\gamma_{\Tsind}\identity_{AB}\preceq \sum_k\lambda_k\hat{T}_{AB}^k  \preceq \gamma_{\Tsind}\identity_{AB},
\end{split}
\end{equation}
where $\vos{\Lambda}:=(\eta_1,\dots,\eta_{n_{\rm Q}},\lambda_1,\dots,\lambda_{n_{\rm T}})$ and, similar to the BB84 example, we have included a parameter $\gamma_{\Tsind}$ to cap the maximum eigenvalue of $\Tobs{AB}$ (defined analogously to~\cref{eq:Tk_obs_definition}). We remark that the primal SDP in~\cref{eq:SDPprimalMDI} is introduced only for illustration purposes, since constructing a suitable dual SDP is equally straightforward: one simply inserts, in the desired operator inequality, a linear combination of all available measurement operators, and takes as the objective function the expected value of this linear combination under normal operation of the protocol.

After solving~\cref{eq:dualSDP-MDI}, one can follow analogous steps to those in~\cref{sec:SecurityAnalysisNew} and Appendix~\ref{app:finite-key} to obtain a probabilistic bound on the phase-error rate. The only relevant difference is that in the VEP associated with the MDI-QKD protocol, the RVs $\rv{\chi}_{{\rm Q},l}^{u}$ take the values $\frac{c_{a,b,\Omega}^l}{p_{\text{test}|a,b,\Omega}}$.

The secret-key length of the protocol can be computed similarly to that in~\cref{sec:SecurityAnalysisNew}. In particular, assuming that the promise
\begin{equation}
    \Pr[\rv{e}_{\rm ph}> e_{\rm ph}^{\rm U}(\set{\rv{N}_{\rm key}^{\Omega}}_\Omega,\set{\rv{N}_{\rm test}^{a,b,\Omega}}_{a,b,\Omega})]\leq \epsilon_{\rm pro},
\end{equation}
holds, then the protocol is $\epsilon_{\rm tot}$-secure with secret-key length
\begin{equation}
\begin{split}
    \rv{l}_{\rm key} ={}&  \rv{N}_{\rm key}\left[1-h\left(e_{\rm ph}^{\rm U}(\set{\rv{N}_{\rm key}^{\Omega}}_\Omega,\set{\rv{N}_{\rm test}^{a,b,\Omega}}_{a,b,\Omega})\right)\right]-\rv{\lambda}_{\rm EC}
    -2\log(1/2\epsilon_{\rm PA})-\log(2/\epsilon_{\rm EV}),
\end{split}
\end{equation}
where $\epsilon_{\rm tot}=\epsilon_{\rm EV}+\epsilon_{\rm sec}$, $\epsilon_{\rm sec}=2\sqrt{\epsilon_{\rm pro}}+\epsilon_{\rm PA}$ (see~\cref{sec:SecurityAnalysisNew} for the definition of the different parameters), and $\rv{N}_{\rm key}=\sum_{\Omega}\rv{N}_{\rm key}^{\Omega}$.

\subsubsection*{Example: coherent-light-based MDI-type protocol}

To illustrate the MDI scenario, we consider the coherent-light-based MDI type of protocol proposed in~\cite{navarrete2021practical}. In this scheme, each of Alice and Bob prepares three states, which in the absence of device imperfections are given by
\begin{equation}\label{eq:ideal_states_MDI}
    \ket{\phi_0}=\ket{\sqrt{\mu}},
    \qquad
    \ket{\phi_1}=\ket{-\sqrt{\mu}},
    \qquad
    \ket{\phi_2}=\ket{\rm vac},
\end{equation}
where $\ket{\sqrt{\mu}}$ and $\ket{-\sqrt{\mu}}$ are coherent states with mean photon number $\mu$ and opposite phases, and $\ket{\rm vac}$ represents the vacuum state. In the intermediate node, the untrusted third party is supposed to interfere the incoming signals with a 50:50 beamsplitter, whose output ports are monitored with two threshold single-photon detectors, and then announce the measurement results (see~\cite{navarrete2021practical} for further details). There are three possible measurement results: $\Omega=\Omega_{\rm c}$ when only the detector associated with "constructive interference" clicks; $\Omega=\Omega_{\rm d}$ when only the detector associated with "destructive interference" clicks; and $\Omega=\Omega_{\varnothing}$ if none of the previous two events occurs. Only the $Z$-basis states $\ket{\sqrt{\mu}}$ and $\ket{-\sqrt{\mu}}$ are used for key generation---i.e., $p_{{\rm key}|a,b,\Omega}=0$ if either $a=2$ or $b=2$---while all states are used for testing the channel. In the key generation rounds where $\Omega=\Omega_{\rm d}$, Bob flips his bit value $b$.

In this protocol, the phase-error operator can be defined as $\hat{E}^{\rm ph}_{ABC}=\sum_{\Omega}\hat{E}_{\rm ph}^{\Omega}\otimes\dyad{\Omega}_C$, where~\cite{navarrete2021practical}
\begin{equation}
\hat{E}_{\rm ph}^{\Omega}=
p_{{\rm key}|Z}\dyad{0_X}_A\otimes\dyad{0_X}_B + p_{{\rm key}|Z}\dyad{1_X}_A\otimes\dyad{1_X}_B,  
\quad  \Omega\in\set{\Omega_{\rm c},\Omega_{\rm d}}.
\end{equation}
That is, here we fix $p_{{\rm key}|a,b,\Omega}=p_{{\rm key}|Z}$ for any $a,b\in\set{0,1}$ and $\Omega\in\set{\Omega_{\rm c},\Omega_{\rm d}}$.

As for the coefficients $c_{a,b,\Omega}^l$ in~\cref{eq:Ql_and_ql_MDI}, we use all the available statistics for each $\Omega\in\set{\Omega_{\rm c},\Omega_{\rm d}}$. Specifically, we consider some binary coefficients such that, for each $l$, $c^{l}_{a,b,\Omega}$ is equal to 1 for one of the different tuples $(a,b,\Omega)$, with $\Omega\in\set{\Omega_{\rm c},\Omega_{\rm d}},a,b\in\set{0,1,2}$, and zero otherwise. This means that $l\in\set{1,\dots,18}$. Given these coefficients, one can readily determine the guesses $q^{\rm gs}_l$ and the operators $\hat{Q}^l_{ABC}$ following~\cref{eq:Ql_and_ql_MDI}. 

As for the operators associated with the fictitious characterization rounds, now both users must perform these fictitious measurements in the VEP, which occurs with a certain probability $\ptrash$. To define these operators, we consider the mapping $i=a+3b$, such that the joint transmitted states can be re-defined using a single index as $\ket{i}_{AB}\equiv\ket{a}_{A}\ket{b}_{B}$. With this, $\hat{T}^{i,j}_{AB}$ is defined identically to the $\hat{T}^{i,j}_A$ given in~\cref{eq:Tk_BB84}, now with $i,j\in\set{0,1,\dots,8}$.

Given~\cref{eq:ideal_states_MDI,eq:Tk_BB84} together with $\epsleak$, one can use the guesses $t_{i,j}^{\rm gs}=\Tr{\hat{T}^{i,j}_{AB}\rho_{AB}^{\epsleak'}}$, where $\rho_{AB}^{\epsleak'}$ is the state of system $AB$, which is completely analogous to~\cref{eq:rhoA_eps_BB84}. In particular, in this case we assume that the transmitted states of both users admit a description like~\cref{eq:assumption1extra_main} for certain $\tilde{\epsleak}$. These two assumptions can then be relaxed to the following single assumption \cite{curras2023security2,curras2025numerical}
\begin{equation}
\begin{split}
    \ket*{\psi_{i^u_{u-L}}^{\epsleak}}_{\tilde{A}\tilde{B}}
    \equiv
    \ket*{\psi_{a^u_{u-L}}^{\tilde{\epsleak}}}_{\tilde{A}}\ket*{\psi_{b^u_{u-L}}^{\tilde{\epsleak}}}_{\tilde{B}}
    =&(1-\tilde{\epsleak})\ket*{\phi_{a_u}}_{\tilde{A}}\ket*{\phi_{b_u}}_{\tilde{B}} 
    + \tilde{\epsleak}\ket*{\phi_{a^u_{u-L}}^{\perp}}_{\tilde{A}}\ket*{\phi_{b^u_{u-L}}^{\perp}}_{\tilde{B}}\\
    &
    +\sqrt{1-\tilde{\epsleak}}\sqrt{\tilde{\epsleak}}\ket*{\phi_{a^u_{u-L}}^{\perp}}_{\tilde{A}}\ket*{\phi_{b_u}}_{\tilde{B}} + \sqrt{1-\tilde{\epsleak}}\sqrt{\tilde{\epsleak}}\ket*{\phi_{a_u}}_{\tilde{A}}\ket*{\phi_{b^u_{u-L}}^{\perp}}_{\tilde{B}}\\
    =&\sqrt{1-\epsleak}\ket*{\phi_{i_u}}_{\tilde{A}\tilde{B}} + \sqrt{\epsleak}\ket*{\phi_{i^u_{u-L}}^{\perp}}_{\tilde{A}\tilde{B}},
\end{split}
\end{equation}
where $\epsleak=1-(1-\tilde{\epsleak})^2\approx 2\tilde{\epsleak}$, $\ket*{\phi_{i_u}}_{\tilde{A}\tilde{B}}:=\ket*{\phi_{a_u}}_{\tilde{A}}\ket*{\phi_{b_u}}_{\tilde{B}}$, and, for any $i\equiv i_u$, $\ket*{\phi_{i^u_{u-L}}^{\perp}}_{\tilde{A}\tilde{B}}$ is an uncharacterized state that is orthogonal to $\ket*{\phi_{i_u}}_{\tilde{A}\tilde{B}}$. As the joint density operator $\rho_{AB}^{\epsleak'}$ of Alice and Bob's ancilla systems $A$ and $B$ depends on the unknown states $\ket*{\phi^{\perp}_{i^u_{u-L}}}_{\tilde{A}\tilde{B}}$, we consider a pessimistic scenario in which $\braket*{\phi^{\perp}_{i^u_{u-L}}}{\phi^{\perp}_{j^u_{u-L}}}_{\tilde{A}\tilde{B}}=\braket*{\phi^{\perp}_{i^u_{u-L}}}{\phi_{j_u}}_{\tilde{A}\tilde{B}}=0$ for $i_u\neq j_u$. Again, this is not an assumption required by the security proof, but only a choice made to select the values of the unconfirmed guesses $t_k^\mathrm{gs}$, which can be set arbitrarily. Finally, each tuple $(i,j)$ can then be re-mapped to a different value of the index $k\in\set{1,\dots,81}$, such that we can rewrite $\hat{T}^{i,j}_{AB}$ as $\hat{T}^{k}_{AB}$ and $t_{i,j}^{\rm gs}$ as $t_{k}^{\rm gs}$. 

Once all the operators and guesses are specified, one can run the SDP given in~\cref{eq:dualSDP-MDI} and follow the steps described in~\cref{sec:SecurityAnalysisNew} and Appendix~\ref{app:finite-key} to obtain a bound on the phase-error rate.

\subsubsection*{Simulations}
The results are shown in~\cref{fig:SKR_MDI}. Again, we fix the security parameters to $\epsilon_{\rm sec}=\epsilon_{\rm EV}=10^{-10}$, and set for simplicity $\epsilon_{\rm PA}=\epsilon_{\rm sec}/2$, $\epsilon_{\rm pro}=(\epsilon_{\rm sec}/4)^2$, $\epsilon_{\rm proK}=\epsilon_{\rm proB}=\epsilon_{\rm proS}=\epsilon_{\rm pro}/5$. Besides, we consider a lossy but noiseless fiber-based channel model with attenuation coefficient $\alpha_{\rm dB}=0.2$ dB/km, and threshold single-photon detectors with detection efficiency $\eta_{\rm det}=0.73$ and dark-count probability $p_{\rm d}=10^{-8}$ \cite{pittaluga2021600}. In the simulations, we use  \texttt{patternsearch} and \texttt{surrogateopt} routines from Matlab's Global Optimization Toolbox~\cite{mathworks_global_optimization_toolbox_2025} to optimize the values of the parameters $p_Z = \piA{0}+\piA{1} = \pB{0}+\pB{1}$, $p_{{\rm key}|a,b,\Omega}$, $\alpha$, $\mu$ and $\gamma_{\Tsind}$ for each distance. The green (purple) lines illustrated in~\cref{fig:SKR_MDI} correspond to $N\to\infty$ ($N=10^{10}$), and we consider side-channel parameters $\epsleak\in\set{10^{-5},10^{-7},0}$. For comparison purposes, we include in the figure the secret-key rate attainable with both our numerical method (dotted lines) and the analytical security proof of~\cite{curras2023security2} (solid lines). Notably, our numerical security proof largely outperforms the latter analysis in all the considered scenarios. This shows the potential of our numerical security framework to provide performance improvements for non-standard scenarios beyond BB84, for which existing analytical security proofs may be suboptimal.

\begin{figure}
    \centering
    \includegraphics[width=0.55\columnwidth]{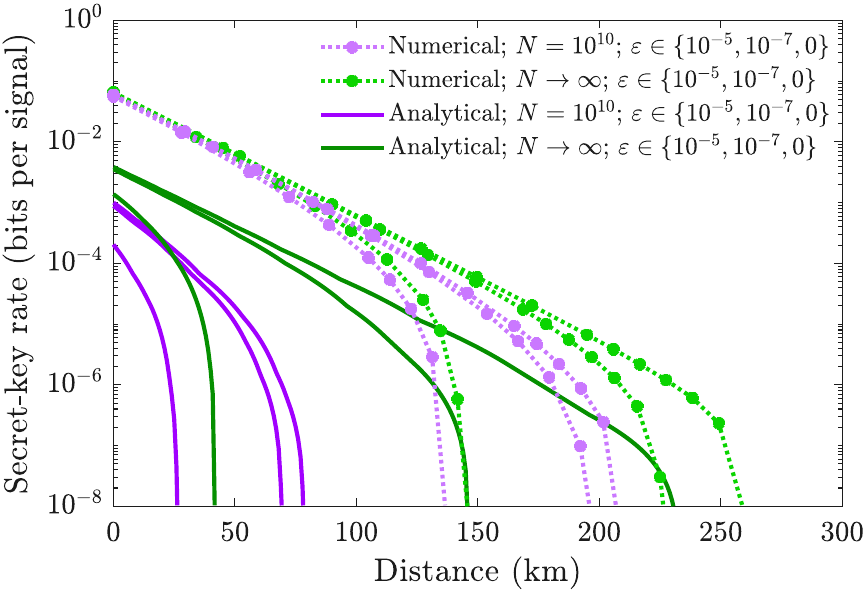}
    \caption{
    Secret-key rate of a coherent-light-based MDI QKD scheme~\cite{navarrete2021practical} in the presence of device imperfections. The green (purple) lines illustrate the scenarios $\epsleak\in\set{10^{-5},10^{-7},0}$ and $N\to\infty$ ($N=10^{10}$). The dotted (solid) lines correspond to the results provided by our numerical security proof (the analytical security proof of~\cite{curras2023security2}). In the simulations, we consider a lossy but noiseless channel model (see e.g.~\cite{curras2023security2}), in which the detectors have a detection efficiency $\eta_{\rm det}=0.73$ and a dark-count probability $p_{\rm d}=10^{-8}$  \cite{pittaluga2021600}. The security parameters are set to $\epsilon_{\rm sec}=\epsilon_{\rm EV}=10^{-10}$.}\label{fig:SKR_MDI}
\end{figure}

\subsection{Decoy-state QKD protocols}\label{app:general_protocols_decoy_state}
\subsubsection*{General setting}
\begin{figure*}
\begin{mybox}{Actual decoy-state protocol}{APdec}
\begin{enumerate}
     \item In each round\footnote{In this box we omit the round index $u$ from the quantities whenever it is clear from the context.} $u\in\set{1,\dots,N}$, 
    \begin{enumerate}
        \item Alice chooses a pair of bit/basis ($a$) and intensity ($\mu$) settings, which for conciseness we group in a single index $i\in\set{0,1,\dots,\nstates-1}$, with probability $\piA{i}\equiv p_{a,\mu}^{\rm A}$. Then she prepares an $n$-photon state with density matrix $\sigma_C^{n,i^u_{u-L}}$---which in general may depend not only on the current setting choices $i \equiv i_u$ but also on the previous settings $i^{u-1}_{u-L}$---with probability $p_{n|i^{u}_{u-L}}$ and sends system $C$ to Bob through the quantum channel.
        \item Bob selects the measurement basis $\beta\in\set{Z,X_1,\dots,X_{n_M}}$ with probability $p_{\beta}^{\rm B}$, and measures the received system $B$ with the POVM $\set{\hat{\Gamma}^{b|\beta}_B}_{b}$, obtaining the measurement outcome $b$. We consider that the basis $Z$ has three possible outcomes $b\in\set{0,1,\perp}$.
    \end{enumerate}
    \item 
    Alice and Bob publicly exchange classical information and decide if the round is assigned as \textit{key} or \textit{test} round. The key rounds---i.e., those in which $\beta=Z$, $a,b\in\set{0,1}$, and $\mu=0$---are used to generate the users' sifted keys. In the test rounds---i.e., the remaining rounds---we assume that the users reveal all the parameters $a,\mu,\beta$, and $b$. Alice and Bob record the number of key rounds ($\rv{N}_{\rm key}$), as well as the number of test rounds for each combination of the parameters $\mu,a,\beta$, and $b$ ($\rv{N}_{\rm test}^{a,\mu,\beta,b}$).
    \item Alice and Bob perform error correction, error verification, and privacy amplification.
\end{enumerate}
\end{mybox}
\end{figure*}
In a decoy-state QKD scheme (see Box \ref{box:APdec}), Alice generates phase-randomized weak-coherent pulses (PRWCPs). For concreteness, here we consider perfect continuous phase randomization, such that the phase-averaged pulses can be described as classical mixtures of Fock states; however, the security analysis also applies to sources with discrete or imperfect continuous phase randomization (see Appendix~\ref{app:imperfect_PR}). These pulses then pass through a bit/basis encoding modulator---e.g., a polarization modulator in the case of polarization encoding---and an intensity modulator that encodes the intensity of each pulse according to Alice's settings $a\in\set{0,1,\dots,n_{\rm enc}-1}$ and $\mu\in\set{0,1,\dots,n_{\rm int}-1}$, which are selected with probability $\piA{a,\mu}$. The encoding settings $a\in\set{0,1}$ correspond to Alice's $Z$ basis. For simplicity, we introduce an index $i\in\set{0,1,\dots,\nstates-1}$ to label the pairs $(a,\mu)$ selected by Alice, with $\nstates$ being the number of pairs with nonzero probability $\piA{i}\equiv\piA{a,\mu}\neq 0$---e.g., $\nstates=n_{\rm enc}n_{\rm int}$ in typical setups where the bit/basis and intensity settings are selected independently. Here, we consider that the transmitted Fock states satisfy the fidelity constraint
\begin{equation}\label{eq:assumption1thm_decoy}
    F\left(\sigma_{C_u}^{n_u,i^u_{u-L}},\ket{\phi_{n_u,i_u}}_{C_u}\right)\geq 1-\epsleak^{n_u,i_u},\quad \forall n_u,i^{u}_{u-L}
\end{equation}
for certain parameters $\epsleak^{n_u,i_u}\geq0$---which for simplicity we fix to the same value, i.e., $\epsleak^{n_u,i_u}=\epsleak$---where $n_u$ is the number of photons contained in the pulse, $L$ denotes the correlation length, and the states $\ket{\phi_{n_u,i_u}}_{C_u}$ are some predefined reference states that can be chosen freely. Note that the states provided in~\cref{eq:assumption1thm_decoy} may, in general, depend on previous bit/basis and intensity setting choices $i_{u-L}^{u-1}$. Moreover, we can use\footnote{Note that~\cref{lemma:density_op_correlations} can be applied to the decoy-state setting by simply mapping each pair $(n_u,i_u)$ to a single $i_u'$.} \cref{lemma:density_op_correlations} (see \cref{app:correlated_sources}) and, without loss of generality, consider for the security proof that the transmitted Fock states are pure. Specifically, the state emitted each round can be generally described as
\begin{equation}\label{eq:states_decoy}
    \ket*{\psi_{n_u,i^{u}_{u-L}}^{\epsleak}}_{C_u} = \sqrt{1-\epsleak}\ket{\phi_{n_u,i_u}}_{C_u} + \sqrt{\epsleak}\ket*{\phi_{n_u,i^{u}_{u-L}}^{\perp}}_{C_u},
\end{equation}
where $\ket*{\phi_{n_u,i^{u}_{u-L}}^{\perp}}$ can be any unknown state satisfying $\braket*{\phi_{n_u,i_u}}{\phi_{n_u,i^{u}_{u-L}}^{\perp}}=0$. 
We remark that, whenever some of the imperfections mentioned above are known to be absent---e.g., if side-channel leakage and correlations only affect the bit/basis information---the analysis can be tightened by adding the corresponding constraints to the SDPs in~\cref{thm:correlations_decoy} (see \cref{rem:remark_extra_constraints} in \cref{app:correlated_sources} for further details).
In the entanglement-based view of the protocol, we can consider that Alice prepares
\begin{equation}\label{eq:global_correlated_state_decoy_main}
\begin{split}
\ket{\Psi^{\epsleak}}_{A_1^NR_1^NC_1^N}
&=
\sum_{n_1^N,i_1^N}\sqrt{p_{n_1^N,i_1^N}}\bigotimes_{u=1}^N \ket{n_u}_{R_u}\ket{i_u}_{A_u}\ket*{\psi_{n_u,i_{u-L}^u}^{\epsleak}}_{C_u}
\\
&=
 \sum_{n_1^N,i_1^N}\bigotimes_{u=1}^N \sqrt{\piA{i_u}p_{n_u|i^u_{u-L}}} \ket{n_u}_{R_u}\ket{i_u}_{A_u}\ket*{\psi_{n_u,i_{u-L}^u}^{\epsleak}}_{C_u},
\end{split}
\end{equation}
where $p_{n_1^N,i_1^N}$ is the joint probability of the sequences $n_1^N$ and $i_1^N$, $A_1^N$ denotes $N$ ancilla systems encoding the setting selections $\ket{i}_A\equiv\ket{a,\mu}_A$, and $R_1^N$ denotes $N$ shield systems storing the photon numbers of the transmitted signals. 
Importantly, in practice, the assumption in \cref{eq:assumption1thm_decoy} is needed only for the subset of Fock states with $n_u\leq n_{\rm cut}$, where $n_{\rm cut}$ is introduced later. Indeed, as we show below, only the reference states with $n_u\leq n_{\rm cut}$ enter the security proof. For $n_u > n_{\rm cut}$, one may simply take Alice to prepare any purification of $\sigma_{C_u}^{n_u,i^u_{u-L}}$.
Note that this remains compatible with~\cref{lemma:density_op_correlations}, as such purifications can always be paired with trivial---generally unknown---reference states identical to them, therefore satisfying~\cref{eq:assumption1thm_decoy}. Besides,
we note that the probability that Alice's transmitted signal in round $u$ contains $n_u$ photons, i.e., $p_{n_u|i^{u}_{u-L}}$, may in general depend not only on the current intensity setting $\mu_u$, but also on the current bit/basis setting $a_u$, as well as on the previous intensity and bit/basis settings $i^{u-1}_{u-L}$. Alice sends systems $C^N_1$ through the quantum channel to Bob, who receives the systems $B^N_1$. We denote by $\rho_{RAB}$ the state of systems $A_u\equiv A$, $R_u\equiv R$ and $B_u\equiv B$ corresponding to the round $u$.

In this scenario, we can define an SDP for a target single-round measurement operator $\dyad{1}_{R}\otimes \hat{O}_{AB}$ for certain $\hat{O}_{AB}\in\set{-\hat{D}_{AB}^{(1)},\hat{E}_{AB}^{{\rm ph},(1)}}$, where $\hat{D}_{AB}^{(1)}$ represents the single-photon sifted-key detection operator---which is associated with a detection event in rounds in which Bob selects $Z$ and Alice selects $Z$ and $\mu=0$---and $\hat{E}_{AB}^{{\rm ph},(1)}$ denotes the single-photon phase-error operator. These operators can be written as
\begin{equation}
    \begin{split}
       \hat{D}^{(1)}_{AB} & :=  \sum_{a,b\in\set{0,1}}\hat{S}_{AB}^{a,0,Z,b} = \sum_{a,b\in\set{0,1}}\hat{V}_{AB}^{a,b},
       \\
       \hat{E}_{AB}^{{\rm ph},(1)} & := \hat{V}_{AB}^{0,1} + \hat{V}_{AB}^{1,0},
    \end{split}
\end{equation}
where, similarly to~\cref{sec:SecurityAnalysisNew}, we have that
\begin{equation}
    \begin{split}
       \hat{S}_{AB}^{a,\mu,\beta,b} & := \pB{\beta} \dyad{a,\mu}_{A} \otimes \hat{\Gamma}_{B}^{b|\beta},
       \\
       \hat{V}_{AB}^{a,b} & := \pB{Z} \dyad{a_X,0}_{A} \otimes \hat{\Gamma}_{B}^{b|V}.
    \end{split}
\end{equation}
That is, the joint POVM that describes the actual entanglement-based protocol is $\set{\dyad{n}_R\otimes \hat{S}_{AB}^{a,\mu,\beta,b}}_{n,a,\mu,\beta,b}$, while the POVM of the virtual protocol is given by $\set{\dyad{1}_R\otimes \hat{V}_{AB}^{a,b}}_{a,b\in\set{0,1}}\cup \set{\dyad{n}_R\otimes \hat{S}_{AB}^{a,\mu,\beta,b}}_{(n,a,\mu,\beta,b)\notin\mathcal{K}}$, where here $\mathcal{K}:=\set{(n,a,\mu,\beta,b):n=1,\mu=0,\beta=Z,a,b\in\set{0,1}}$. With this, we can construct the following SDP for the single-round scenario:
\begin{equation}
    \begin{split}
        \max_{\rho_{RAB}} \quad &  \Tr\set{(\dyad{1}_{R}\otimes\hat{O}_{AB})\rho_{RAB}}\\
        \text{s.t.} \quad & \rho_{RAB}\succeq 0, \\
        \quad & \Tr\set{\identity_{R}\otimes\hat{Q}_{AB}^{l} \rho_{RAB}} = q_l,\qquad (\forall l)\\
        \quad & \Tr\set{\dyad{n}_{R}\otimes\hat{T}^k_A\otimes\identity_B \rho_{RAB}} = t_{k,n},\qquad (n\in\set{0,\dots,n_{\rm cut}},\forall k)\\
        \quad & \Tr{\sum_{n=n_{\rm cut}+1}^{\infty}\dyad{n}_{R}\otimes\dyad{i}_{A}\otimes\identity_{B} \rho_{RAB}} = t_{i,\infty},\qquad (i\in\set{0,1,...,\nstates-1})\\
    \end{split}
\end{equation}
where $n_{\rm cut}$ is a parameter that caps the maximum value of $n$ that is considered individually in the SDP. Here, $\hat{T}_A^k\equiv\hat{T}^{i,j}_A$ is defined as in~\cref{eq:Tk_BB84}.
Note that, to simplify the numerics, one might disregard pairs $(i,j)$ whose indices $i$ and $j$ correspond, respectively, to tuples $(a,\mu)$ and $(a',\mu')$ that simultaneously satisfy $a\neq a'$ and $\mu\neq\mu'$, as their impact on the estimation is expected to be small. As for the operators $\hat{Q}_{AB}^{l}$, they are defined analogously to those in~\cref{eq:linear_comb_c}, i.e.,
\begin{equation}\label{eq:Q_and_q_DS}
\begin{split}
    \hat{Q}_{AB}^{l} &= \sum_{a,\mu,\beta,b}
    c_{a,\mu,\beta,b}^{l}\dyad{a,\mu}_A\otimes\hat{\Gamma}^{b|\beta}_B,
    \\
    q_{l}&:= \sum_{a,\mu,\beta,b} c_{a,\mu,\beta,b}^{l} p_{a,\mu}^{A}p_{b|a,\mu,\beta},
\end{split}
\end{equation}
where the coefficients $c_{a,\mu,\beta,b}^{l}$ may, in general, depend on the target operator $\hat{O}_{AB}$. For simplicity, we omit this dependence---along with that inherited by $\hat{Q}_{AB}^{l}$ and $q_{l}$---from the notation. 
Moreover, note that one must impose $c_{0,0,Z,0}^{l}=c_{0,0,Z,1}^{l}=c_{1,0,Z,0}^{l}=c_{1,0,Z,1}^{l}$ to ensure consistency with the observed protocol statistics, since the bit value is not revealed in the sifted-key rounds of the actual protocol.

The corresponding dual SDP for the single-round scenario given certain unconfirmed guesses---which can be obtained similarly to the BB84 and MDI scenarios---is therefore
\begin{equation}\label{eq:dualSDP2_decoy}
\begin{split}
    \min_{\substack{\set{\lambda_{k,n}}_{n\in\set{0,\dots,n_{\rm cut}},k} \\ \set{\lambda_{i,\infty}}_i,\set{\eta_l}_l}} 
    \quad &  
    \sum_{l}\eta_l q_l^{\rm gs}+\sum_{k}\sum_{n=0}^{n_{\rm cut}}\lambda_{k,n} t_{k,n}^{\rm gs} + \sum_i\lambda_{i,\infty} t_{i,\infty}^{\rm gs}
    \\
    \text{s.t.} 
    \quad & \dyad{1}_{R}\otimes\hat{O}_{AB} 
    -\identity_{R}\otimes\sum_l\eta_l\hat{Q}_{AB}^{l}  
    -\sum_{k}\sum_{n=0}^{n_{\rm cut}}\lambda_{k,n}(\dyad{n}_{R}\otimes\hat{T}_A^k\otimes\identity_B)
    \\
    &
    -\Pi_R^{n>n_{\rm cut}}\otimes\sum_i\lambda_{i,\infty}\dyad{i}_A\otimes\identity_{B}
    \preceq 0,
\end{split}
\end{equation}
where we have introduced the projector $\Pi_R^{n>n_{\rm cut}}=\sum_{n=n_{\rm cut}+1}^{\infty}\dyad{n}_R$. The operator constraint of the previous SDP can be decomposed into independent constraints for the different photon-number subspaces, and thus we can remove the photon-number register $R$ by considering the following equivalent problem
\begin{equation}\label{eq:dualSDP2_decoy2}
\begin{split}
     \min_{\substack{\set{\lambda_{k,n}}_{n\in\set{0,\dots,n_{\rm cut}},k} \\ \set{\lambda_{i,\infty}}_i,\set{\eta_l}_l}} \quad &  \sum_{l}\eta_l q_l^{\rm gs}+\sum_{k}\sum_{n=0}^{n_{\rm cut}}\lambda_{k,n} t_{k,n}^{\rm gs} + \sum_i\lambda_{i,\infty} t_{i,\infty}^{\rm gs}
     \\
    \text{s.t.} 
    \quad & \hat{O}_{AB} \preceq \sum_l\eta_l\hat{Q}_{AB}^{l}  + \sum_{k}\lambda_{k,1}\hat{T}_A^k\otimes\identity_B,
    \\
    \quad & 0 \preceq \sum_l\eta_l\hat{Q}_{AB}^{l}  + \sum_{k}\lambda_{k,n}\hat{T}_A^k\otimes\identity_B, \quad (n\in\set{0,2,3,\dots,n_{\rm cut}})
    \\
    \quad & 0 \preceq \sum_l\eta_l\hat{Q}_{AB}^{l}  + \sum_i\lambda_{i,\infty}\dyad{i}_A\otimes\identity_{B}.
\end{split}
\end{equation}
Note that, similarly to~\cref{eq:dualSDP3}, one can always impose some extra constraints of the form $-\gamma_{{\Tsind},n}\identity_A \preceq \sum_{k}\lambda_{k,n}\hat{T}_A^k \preceq \gamma_{{\Tsind},n}\identity_A$ in~\cref{eq:dualSDP2_decoy2}, for certain $\gamma_{{\Tsind},n}$, to restrict the values of the maximum and minimum eigenvalues of the operators $\sum_{k}\lambda_{k,n}\hat{T}_A^k$. For simplicity here we consider a single parameter $\gamma_{\Tsind}$, i.e., we fix $\gamma_{{\Tsind},n}=\gamma_{\Tsind}$ for all $n$.

\begin{figure*}
\begin{mybox}[myboxcolor=green]{Virtual Estimation Protocols (VEPs) for the decoy-state protocol}{VEPdec}
\raggedright For conciseness, here we introduce a unified description with identical notation that is valid for both the $E_{AB}^{\rm ph,(1)}$-VEP and the $D_{AB}^{(1)}$-VEP. We remark, however, that these are two distinct and independent protocols\footnote{In particular, the states $\set{\ket{\omega}}_{\omega}$ in general differ between the two VEPs, as may the coefficients $c_{a,\mu,\beta,b}^l$, \textit{inter alia}.}.
\begin{enumerate}
    \item Alice prepares the entangled state $\ket{\Psi^{\epsleak}}_{A_1^NR_1^NC_1^N}$---defined in~\cref{eq:global_correlated_state_decoy_main} for certain $\epsleak$---and sends all systems $C_1^N$ through the quantum channel to Bob, who receives the systems $B_1^N$. 
    \item Then, in each round $u$, Alice and Bob perform the POVM
    \begin{equation}\nonumber
    \begin{split}
        &\set{(1-\ptrash)\dyad{1}_R\otimes \hat{V}_{AB}^{a,b}}_{a,b\in\set{0,1}}\cup \set{(1-\ptrash)\dyad{n}_R\otimes \hat{S}_{AB}^{a,\mu,\beta,b}}_{(n,a,\mu,\beta,b)\notin\mathcal{K}}\cup\set{\ptrash\dyad{\omega}_{RA}\otimes\identity_{B}}_{\omega},
    \end{split}
    \end{equation}
    on the systems $RAB$, where $\ptrash$ is the probability that a round is tagged as \LC{}, $\mathcal{K}:=\set{(n,a,\mu,\beta,b):n=1,\mu=0,\beta=Z,a,b\in\set{0,1}}$, and the states $\ket{\omega}_{RA}$ are eigenvectors of the operator $\Tobs{RA}$ defined in~\cref{eq:Tobs_RAB} with eigenvalues $\omega\in\bigcup_{n=0}^{n_{\rm cut}}\set{\omega_n}_{\omega_n}\cup\set{\lambda_{i,\infty}^*}_i$. As a result, they obtain the values of the following random variables:
    \begin{itemize}\label{it:RVs_decoy}
        \item[\small $\bullet$] $\rv{\chi}_{\rm key,1}^{u}$: it takes the value 1 if a $Z$-basis single-photon detection with intensity $\mu=0$ occurs, which is associated with the POVM elements $(1-\ptrash)\dyad{1}_R\otimes \hat{V}_{AB}^{a,b}$ with $a,b\in\set{0,1}$. Otherwise, it takes the value 0.
        \item[\small $\bullet$] $\rv{\chi}_{\rm ph}^{u}$: it takes the value 1 if a phase error occurs, which is associated with the POVM elements $(1-\ptrash)\dyad{1}_R\otimes \hat{V}_{AB}^{0,1}$ and $(1-\ptrash)\dyad{1}_R\otimes \hat{V}_{AB}^{1,0}$. Otherwise, it takes the value 0.
        \item[\small $\bullet$] $\rv{\chi}_{\rm key}^{u}$: it takes the value 1 if a $Z$-basis detection with intensity $\mu=0$ occurs, which is associated with the POVM elements $(1-\ptrash)\dyad{n}_R\otimes \hat{S}_{AB}^{a,0,Z,b}$ with $a,b\in\set{0,1}$ and $n\neq 1$, and the POVM elements $(1-\ptrash)\dyad{1}_R\otimes \hat{V}_{AB}^{a,b}$ with $a,b\in\set{0,1}$. Otherwise, it takes the value 0.
        \item[\small $\bullet$] $\rv{\chi}_{{\rm Q},l}^{u}$: it takes the value $\frac{c_{a,\mu,\beta,b}^l}{\pB{\beta}}$ if Alice and Bob observe a measurement outcome associated with the POVM element $(1-\ptrash)\dyad{n}_R\otimes \hat{S}_{AB}^{a,\mu,\beta,b}$, for any $n$ satisfying $(n,a,\mu,\beta,b)\notin\mathcal{K}$. Besides, it takes the value\footnote{The coefficients $c_{a,0,Z,b}^l$ are assumed to be identical for all $a,b\in\set{0,1}$.} $\frac{c_{0,0,Z,0}^l}{\pB{Z}}$ if $\rv{\chi}_{\rm key}^{u}=1$. Otherwise, it takes the value 0.
        \item[\small $\bullet$]  $\chiT{u}$: it takes the value $\omega$ if Alice observes a measurement outcome associated with the POVM element $\ptrash\dyad{\omega}_{RA}\otimes\identity_{B}$. Otherwise, it takes the value 0.
    \end{itemize}
    \item Alice and Bob compute 
    $\rv{M}_{\rm key,1}=\sum_u\rv{\chi}_{\rm key,1}^{u}$,
    $\rv{M}_{\rm ph}=\sum_u\rv{\chi}_{\rm ph}^{u}$,
    $\rv{M}_{\rm key}=\sum_u\rv{\chi}_{\rm key}^{u}$,
    $\rv{M}_{{\rm Q},l}=\sum_u\rv{\chi}_{{\rm Q},l}^{u}$ and $\rv{M}_{{\Tsind}}=\sum_u\chiT{u}$.
\end{enumerate}
\end{mybox}
\end{figure*}

After solving the dual SDP given by~\cref{eq:dualSDP2_decoy2}, we substitute the solutions $\eta_l^*$, $\lambda_{k,n}^*$ and $\lambda_{i,\infty}^*$ ($\forall l,k,i,n$) into the operator inequalities of the SDP and build a single operator inequality of the form
\begin{equation}
\begin{split}
\dyad{1}_{R}\otimes \hat{O}_{AB} \preceq & \sum_l\eta_l^*(\identity_{R}\otimes\hat{Q}_{AB}^{l})  + \Tobs{RA}\otimes\identity_B,
\end{split}
\end{equation}
where
\begin{equation}\label{eq:Tobs_RAB}
\begin{split}
    \Tobs{RA} =& 
    \sum_{n=0}^{n_{\rm cut}}\dyad{n}_{R}\otimes\Tobsn{A} +
    \Pi_R^{n>n_{\rm cut}}\otimes\sum_i\lambda_{i,\infty}^*\dyad{i}_A,
    \\
    = & \sum_{\omega} \omega\dyad{\omega}_{RA}.
\end{split}
\end{equation}
Here we have introduced the observables $\Tobsn{A}:=\sum_{k}\lambda_{k,n}^*\hat{T}_A^k=\sum_{\omega_n} \omega_n\dyad{\omega_n}$ (for $n\in\set{0,\dots,n_{\rm cut}}$), and the shortcut notation $\omega\in\bigcup_{n=0}^{n_{\rm cut}}\set{\omega_n}_{\omega_n}\cup\set{\lambda_{i,\infty}^*}_i$ and $\ket{\omega}\in\bigcup_{n=0}^{n_{\rm cut}}\set{\ket{n}_R\otimes\ket{\omega_n}_A}_{\omega_n} \cup \bigcup_{n>n_{\rm cut}}\set{\ket{n}_R\otimes\ket{i}_A}_i$.
Then, we can proceed similarly to~\cref{sec:SecurityAnalysisNew} and define an independent VEP for each target operator $\hat{O}_{AB}$ (see Box \ref{box:VEPdec}), namely the $E_{AB}^{\rm ph,(1)}$-VEP and the $D_{AB}^{(1)}$-VEP, which serve to derive the inequalities
\begin{equation}\label{eq:ineq_conditional_expect_Main_decoy}
\begin{split}
& \mathbb{E}\left[\rv{\chi}_{\rm ph}^{u}|\mathcal{F}_{u-1}\right]
 \leq 
 \sum_l\eta_l^*\mathbb{E}\left[\rv{\chi}_{{\rm Q},l}^{u}|\mathcal{F}_{u-1}\right] 
 + \frac{1-\ptrash}{\ptrash}\mathbb{E}\left[\chiT{u}|\mathcal{F}_{u-1}\right],
 \\
 & \mathbb{E}\left[\rv{\chi}_{\rm key,1}^{u}|\mathcal{F}_{u-1}\right]
 \geq
 - \sum_l\dot{\eta}_l^*\mathbb{E}\left[\rv{\chi}_{{\rm Q},l}^{u}|\mathcal{F}_{u-1}\right] 
 -  \frac{1-\ptrash}{\ptrash}\mathbb{E}\left[\chiT{u}|\mathcal{F}_{u-1}\right],
\end{split}
\end{equation}
respectively. Here and in what follows, we use the dotted notation $\dot{x}$ to indicate that the quantities are associated with the $D_{AB}^{(1)}$-VEP.

From \cref{eq:ineq_conditional_expect_Main_decoy}, one can follow an analogous procedure to that shown in Appendix \ref{app:finite-key} to demonstrate that, for $\epsproB,\epsproK>0$, the number of phase errors in the $E_{AB}^{\rm ph,(1)}$-VEP satisfies an inequality of the form
\begin{equation}\label{eq:Kfinal_decoy_ph}
\begin{split}
    \rv{M}_{\rm ph} 
    &\leq 
    \frac{N}{\sqrt{N}-2\tilde{a}_{\rm U}}\left[\frac{1}{\sqrt{N}}
    \left(
    \sum_l\eta_l^*K_{\epsilon_l,N,x_{\min}^l,x_{\max}^l}^{\text{sign}(\eta_l^*)}(\rv{M}_{{\rm Q},l})
    + 
    \frac{1-\ptrash}{\ptrash} K_{\epsproK,N,\omega_{\min},\omega_{\max}}^{1}(M_{\Tsind}^{\rm U})
    \right)
    +\tilde{b}_{\rm U}-\tilde{a}_{\rm U}\right]
    \\
    &=: M_{\rm ph}^{\rm U}\left(\rv{M}_{{\rm Q},l}\right),
\end{split}
\end{equation}
except with probability less than or equal to $3\epsproK+\epsproB$, where the quantities $\tilde{a}_{\rm U}$ and $\tilde{b}_{\rm U}$ are computed prior to the protocol execution according to~\cref{eq:optimal_ab_tilde}; $x_{\min}^l$ and $x_{\max}^l$ ($\omega_{\min}$ and $\omega_{\max}$) denote the minimum and maximum possible outcomes of the RVs $\rv{\chi}_{{\rm Q},l}^u$ ($\chiT{u}$), respectively; the quantities $\epsilon_l$ represent some arbitrary failure probabilities satisfying $\sum_{l} \epsilon_l=\epsproK$; the functions $K_{\epsilon,N,x_{\min},x_{\max}}^{\pm 1}$ are defined in~\cref{eq:Kfunctions_unnormalized}; and $\text{sign}(\cdot)$ denotes the sign function. We refer the reader to Appendix~\ref{app:finite-key} for further details.
Similarly, the number of $Z$-basis single-photon detections with $\mu=0$ in the $D_{AB}^{(1)}$-VEP satisfies an inequality of the form
\begin{equation}\label{eq:Kfinal_decoy_key}
\begin{split}
    \rv{M}_{\rm key,1} 
    &\geq 
    \frac{N}{\sqrt{N}+2\tilde{a}_{\rm L}}
    \left[\frac{-1}{\sqrt{N}}
    \left(
    \sum_l\dot{\eta}_l^*K_{\epsilon_l,N,x_{\min}^l,x_{\max}^l}^{\text{sign}(\dot{\eta}_l^*)}(\rv{M}_{{\rm Q},l})
    + 
    \frac{1-\ptrash}{\ptrash}
    K_{\epsproK,N,\omega_{\min},\omega_{\max}}^{1}(\dot{M}_{\Tsind}^{\rm U})
    \right)
    -\tilde{b}_{\rm L}+\tilde{a}_{\rm L}\right]
    \\
    &=: M_{\rm key,1}^{\rm L}\left(\rv{M}_{{\rm Q},l}\right),
\end{split}
\end{equation}
except with probability less than or equal to $3\epsproK+\epsproB$. 
Here we consider that the quantities $\rv{M}_{{\rm Q},l}$ in~\cref{eq:Kfinal_decoy_ph} are constructed from the statistics of rounds in which Bob selects the $X$ basis, while the quantities $\rv{M}_{{\rm Q},l}$ in~\cref{eq:Kfinal_decoy_key} are constructed from the detection statistics of rounds in which Alice and Bob select the $Z$ basis (see the specific example below).

To obtain the bound $M_{\Tsind}^{\rm U}$---the bound $\dot{M}_{\Tsind}^{\rm U}$ is obtained similarly---we apply~\cref{thm:correlations_decoy} (see Appendix~\ref{app:correlated_sources}), which is a generalization of~\cref{thm:correlations} to the decoy-state setting with bit/basis and intensity correlations. This theorem requires the use of SDPs to compute the worst-case expectation of the sums of the RVs $\chi_{\Tsind}^{u}$, and subsequently take the worst-case bound given all possible setting sequences $j^{u+L}_{u+1}j^{u-1}_{u-L}$.

Finally, unlike in Appendix~\ref{app:finite-key}, here we simplify the analysis by disregarding any key that could be generated from the \LC{} rounds. To this end, we can consider that, in the VP, the users also tag each round as \LC{} with probability $\alpha$---as they do in the two VEPs---and generate key only from the subset of untagged single-photon $Z$-basis rounds with $\mu=0$. Thus, we can use~\cref{eq:Kfinal_decoy_key,eq:Kfinal_decoy_ph} directly to establish probabilistic bounds on the number of $Z$-basis single-photon detections with $\mu=0$ and the number of single-photon phase errors of the virtual protocol. These bounds have the form 
\begin{equation}\label{eq:probabilistic_statement_decoy}
\begin{split}
    &\Pr[\rv{M}_{\rm key,1}< M_{\rm key,1}^{\rm L}(\vos{\rv{M}}_{Z})] \leq \frac{\epsilon_{\rm pro}}{2},
    \\
    &\Pr[\rv{e}_{\rm ph,1} > \frac{M_{\rm ph}^{\rm U}(\vos{\rv{M}}_{X})}{\rv{M}_{\rm key,1}}] \leq \frac{\epsilon_{\rm pro}}{2},
\end{split}
\end{equation}
where $\vos{\rv{M}}_{Z}:=\set{\rv{M}_{\rm key},\set{\rv{M}_{\rm test}^{\mu,a,Z,b}}_{\mu,a,b}}$, $\vos{\rv{M}}_{X}:=\set{\rv{M}_{\rm test}^{\mu,a,X,b}}_{\mu,a,b}$,  $\rv{e}_{\rm ph,1}$ is the single-photon phase-error rate of the VEP and $\epsilon_{\rm pro}/2=\epsproB+3\epsproK$. Importantly, similarly to Appendix~\ref{app:finite-key}, the previous bounds can be readily recast in terms of the observed quantities $\rv{N}_{\rm key}$ and $\set{\rv{N}_{\rm test}^{\mu,a,\beta,b}}_{\mu,a,\beta,b}$ by applying standard concentration bounds for sums of independent random variables. 

Using arguments analogous to those for the single-photon BB84 scenario (see~\cref{rem:remark_detectors}), we have that Bob's POVM elements do not need to be characterized---we can simply assume that they are orthogonal projectors in the SDPs---if both the SDP associated with the $E_{AB}^{\rm ph,(1)}$-VEP and that associated with the $D_{AB}^{(1)}$-VEP rely only on statistics from a single basis of Bob---which is often the case, as shown in the example below. Importantly, this allows the security of the decoy-state scenario to be extended to setups with partially uncharacterized receivers suffering from basis detection efficiency mismatch. For this, we apply~\cite[Corollary 2]{curras2025security}, which guarantees that, given that~\cref{eq:probabilistic_statement_decoy} holds for ideal receivers, then 
\begin{equation}\label{eq:detection_eff_mismatch_DS}
\begin{split}
    &
    \Pr[\left(\rv{M}_{\rm key,1}< M_{\rm key,1}^{\rm L}(\vos{\rv{M}}_{Z}) \right)
    \bigcup 
    \left(\rv{e}_{\rm ph,1} > \frac{\frac{M_{\rm ph}^{\rm U}(\vos{\rv{M}}_{X})}{M_{\rm key,1}^{\rm L}(\vos{\rv{M}}_{Z})} +  \delta_1 + \gamma_{ \rm bin}^{\epsilon_{\rm dep-1}}(M_{\rm key,1}^{\rm L}(\vos{\rv{M}}_{Z}),\delta_1)}{1 - \delta_2 - \gamma_{ \rm bin}^{\epsilon_{\rm dep-2}}(M_{\rm key,1}^{\rm L}(\vos{\rv{M}}_{Z}),\delta_2)}\right)]
    \leq 
    \epsilon_{\rm pro}+\epsilon_{\rm dep-1}^2 + \epsilon_{\rm dep-2}^2,
    \end{split}
\end{equation}
holds under detection-efficiency mismatch parameterized by $\delta_1$ and $\delta_2$ (see \cref{app:finite-key} for further details).

Given~\cref{eq:probabilistic_statement_decoy}---a similar result can be readily derived from~\cref{eq:detection_eff_mismatch_DS} for a scenario with detection efficiency mismatch---it follows via the entropic uncertainty relation~\cite{tomamichel2017largely}---though a similar result can be obtained via the phase-error correction approach~\cite{koashi2009simple}---that the decoy-state protocol is $\epsilon_{\rm tot}$-secure with a secret-key length~\cite{tupkary2024phase}
\begin{equation}
\label{eq:key_length_decoy}
\begin{split}
    \rv{l}_{\rm key} =&  M_{\rm key,1}^{\rm L}(\vos{\rv{M}}_{Z})\left[1-h\left(\frac{M_{\rm ph}^{\rm U}(\vos{\rv{M}}_{X})}{ M_{\rm key,1}^{\rm L}(\vos{\rv{M}}_{Z})}\right)\right]
    -\rv{\lambda}_{\rm EC}-2\log(1/2\epsilon_{\rm PA})-\log(2/\epsilon_{\rm EV}),
\end{split}
\end{equation}
where $\epsilon_{\rm tot}=\epsilon_{\rm EV}+\epsilon_{\rm sec}$, and $\epsilon_{\rm sec}=2\sqrt{\epsilon_{\rm pro}}+\epsilon_{\rm PA}$ (see~\cref{sec:SecurityAnalysisNew} for the definition of the parameters).

\subsubsection*{Example: decoy-state BB84 scheme with three intensity settings}
Here, we evaluate the secret-key rate of a realistic decoy-state BB84 setup in the presence of SPFs, bit/basis and intensity leakage and correlations, and basis-dependent detection-efficiency mismatch. For this, we consider a typical polarization-based scheme with three intensity settings $\mu\in\set{0,1,2}$, with $\Iav{0}>\Iav{1}>\Iav{2}$, where $\Iav{\mu}$ denotes the average intensity given that Alice selects $\mu$ (see \cref{app:model_correlations}). In our simulations, we further impose the convention that Alice sets $a=0$ whenever she selects the weakest intensity setting, $\mu=2$. This intensity setting typically represents the vacuum decoy, for which the particular encoding choice is essentially irrelevant. This convention reduces the number of setting pairs $(a,\mu)$ to nine; we again label these pairs by a single index $i\in\set{0,1,\dots,8}$. Importantly, this simulation convention is not required by the security analysis, which remains valid for standard implementations where the bit/basis and intensity settings are selected independently. In this latter scenario, the number of setting pairs $(a,\mu)$ would be simply larger. Besides, we consider that $\Iav{2}=10^{-5}$ to model the finite extinction ratio of the intensity modulator.

Furthermore, we assume that the polarization of the transmitted states deviates according to the same model used in the BB84 scenario. That is, Alice emits phase-randomized weak coherent pulses with polarization $\kappa_a\theta_a$, where the bit/basis settings $a\in\set{0,1,2,3}$ correspond to the ideal polarizations $(\theta_0,\theta_1,\theta_2,\theta_3)=(0,\frac{\pi}{2},\frac{\pi}{4},\frac{3\pi}{4})$, and the coefficients $\kappa_a=(1+\Delta\theta/\pi)$ quantify, for a given $\Delta\theta$, the characterized deviations due to SPFs~\cite{pereira2020quantum,curras2025numerical,curras2023security2}. In the absence of side channels, the states generated by Alice in this model can be described as mixtures of Fock states $\ket{\phi_{n,a}}_C$ of the form
\begin{equation}\label{eq:ideal_states_DS} 
\ket{\phi_{n,a}}_C = 
\frac{1}{\sqrt{n!}} \left( \cos(\kappa_a\theta_a)\, \hat{a}_H^\dagger + \sin(\kappa_a\theta_a)\, \hat{a}_V^\dagger \right)^n \ket{ 0}_C 
= 
\sum_{k=0}^{n} \sqrt{\binom{n}{k}} \cos(\kappa_a\theta_a)^k \sin(\kappa_a\theta_a)^{\,n-k} \ket{k}_{C_H}\ket{n-k}_{C_V},
\end{equation} 
where $\hat{a}_H^{\dagger}$ ($\hat{a}_V^{\dagger}$) is the horizontal (vertical) creation operator associated with the horizontal polarization mode $C_H$ (vertical polarization mode $C_V$) of the transmitted system $C$, and $\ket{n}$ denotes an $n$-photon Fock state. That is, the states in~\cref{eq:ideal_states_DS}, which we use as the reference states for $n\in\set{0,1,\dots,n_{\rm cut}}$, already incorporate SPFs. Note that these reference states do not depend on the intensity setting. We remark, however, that the choice of reference states does not affect security, since they can be selected freely.
In fact, in this example we shall assume that the actual emitted Fock states do depend on both the bit/basis and intensity settings and, moreover, exhibit both bit/basis and intensity correlations. That is, the $n$-photon Fock states emitted by Alice have the general form $\sigma_{C_u}^{n_u,i^u_{u-L}}$, and we require them to satisfy~\cref{eq:assumption1thm_decoy}.

As for the emission probabilities $p_{n_u|i^u_{u-L}}$, we shall consider, for this example, the model for intensity correlations introduced in Appendix~\ref{app:model_correlations}. In this model, the Fock-state statistics in a given round follow a Poisson distribution fully determined by the current intensity $I_{u}(\mu_{u-L}^{u})$, which in turn depends on the current and on the $L$ previous intensity settings. That is, the Fock-state statistics are given by $p_{n_u|I_{u}}=e^{-I_{u}}I_{u}^{n_u}/n_u!$, where we use the shorthand notation $I_u$ to denote the current intensity. This corresponds to a typical scenario in practical high-speed QKD implementations. In the simulations we disregard, for simplicity, cross correlations, i.e., the fact that the bit/basis encoding settings $a^u_{u-L}$ may also influence the intensity of the transmitted signals. We remark, however, that cross correlations can be incorporated straightforwardly in the analysis. Moreover, one could consider more general scenarios in which the actual intensity $I_{u}$ is not a deterministic function of the setting choices, but is instead drawn from a given probability density function determined by those choices, generally leading to non-Poissonian statistics $p_{n_u|i^u_{u-L}}$.

Now, to compute the dual SDPs in~\cref{eq:dualSDP2_decoy2} for the single-photon $Z$-basis detections, we consider operators $\hat{Q}_{AB}^{l}$ (see~\cref{eq:Q_and_q_DS}) with coefficients such that 
\begin{equation}\label{eq:coefs_c_Z_DS}
c^{l}_{a,\mu,\beta,b} = 
\begin{cases}
1 &\text{ if } a\in\set{0,1}, \mu=l, b\neq\perp, \beta=Z, \\
0 &\text{ otherwise, }
\end{cases}
\end{equation}
where $l\in\set{0,1,2}$. Note that, in the scenario considered in this example, $a=1$ and $\mu=2$ can never occur simultaneously, but this does not affect the above definition of the coefficients. That is, with~\cref{eq:coefs_c_Z_DS} we basically incorporate into the SDP the information of the three observed gains in the $Z$ basis, one per intensity setting. On the other hand, for the dual SDP associated with the single-photon phase errors we set
\begin{equation}
c^{l}_{a,\mu,\beta,b} = 
\begin{cases}
1 &\text{ if } a=a', \mu=\mu', b=b', \beta=X, \\
0 &\text{ otherwise, }
\end{cases}
\end{equation}
where now each $l\in\set{0,1,\dots,23}$ is associated with a different combination $(a',\mu',b')$, with $a'\in\set{0,1,2,3}$, $\mu'\in\set{0,1,2}$ and $b'\in\set{0,1}$. 

For the operators $\hat{T}^{k}_{A}\equiv\hat{T}^{i,j}_{A}$, we consider those defined in~\cref{eq:Tk_BB84}, where $i,j\in\set{0,1,\dots,8}$. As for the guesses $t_{k,n}^{\rm gs}\equiv t_{i,j,n}^{\rm gs}$, we compute $t_{i,j,n}^{\rm gs}=\Tr{\hat{T}^{i,j}_A\rho_A^{n,\epsleak'}}$ for the marginal state
\begin{equation}
    \rho_{A}^{n,\varepsilon'} = \sum_{i,j}\sqrt{p_{n,i}p_{n,j}}\left[(1-\varepsilon')\braket{\phi_{n,j}}{\phi_{n,i}} + \varepsilon'\braket*{\phi_{n,j}^{\perp}}{\phi_{n,i}^{\perp}} + \sqrt{(1-\varepsilon')\varepsilon'} (\braket*{\phi_{n,j}}{\phi_{n,i}^{\perp}} + \braket*{\phi_{n,j}^{\perp}}{\phi_{n,i}})\right] \ketbra{i}{j}_A,
\end{equation}
where $\varepsilon'=1-(1-\epsleak)^{L+1}$, $p_{n,i}=\piA{i}p_{n|i}$, and $p_{n|i}$ denotes the probability of emitting an $n$-photon Fock state given that the average intensity is $\Iav{\mu^{(i)}}$, with $\mu^{(i)}$ being the intensity setting associated with the index $i$. Similarly to the BB84 scenario, since the matrix $\rho_{A}^{n,\epsleak'}$ depends on the unknown states $\ket*{\phi^{\perp}_{n,i}}$, here we shall conservatively consider that $\braket*{\phi^{\perp}_{n,i}}{\phi^{\perp}_{n,j}}=\braket*{\phi^{\perp}_{n,i}}{\phi_{n,j}}=0$ for $i\neq j$ and for $n\in\set{0,1,\dots,n_{\rm cut}}$. Finally, we set $t_{i,\infty}^{\rm gs} = \sum_{n=n_{\rm cut}+1}^{\infty}\Tr{\dyad{i}_{A} \rho_{A}^{n,\varepsilon'}}=\piA{i}\left(1-\sum_{n=0}^{n_{\rm cut}}p_{n|i}\right)$.

With this, we can solve the corresponding SDPs and obtain~\cref{eq:Kfinal_decoy_ph,eq:Kfinal_decoy_key}. Subsequently, given the model of correlations, we compute the bounds $M_{\Tsind}^{\rm U}$ and $\dot{M}_{\Tsind}^{\rm U}$ required in these two equations via~\cref{thm:correlations_decoy}. Specifically, we solve the SDPs in~\cref{thm:correlations_decoy} for each possible intensity-setting context $\mathcal{L}_u=\mu^{u+L}_{u+1}\mu^{u-1}_{u-L}$ and take the worst-case bound, as prescribed by the theorem. Regarding the quantities $\xi_{i,j,\mathcal{L}_u}$ also required in~\cref{thm:correlations_decoy}, we note that, since the Fock-state statistics follow a Poisson distribution in the model considered, they can be computed analytically for each $\mathcal{L}_u$ as
$$\xi_{i,j,\mathcal{L}_u}=\prod_{k=1}^{L}e^{-\frac{1}{2}\left(\sqrt{I_{u+k}(j^{u+k}_{u+1} i j^{u-1}_{u-L+k})}-\sqrt{I_{u+k}(j^{u+k}_{u+1} j j^{u-1}_{u-L+k})}\right)^2}.$$

\subsubsection*{Simulations}
The simulations for this example are shown in~\cref{fig:SKR_decoy_main} (we refer the reader to~\cref{sec:Results} in the main text for further details). To obtain them, we compute the expected statistics $ p_{b|a,\mu,\beta}$ by using a typical channel model. In particular, for the channel model considered, the statistics $p_{b|a,\mu,Z}$ are given by
\begin{equation}
\begin{split}
    p_{\perp|a,\mu,Z}=&(1-p_d)^2e^{-\eta\Iav{\mu}},\\
    p_{0|a,\mu,Z}=&\left(1-(1-p_d)e^{-\eta \Iav{\mu}\cos^2(\kappa_a\theta_a)}\right)
    \left[(1-p_d)e^{-\eta\Iav{\mu}\sin^2(\kappa_a\theta_a)} + \frac{1}{2}\left(1-(1-p_d)e^{-\eta\Iav{\mu}\sin^2(\kappa_a\theta_a)}\right)\right],\\
    p_{1|a,\mu,Z}=&\left(1-(1-p_d)e^{-\eta\Iav{\mu}\sin^2(\kappa_a\theta_a)}\right)
    \left[(1-p_d)e^{-\eta\Iav{\mu}\cos^2(\kappa_a\theta_a)} + \frac{1}{2}\left(1-(1-p_d)e^{-\eta\Iav{\mu}\cos^2(\kappa_a\theta_a)}\right)\right],\\
\end{split}
\end{equation}
where $\eta$ is the overall transmittance of the system, including the channel loss and the efficiency of Bob's detectors $\eta_{\rm det}$. The statistics $p_{b|a,\mu,X}$ associated with Bob's $X$-basis measurement can be directly obtained from the previous expressions by substituting $\kappa_a\theta_a$ with $\kappa_a\theta_a-\pi/4$. 

In the simulations, we use \texttt{patternsearch} and/or \texttt{surrogateopt} routines from Matlab's Global Optimization Toolbox~\cite{mathworks_global_optimization_toolbox_2025} to optimize, for each distance, the average values of the two higher intensities, the probabilities of selecting each intensity setting, the parameter $\alpha$, the basis selection probability $p_Z$---for simplicity, we select the same probability for Alice and Bob---and the parameter $\gamma_{\Tsind}$. The weakest intensity is fixed to $10^{-5}$ to account for the finite extinction ratio of practical intensity modulators.

\subsection{Decoy-state QKD protocols with imperfect phase randomization}\label{app:imperfect_PR}

Here we show that the decoy-state analysis developed above also applies to sources that are not perfectly phase randomized, i.e., sources whose emitted states are not Poissonian mixtures of Fock states. We illustrate this with two examples: discrete phase randomization~\cite{caoDiscretephaserandomizedCoherent2015}, and the model of continuous but imperfect phase randomization introduced in Ref.~\cite{naharImperfectPhase2023}. For simplicity, we consider uncorrelated sources ($L=0$) here and omit the round index $u$. 

\subsubsection*{General principle}

Although for simplicity Box~\ref{box:APdec} describes Alice's source as emitting $n$-photon states, the security analysis never uses the fact that $n$ is a photon number. It only requires that the state emitted for each setting $i$ admits a decomposition
\begin{equation}\label{eq:PR_general_decomposition}
	\sigma_{C}^{i} = \sum_{n} p_{n|i}\, \sigma_{C}^{n,i},
\end{equation}
with known weights $p_{n|i}$, such that the components with $n\le n_{\rm cut}$ satisfy the fidelity constraint in~\cref{eq:assumption1thm_decoy} with respect to known reference states $\ket{\phi_{n,i}}$, while those with $n>n_{\rm cut}$ are arbitrary. To see why this is sufficient, recall that the entanglement-based state in~\cref{eq:global_correlated_state_decoy_main} reproduces the emitted states for any such decomposition, with the label $n$ stored in Alice's shield system $R$. Since all the operators entering the SDPs and the VEPs are block diagonal with respect to $R$, the emitted states only enter the analysis through Alice's marginal states associated with each label $n$, which depend only on the inner products between components with the same label. Hence, the source enters the analysis only through the weights $p_{n|i}$, the fidelity parameters in~\cref{eq:assumption1thm_decoy}, and the inner products $\braket{\phi_{n,i}}{\phi_{n,j}}$ between reference states with the same label $n$. Consequently, components with different labels need not be orthogonal, and the dimension of the emitted states does not affect the analysis. The reference states may also depend on the intensity setting, in which case the resulting deviation from the ideal decoy-state condition is automatically accounted for through these inner products. 

Moreover, suppose that the emitted states can be written as $\sigma_C^{i}=\mathcal{G}\big(\sigma_C^{{\rm mod},i}\big)$, where $\mathcal{G}$ is a CPTP map that does not depend on Alice's setting choice. Since such a map can be absorbed into Eve's attack (see~\cite[Lemma~8]{nahar2024postselection}), as in the proof of~\cref{lemma:density_op_correlations}, it suffices to prove security when Alice emits the states $\sigma_C^{{\rm mod},i}$. Hence, it is enough that the states $\sigma_C^{{\rm mod},i}$, rather than the emitted states themselves, admit a decomposition of the form in~\cref{eq:PR_general_decomposition}. Below, we give an example of each case.

We use the following notation throughout both examples. Let $\ket{n_a}$ denote the state with $n$ photons associated with the encoding setting $a$. In the polarization-encoding example above, this is the state denoted by $\ket{\phi_{n,a}}$ in~\cref{eq:ideal_states_DS} (we use a different symbol here because, in the examples below, the reference states $\ket{\phi_{n,i}}$ are no longer Fock states). For a setting $i=(a,\mu)$, let $I_i$ be the mean photon number and $\ket{\nu_i}:=\sum_n\sqrt{P_n(I_i)}\ket{n_a}$ the corresponding coherent state, where $P_n(I):=e^{-I}I^n/n!$. A phase rotation by $\theta$ gives $e^{{\rm i}\theta\hat N}\ket{\nu_i}$, where $\hat N$ is the total photon-number operator.

In both examples, the components with $n\le n_{\rm cut}$ that we obtain are pure and known exactly, so we can simply choose them as the reference states $\ket{\phi_{n,i}}$ themselves, in which case~\cref{eq:assumption1thm_decoy} holds with $\epsleak=0$. State-preparation flaws are already included through the optical modes, and additional imperfections, such as side channels, can be included by allowing $\epsleak>0$, provided that suitable fidelity bounds are available for these components.

\subsubsection*{Example 1: discrete phase randomization}

In discrete phase randomization~\cite{caoDiscretephaserandomizedCoherent2015}, Alice selects the phase of each pulse uniformly at random from the set $\set{2\pi m/J}_{m=0}^{J-1}$, for some integer $J\ge2$. Averaging $\dyad{\nu_i}$ over these $J$ phases results in the emitted state
\begin{equation}
	\sigma_C^{i}=\frac{1}{J}\sum_{m=0}^{J-1}e^{{\rm i}2\pi m\hat N/J}\dyad{\nu_i}e^{-{\rm i}2\pi m\hat N/J}
	=\sum_{k=0}^{J-1}P_k^{J}(I_i)\dyad{\Upsilon_{k,i}},
\end{equation}
where
\begin{equation}
	\ket{\Upsilon_{k,i}}:=\frac{1}{\sqrt{P_k^{J}(I_i)}}\sum_{l=0}^{\infty}\sqrt{P_{lJ+k}(I_i)}\,\ket{(lJ+k)_a},
	\qquad
	P_k^{J}(I):=\sum_{l=0}^{\infty}P_{lJ+k}(I),
\end{equation}
are the so-called pseudo-Fock states. This is directly of the form in~\cref{eq:PR_general_decomposition}, with $p_{n|i}=P_n^{J}(I_i)$ and $\ket{\phi_{n,i}}=\ket{\Upsilon_{n,i}}$ for $n\le n_{\rm cut}\le J-1$. Unlike Fock states, the pseudo-Fock states depend on the intensity setting, but as discussed above, this is automatically accounted for through the inner products
\begin{equation}\label{eq:overlap_pseudoFock}
	\braket{\Upsilon_{n,i}}{\Upsilon_{n,j}}=\frac{e^{-(I_i+I_j)/2}}{\sqrt{P_n^{J}(I_i)\,P_n^{J}(I_j)}}\sum_{l=0}^{\infty}\frac{x_{ij}^{lJ+n}}{(lJ+n)!},
\end{equation}
where $x_{ij}:=\braket{1_a}{1_{a'}}\sqrt{I_iI_j}$ for $i=(a,\mu)$ and $j=(a',\mu')$, and the series converges rapidly. In the limit $J\to\infty$, one recovers $\ket{\Upsilon_{n,i}}=\ket{n_a}$ and $P_n^{J}(I_i)=P_n(I_i)$, i.e., the continuously phase-randomized case.

\subsubsection*{Example 2: imperfect continuous phase randomization}

Following Ref.~\cite{naharImperfectPhase2023}, assume now that the phase of the pulse in each round is drawn, independently of Alice's setting choice and of the other rounds, from an unknown distribution whose probability density function (PDF) $h(\theta)$ is lower bounded by $\delta_{\rm ph}/(2\pi)$, for some known $\delta_{\rm ph}\in[0,1]$ (this parameter is denoted by $q$ in Ref.~\cite{naharImperfectPhase2023}). We can write $h(\theta)=\delta_{\rm ph}/(2\pi)+(1-\delta_{\rm ph})g(\theta)$, where $g(\theta)$ is an unknown probability density function. The emitted states can then be expressed as
\begin{equation}
	\sigma_C^{i} =  \mathcal{G}\big(\sigma_C^{{\rm mod},i}\big),
\end{equation}
where $\mathcal{G}(\rho):=\int g(\theta)\,e^{{\rm i}\theta\hat N}\rho\,e^{-{\rm i}\theta\hat N}d\theta$ is a setting-independent random phase rotation, and
\begin{equation}
	\sigma_C^{{\rm mod},i} := \delta_{\rm ph}\sum_{n=0}^{\infty}P_n(I_i)\dyad{n_a} + (1-\delta_{\rm ph})\dyad{\nu_i}.
\end{equation}
Since $\mathcal{G}$ does not depend on $i$, it suffices to prove security assuming that Alice directly emits the states $\sigma_C^{{\rm mod},i}$. To obtain the required decomposition of the form in~\cref{eq:PR_general_decomposition}, we fix an integer $n_{\rm fr}\ge n_{\rm cut}$ and write $\sigma_C^{{\rm mod},i}=\sigma_C^{{\rm fr},i}+\sigma_C^{{\rm tail},i}$, where
\begin{equation}
	\sigma_C^{{\rm fr},i}:=\delta_{\rm ph}\sum_{n=0}^{n_{\rm fr}}P_n(I_i)\dyad{n_a}+(1-\delta_{\rm ph})\dyad{\nu_i},
	\qquad
	\sigma_C^{{\rm tail},i}:=\delta_{\rm ph}\sum_{n>n_{\rm fr}}P_n(I_i)\dyad{n_a}.
\end{equation}
Note that both terms are positive semidefinite, and the first has finite rank since $\ket{\nu_i}$ lies in the span of
$\{\ket{0_a},\dots,\ket{(n_{\rm fr})_a},$
\allowbreak
$\ket{\nu^{\rm tail}_i}\}$, where
\begin{equation}
	\ket{\nu^{\rm tail}_i}:=\frac{1}{\sqrt{P_{>n_{\rm fr}}(I_i)}}\sum_{n>n_{\rm fr}}\sqrt{P_n(I_i)}\ket{n_a},
	\qquad
	P_{>n_{\rm fr}}(I):=1-\sum_{n=0}^{n_{\rm fr}}P_n(I).
\end{equation}
Therefore, $\sigma_C^{{\rm fr},i}$ is supported on a subspace of dimension at most $n_{\rm fr}+2$, and its spectral decomposition $\sigma_C^{{\rm fr},i}=\sum_{m=0}^{n_{\rm fr}+1}w_{m|i}\dyad{f_{m,i}}$ can be computed exactly by diagonalizing a real symmetric $(n_{\rm fr}+2)\times(n_{\rm fr}+2)$ matrix. This gives the exact decomposition
\begin{equation}\label{eq:PR_exact_decomposition}
	\sigma_C^{{\rm mod},i}=\sum_{m=0}^{n_{\rm fr}+1}w_{m|i}\dyad{f_{m,i}}+\sigma_C^{{\rm tail},i}.
\end{equation}

We can now set $p_{n|i}=w_{n|i}$ and $\ket{\phi_{n,i}}=\ket{f_{n,i}}$ for $n\le n_{\rm cut}$, and assign the remaining eigenvectors, together with $\sigma_C^{{\rm tail},i}$, to the block $n>n_{\rm cut}$. Note that $\sigma_C^{{\rm tail},i}$ is not orthogonal to the eigenvectors $\ket{f_{m,i}}$. As discussed above, however, this is not an issue, since components with different labels need not be orthogonal. Finally, the required overlaps $\braket{f_{n,i}}{f_{n,j}}$ can be readily computed from the eigenvectors $\ket{f_{n,i}}$ and the basis states introduced above. The analysis above then applies without modification. 
For $\delta_{\rm ph}=1$, $\sigma_C^{{\rm fr},i}$ is diagonal in the Fock basis, so one can take $\ket{f_{n,i}}=\ket{n_a}$ and $w_{n|i}=P_n(I_i)$ for $n\le n_{\rm fr}$, recovering the perfectly phase-randomized case.

\section{Finite-key bounds}\label{app:finite-key}
\subsection{Auxiliary results on concentration bounds}
The finite-key analysis requires the use of certain results on concentration bounds for sums of RVs that we include below for completeness:
\begin{lemma}[Bernstein's inequality~\cite{bernstein1924modification,boucheron2013concentration}]\label{thm:bernstein}
Let \(\rv{X}_1,\dots,\rv{X}_n\) be independent, zero-mean random variables satisfying \(\lvert \rv{X}_i\rvert \le c\) almost surely. Define $\rv{s}=\frac{1}{n}\sum_{u=1}^n \rv{X}_u$ and $v = \frac{1}{n}\sum_{u=1}^n \mathrm{Var}(\rv{X}_u)$. Then for any \(\Delta_{\rm B} > 0\),
\begin{equation}
\Pr\bigl(\rv{s} \ge \Delta_{\rm B}\bigr)
\;\le\;
\exp\!\Bigl(-\frac{n\Delta_{\rm B}^2}{2v + \tfrac{2}{3}\,c\,\Delta_{\rm B}}\Bigr). 
\end{equation}
Naturally, if the random variables have non-zero mean, but satisfy \(\lvert \rv{X}_i-\mathbb{E}[\rv{X}_i]\rvert \le c\) almost surely, then
\begin{equation}
  \Pr\bigl(\rv{s} \ge \mathbb{E}[\rv{s}]+ \Delta_{\rm B}\bigr)
  \;\le\;
  \exp\!\Bigl(-\frac{n\Delta_{\rm B}^2}{2v + \tfrac{2}{3}\,c\,\Delta_{\rm B}}\Bigr)=:\epsproB,
\end{equation}
where the deviation term $\Delta_{\rm B}$ can be written as a function of the maximum failure probability $\epsproB$ of the bound as
\begin{equation}\label{eq:DeltaB}
\Delta_{\rm B}(n,v,c,\epsproB) = \frac{c \ln(1/\epsproB)}{3n} + \sqrt{ \left(\frac{c \ln(1/\epsproB)}{3n}\right)^2  + \frac{2 v \ln(1/\epsproB)}{n} }.
\end{equation}
Moreover, we also have the converse bound
\begin{equation}
  \Pr\bigl(\rv{s} \le \mathbb{E}[\rv{s}] - \Delta_{\rm B}\bigr)
  \;\le\;\epsproB.
\end{equation}
\end{lemma}
\begin{lemma}[Serfling's inequality~\cite{serfling1974probability}]\label{thm:serfling}
Let $\rv{X}_1, \ldots, \rv{X}_{r+s}$ be bit-valued random variables. Let $\mathcal{S}_m$ be a uniformly random subset of $s$ positions out of the $r+s$ positions. Then, for any $\epsilon>0$,
\begin{equation}
\begin{split}
    \Pr\left[ 
        \frac{1}{r}\sum\limits_{u \notin \mathcal{S}_m}\rv{X}_u 
        \geq 
        \frac{1}{s}\sum\limits_{u \in \mathcal{S}_m}\rv{X}_u 
        + \Delta_{\rm S}(r,s,\epsilon) 
    \right] 
    &\leq \epsilon,
    \\
    \Pr\left[ 
        \frac{1}{r}\sum\limits_{u \notin \mathcal{S}_m}\rv{X}_u 
        \leq 
        \frac{1}{s}\sum\limits_{u \in \mathcal{S}_m}\rv{X}_u 
        - \Delta_{\rm S}(r,s,\epsilon) 
    \right] 
    &\leq \epsilon,
\end{split}
\end{equation}
where
\begin{equation}\label{eq:Delta_S}
   \Delta_{\rm S}(r,s,\epsilon) := \sqrt{\frac{(r+s)(s+1)\ln(1/\epsilon^2)}{r s^2}}.
\end{equation}
\end{lemma}

\begin{lemma}[Kato's inequality~\cite{kato}]\label{lemma:kato}
Let \(\{\rv{X}_u\}_{u=1}^N\) be a sequence of random variables, and let \(\{\mathcal F_u\}_{u=0}^N\) be a filtration such that \(\rv{X}_1,\dots,\rv{X}_u\) are \(\mathcal F_u\)-measurable for each \(u\); that is, \(\mathcal F_u\subseteq \mathcal F_{u+1}\) and \(\mathbb{E}[\rv{X}_{u'}\mid \mathcal F_u]=\rv{X}_{u'}\) for all \(u\ge u'\).
Define
\begin{equation}
\rv{S}_k:=\sum_{u=1}^k \rv{X}_u,
\qquad
\rv{E}_k^{\mathcal F}:=\sum_{u=1}^k \mathbb{E}[\rv{X}_u\mid \mathcal F_{u-1}].\nonumber
\end{equation}
Suppose that \(0\le \rv{X}_u\le 1\) for all \(u\). Then, for any \(N\in\mathbb N\), \(a\in\mathbb R\), and \(b\in\mathbb R_{\ge 0}\),
\begin{equation}\label{eq:Kato_original_U}
\Pr\!\left[
\rv{E}_{N}^{\mathcal{F}} \geq \rv{S}_N
+ \left(b+a\left(\frac{2\rv{S}_N}{N}-1\right)\right)\sqrt{N}
\right]
\leq
e^{-\frac{2(b^2-a^2)}{\left(1+\frac{4a}{3\sqrt{N}}\right)^2}},
\end{equation}
and
\begin{equation}\label{eq:Kato_original_L}
\Pr\!\left[
\rv{E}_{N}^{\mathcal{F}} \leq \rv{S}_N
- \left(b+a\left(\frac{2\rv{S}_N}{N}-1\right)\right)\sqrt{N}
\right]
\leq
e^{-\frac{2(b^2-a^2)}{\left(1-\frac{4a}{3\sqrt{N}}\right)^2}}.
\end{equation}
\end{lemma}

Since the parameters $a$ and $b$ introduced in~\cref{lemma:kato} can be freely chosen, the bounds provided in~\cref{eq:Kato_original_U,eq:Kato_original_L} can be fine-tuned prior to the protocol execution based on a prediction of the random variable $\rv{S}_N$. In particular, to get tight upper and lower bounds on the sum of conditional probabilities, one can calculate the values of $a$ and $b$ that minimize the deviation term in~\cref{eq:Kato_original_U,eq:Kato_original_L}, respectively, for a prediction $S_N^{\rm gs}$ of $\rv{S}_N$ and for a fixed failure probability $\epsilon$. That is, for the upper bound we minimize
\begin{equation}\label{eq:opt_problem}
\begin{split}
    \min_{a,b} & \quad  \left[ b + a\left( \frac{2S_N^{\rm gs}}{N} - 1 \right) \right] \sqrt{N}
    \\
    \text{s.t.} & \quad
    \exp\left( \frac{-2(b^2 - a^2)}{\left(1 + \frac{4a}{3\sqrt{N}} \right)^2} \right) = \epsilon, 
\quad b \geq |a|,
\end{split}
\end{equation}
and similarly for the lower bound. The specific analytical solutions $b_{\rm U}$ and $a_{\rm U}$ for the upper bound, and $b_{\rm L}$ and $a_{\rm L}$ for the lower bound---which implicitly depend on $N$, $S_N^{\rm gs}$ and $\epsilon$---are given by~\cite{curras2021finite} 
\begin{equation}\label{eq:solutions_ab}
\begin{split}
a_{\rm U} &= 
\frac{
3\bigl[
72 \sqrt{N}\, S^{\rm gs}_{N}(N - S^{\rm gs}_{N}) \ln \epsilon 
- 16 N^{3/2} (\ln \epsilon)^2 
+ 9\sqrt{2}N 
\sqrt{-(N - 2 S^{\rm gs}_{N})^2\ln \epsilon \left( 9 S^{\rm gs}_{N}(N - S^{\rm gs}_{N}) - 2N \ln \epsilon \right)}
\bigr]
}{
4(9N - 8\ln \epsilon)\left[9 S^{\rm gs}_{N}(N - S^{\rm gs}_{N}) - 2N \ln \epsilon\right]
},
\\
b_{\rm U} &=
\frac{
\sqrt{
18 a^2_{\rm U} N 
- (16 a^2_{\rm U} + 24 a_{\rm U} \sqrt{N} + 9N)\ln \epsilon
}
}{
3\sqrt{2N}
},
\\
a_{\rm L} &=
\frac{
3\Bigl[
-72\sqrt{N}\, S^{\rm gs}_{N}(N - S^{\rm gs}_{N}) \ln \epsilon
+ 16 N^{3/2} (\ln \epsilon)^2
+ 9\sqrt{2}N
\sqrt{-(N - 2 S^{\rm gs}_{N})^2\ln \epsilon \left( 9 S^{\rm gs}_{N}(N - S^{\rm gs}_{N}) - 2N \ln \epsilon \right)}
\Bigr]
}{
4(9N - 8\ln \epsilon)\left[9 S^{\rm gs}_{N}(N - S^{\rm gs}_{N}) - 2N \ln \epsilon\right]
},
\\
b_{\rm L} &=
\frac{
\sqrt{
18 a^2_{\rm L} N
- (16 a^2_{\rm L} - 24 a_{\rm L} \sqrt{N} + 9N)\ln \epsilon
}
}{
3\sqrt{2N}
}.
\end{split}
\end{equation}
Given~\cref{eq:solutions_ab}, we can write the deviation terms as
\begin{equation}\label{eq:deltasK}
\begin{split}
    \Delta_{\mathrm{K}}^{\mathrm{L}} (N,\rv{S}_N,S_N^{\rm gs},\epsilon)
    &= \left[ b_{\rm L} + a_{\rm L}\left( \frac{2\rv{S}_N}{N} - 1 \right) \right] \sqrt{N},\\
    \Delta_{\mathrm{K}}^{\mathrm{U}} (N,\rv{S}_N,S_N^{\rm gs},\epsilon)
    &= \left[ b_{\rm U} + a_{\rm U}\left( \frac{2\rv{S}_N}{N} - 1 \right) \right] \sqrt{N},
\end{split}
\end{equation}
such that each of the bounds
\begin{equation}\label{eq:Kfunctions}
\begin{split}
    \sum_{u=1}^{N}\mathbb{E}[\rv{X}_{u}|\mathcal{F}_{u-1}] &\geq \rv{S}_N - \Delta_{\rm K}^{\rm L}(N,\rv{S}_N,S_N^{\rm gs},\epsilon),\\
    \sum_{u=1}^{N}\mathbb{E}[\rv{X}_{u}|\mathcal{F}_{u-1}] &\leq \rv{S}_N + \Delta_{\rm K}^{\rm U}(N,\rv{S}_N,S_N^{\rm gs},\epsilon),
\end{split}
\end{equation}
holds except with probability $\epsilon$.

Besides, from~\cref{lemma:kato} one can readily derive tight upper and lower bounds on $\rv{S}_N$ \cite{curras2021finite,navarrete2022improved}. In particular, for $\tilde{a}_{\rm U} < \sqrt{N}/2$, the following inequality holds
\begin{equation}\label{eq:K3_first}
\begin{split}
    \rv{S}_{N} 
    &\leq 
    \frac{N}{\sqrt{N}-2\tilde{a}_{\rm U}}\left(\frac{1}{\sqrt{N}} \rv{E}_{N}^{\mathcal{F}}+\tilde{b}_{\rm U}-\tilde{a}_{\rm U}\right),
\end{split}
\end{equation}
except with failure probability less than or equal to $\exp\left[-2(\tilde{b}_{\rm U}^2-\tilde{a}_{\rm U}^2)/(1-\frac{4\tilde{a}_{\rm U}}{3\sqrt{N}})^2\right]$. In~\cref{eq:K3_first}, the best $\tilde{a}_{\rm U}$ and $\tilde{b}_{\rm U}$ can be found again by fixing the failure probability to $\epsilon$ and considering a particular guess $E_{N}^{\mathcal{F},{\rm gs}}$ of $\rv{E}_{N}^{\mathcal{F}}$. Specifically, the values of $\tilde{a}_{\rm U}$ and $\tilde{b}_{\rm U}$ that minimize the RHS of~\cref{eq:K3_first} are given by
\begin{equation}\label{eq:optimal_ab_tilde}
\begin{split}
\tilde{a}_{\rm U} &=
\frac{3\sqrt{N}}{4}
\frac{
9  \sqrt{(N - 2 E_{N}^{\mathcal{F},{\rm gs}})^2N \ln\epsilon\,\bigl[N \ln\epsilon + 18 E_{N}^{\mathcal{F},{\rm gs}}(E_{N}^{\mathcal{F},{\rm gs}} - N)\bigr]}
+ 4N (\ln\epsilon)^{2}
+ 9 \ln\epsilon\, [3N^{2} - 8N E_{N}^{\mathcal{F},{\rm gs}} + 8 (E_{N}^{\mathcal{F},{\rm gs}})^{2}]
}{
36 \ln\epsilon\, [N^{2} - 2N E_{N}^{\mathcal{F},{\rm gs}} + 2 (E_{N}^{\mathcal{F},{\rm gs}})^{2}]
+ 4N (\ln\epsilon)^{2}
+ 81N E_{N}^{\mathcal{F},{\rm gs}}(N - E_{N}^{\mathcal{F},{\rm gs}})
},
\\
\tilde{b}_{\rm U} &=
\frac{
\sqrt{
18 \tilde{a}^{2}_{\rm U} N
- (16\tilde{a}^{2}_{\rm U} - 24\tilde{a}_{\rm U}\sqrt{N} + 9N)\ln\epsilon
}
}{
3\sqrt{2N}
}.
\end{split}
\end{equation}
As for the lower bound, we have that, for $\tilde{a}_{\rm L} > - \sqrt{N}/2$, the following inequality holds
\begin{equation}\label{eq:K4}
\begin{split}
    \rv{S}_{N} 
    &\geq 
    \frac{N}{\sqrt{N}+2\tilde{a}_{\rm L}}\left(\frac{1}{\sqrt{N}} \rv{E}_{N}^{\mathcal{F}}-\tilde{b}_{\rm L}+\tilde{a}_{\rm L}\right),
\end{split}
\end{equation}
except with failure probability less than or equal to $\exp\left[-2(\tilde{b}_{\rm L}^2-\tilde{a}_{\rm L}^2)/(1+\frac{4\tilde{a}_{\rm L}}{3\sqrt{N}})^2\right]$. Again, by fixing the failure probability to $\epsilon$ and now maximizing the RHS of~\cref{eq:K4}, we obtain
\begin{equation}
\begin{split}
\tilde{a}_{\rm L}
& =
\frac{3\sqrt{N}}{4}
\frac{
9 
\sqrt{(N-2 E_{N}^{\mathcal{F},{\rm gs}})^2N \ln \epsilon [N \ln \epsilon + 18 E_{N}^{\mathcal{F},{\rm gs}} (E_{N}^{\mathcal{F},{\rm gs}}-N)]}
-4 N (\ln\epsilon)^2 -9 \ln\epsilon [3 N^2-8 N E_{N}^{\mathcal{F},{\rm gs}}+8 (E_{N}^{\mathcal{F},{\rm gs}})^2]
}{
36 \ln (\epsilon) [N^2-2 N E_{N}^{\mathcal{F},{\rm gs}}+2 (E_{N}^{\mathcal{F},{\rm gs}})^2]+4 N (\ln\epsilon)^2+81 N E_{N}^{\mathcal{F},{\rm gs}} (N-E_{N}^{\mathcal{F},{\rm gs}})},
\\
\tilde{b}_{\rm L}
& = \frac{\sqrt{18 \tilde{a}_{\rm L}^2 N -(16 \tilde{a}_{\rm L}^2  + 24 \tilde{a}_{\rm L} \sqrt{N} +9 N) \ln\epsilon}
}{
3 \sqrt{2N}}.
\end{split}
\end{equation}

\Cref{lemma:kato} assumes normalized random variables satisfying $0\le \rv{X}_u\le1$. To deal with general unnormalized random variables $\rv{X}_u$ satisfying $x_{\rm min}\leq \rv{X}_u\leq x_{\rm max}$, we define associated normalized random variables as
\begin{equation}\label{eq:normalizationRVs}
    \tilde{\rm n}[\rv{X}_u] = \frac{\rv{X}_u - x_{\rm min}}{x_{\rm max}-x_{\rm min}}
\end{equation}
that do satisfy $0\leq \tilde{\rm n}[\rv{X}_u]\leq 1$. Thus, from the linearity of the expectation, it immediately follows from~\cref{eq:Kfunctions} that the unnormalized random variables satisfy the probabilistic bounds
\begin{equation}\label{eq:Kfunctions_unnormalized}
\begin{split}
    \sum_{u=1}^{N}\mathbb{E}[\rv{X}_{u}|\mathcal{F}_{u-1}] 
    &\geq  \rv{S}_N - (x_{\rm max}-x_{\rm min})\Delta_{\rm K}^{\rm L}(N,{\rm n}[\rv{S}_N],{\rm n}[S_N^{\rm gs}],\epsilon)\\
    & =: K_{\epsilon,N,x_{\min},x_{\max}}^{-1}(\rv{S}_N),
    \\
    \sum_{u=1}^{N}\mathbb{E}[\rv{X}_{u}|\mathcal{F}_{u-1}] 
    &\leq  \rv{S}_N + (x_{\rm max}-x_{\rm min})\Delta_{\rm K}^{\rm U}(N,{\rm n}[\rv{S}_N],{\rm n}[S_N^{\rm gs}],\epsilon)\\
    & =: K_{\epsilon,N,x_{\min},x_{\max}}^{1}(\rv{S}_N),
\end{split}
\end{equation}
except with probability $\epsilon$ each, where
\begin{equation}
\begin{split}
    {\rm n}[\rv{S}_N] :=  \sum_{u=1}^N \tilde{\rm n}\left[\rv{X}_u\right] = \frac{\rv{S}_N - N x_{\rm min}}{x_{\rm max}-x_{\rm min}} = N\tilde{\rm n}[\rv{S}_N],
\end{split}
\end{equation}
and ${\rm n}[S_N^{\rm gs}]$ is an arbitrary prediction on ${\rm n}[\rv{S}_N]$. Note that ${\rm n}[\rv{S}_N]$ is not the normalized random variable associated with $\rv{S}_N$; that role is played by $\tilde{\rm n}[\rv{S}_N]$.

\subsection{Upper bound on the number of phase errors of the VEP}
Now we show how to derive a probabilistic upper bound on the number of phase errors of the VEP, namely $\rv{M}_{\rm ph}$, by using Kato's inequality. We start from~\cref{eq:ineq_conditional_expect_Main}, which comprises three summations related to different types of RVs: those associated with the phase errors ($\rv{\chi}_{\rm ph}^{u}$); those associated with the observable statistics ($\rv{\chi}_{{\rm Q},l}^{u}$); and those associated with the fictitious local measurements ($\chiT{u}$). Our first goal is to obtain an upper bound on the RHS of~\cref{eq:ineq_conditional_expect_Main} that solely depends on the actual observable statistics. 

For this, we note that the first term in the RHS of~\cref{eq:ineq_conditional_expect_Main} can be bounded directly by employing~\cref{eq:Kfunctions_unnormalized}. In doing so, we obtain the probabilistic bound
\begin{equation}\label{eq:K1}
\begin{split}
    &\sum_l\eta_l^*\sum_{u=1}^{N} \mathbb{E}\left[\rv{\chi}_{{\rm Q},l}^{u}|\mathcal{F}_{u-1}\right] 
    \leq \sum_l\eta_l^*K_{\epsilon_l,N,x_{\min}^l,x_{\max}^l}^{\text{sign}(\eta_l^*)}(\rv{M}_{{\rm Q},l}),
\end{split}
\end{equation}
which holds except with probability at most $\epsproK=\sum_l \epsilon_l$, where the function $\text{sign}(x)$ returns $1$ ($-1$) when $x\geq 0$ ($x<0$). The quantities $x_{\min}^l$ and $x_{\max}^l$ represent the minimum and maximum values that the random variables $\rv{\chi}_{{\rm Q},l}^{u}$ can take, respectively, and $\epsilon_l$---we fix $\epsilon_l=\epsproK/n_{\rm Q}$ for simplicity---represents the failure probability of the $l$-th Kato bound for the summation in the RHS of~\cref{eq:K1}.  

Similarly, to upper bound the second term in the RHS of~\cref{eq:ineq_conditional_expect_Main}, we can use~\cref{eq:Kfunctions_unnormalized,eq:deltasK} to obtain
\begin{equation}\label{eq:K2}
\begin{split}
    \sum_{u=1}^N \mathbb{E}\left[\chiT{u}|\mathcal{F}_{u-1}\right] 
    &\leq 
    K_{\epsproK,N,\omega_{\min},\omega_{\max}}^{1}(\rv{M}_{\Tsind})
    \\
    &\leq
    K_{\epsproK,N,\omega_{\min},\omega_{\max}}^{1}(M_{\Tsind}^{\rm U}),
\end{split}
\end{equation}
which holds except with failure probability at most $\epsproK+\epsproB$. Here, $M_{\Tsind}^{\rm U}$ is an upper bound on $\rv{M}_{\Tsind}:=\sum_{u=1}^N\chiT{u}$ that holds except with probability $\epsproB$---which is provided in~\cref{thm:correlations}---and the quantities $\omega_{\min}$ and $\omega_{\max}$ represent the minimum and maximum values that the random variables $\chiT{u}$ can take, respectively. The second inequality of~\cref{eq:K2} holds if  $K_{\epsproK,N,\omega_{\min},\omega_{\max}}^{1}$ is non-decreasing, which can be enforced by substituting $a_{\rm U}$ with $\max\set{-\sqrt{N}/2,a_{\rm U}}$ in the definition of $b_{\rm U}$ in \cref{eq:solutions_ab}, and then again substituting $a_{\rm U}$ with $\max\set{-\sqrt{N}/2,a_{\rm U}}$ and the new definition of $b_{\rm U}$ in \cref{eq:deltasK}.

Finally, we can use~\cref{eq:K3_first} to obtain the inequality
\begin{equation}\label{eq:K3}
\begin{split}
    \rv{M}_{\rm ph} 
    &\leq 
    \frac{N}{\sqrt{N}-2\tilde{a}_{\rm U}}\left(\frac{1}{\sqrt{N}}\sum_{u=1}^N \mathbb{E}\left[\rv{\chi}_{\rm ph}^{u}|\mathcal{F}_{u-1}\right]+\tilde{b}_{\rm U}-\tilde{a}_{\rm U}\right),
\end{split}
\end{equation}
which holds except with failure probability at most $\epsproK$. The parameters $\tilde{a}_{\rm U}$ and $\tilde{b}_{\rm U}$ are given in~\cref{eq:optimal_ab_tilde} and are determined by the failure probability $\epsproK$ together with an unconfirmed guess of the sum of conditional expectations $\sum_{u=1}^N\mathbb{E}\left[\rv{\chi}_{\rm ph}^{u}|\mathcal{F}_{u-1}\right]$---or a guess of its upper bound. Such a guess can be obtained from a model of the channel or from previous executions of the protocol. By combining~\cref{eq:ineq_conditional_expect_Main} in the main text with~\cref{eq:K1,eq:K2,eq:K3}, we find that the inequality
\begin{equation}\label{eq:Kfinal}
\begin{split}
    \rv{M}_{\rm ph} 
    &\leq 
    \frac{N}{\sqrt{N}-2\tilde{a}_{\rm U}}\left[\frac{1}{\sqrt{N}}
    \left(
    \sum_l\eta_l^*K_{\epsilon_l,N,x_{\min}^l,x_{\max}^l}^{\text{sign}(\eta_l^*)}(\rv{M}_{{\rm Q},l})
    + 
    \frac{1-\ptrash}{\ptrash} K_{\epsproK,N,\omega_{\min},\omega_{\max}}^{1}(M_{\Tsind}^{\rm U}) 
    \right)
    +\tilde{b}_{\rm U}-\tilde{a}_{\rm U}\right]
    \\
    &=: \rv{M}_{\rm ph}^{\rm U},
\end{split}
\end{equation}
holds except with probability at most $3\epsproK+\epsproB$.

\subsection{Imperfect detectors}
Next, we extend the upper bound on the number of phase errors of the VEP, $\rv{M}_{\rm ph}^{\rm U}$, to the case in which there is a detection-efficiency mismatch at Bob's receiver.
In particular, we consider that the detection efficiencies and dark counts of Bob's single-photon detectors are not known precisely, but are only characterized within some known tolerances $\Delta_\eta$ and $\Delta_{\rm dc}$, respectively. That is,
\begin{equation}\label{eq:etadcmodel}
\begin{aligned}
\eta_{\beta_b} &\in [\eta_{\rm det}(1- \Delta_\eta),\eta_{\rm det}(1+ \Delta_\eta)], \\
d_{\beta_b} &\in [d_{\rm det}(1- \Delta_{\rm dc}),d_{\rm det}(1+ \Delta_{\rm dc})],
\end{aligned}
\end{equation} 
where $\eta_{\beta_b}$ ($d_{\beta_b}$) denotes the detection efficiency (dark-count rate) of Bob's detector associated with the basis $\beta$ and the bit value $b$.

Now, note that $\rv{M}_{\rm ph}^{\rm U}$ in \cref{eq:Kfinal} is independent of $\rv{M}_{\rm key}$ and is therefore non-decreasing with respect to $\rv{M}_{\rm key}$. Moreover, the ratio $\rv{M}_{\rm ph}^{\rm U}/\rv{M}_{\rm key}$ is non-increasing with respect to $\rv{M}_{\rm key}$. Consequently, as shown in \cite{curras2025security}, we can define the extended upper bound $\rv{M}_{\rm ph,\delta}^{\rm U}$ as
\begin{equation}
    \rv{M}_{\rm ph,\delta}^{\rm U} = \rv{M}_{\rm key}\left[\frac{\frac{\rv{M}_{\rm ph}^{\rm U}}{\rv{M}_{\rm key}} + \delta_1 + \gamma_{ \rm bin}^{\epsilon_{\rm dep-1}}(\rv{M}_{\rm key},\delta_1)}{1 - \delta_2 - \gamma_{ \rm bin}^{\epsilon_{\rm dep-2}}(\rv{M}_{\rm key},\delta_2) }\right],
    \label{eq:M_ph_delta}
\end{equation}
where $\delta_1$ and $\delta_2$ quantify the detection-efficiency mismatch and $\gamma_{\rm bin}$ is a finite-size deviation term defined as  
\begin{equation}
    \gamma_{\rm bin}^{\epsilon} (M, \delta) := \min \left\{ x \geq 0: \sum^{M}_{i = \lfloor M (\delta + x)\rfloor} \binom{M}{i} \delta^i (1-\delta)^{M-i} \leq \epsilon^2  \right\},
\end{equation} 
for any failure probabilities $\epsilon_{\rm dep-1},\epsilon_{\rm dep-2} > 0$. Then, in the presence of a detection-efficiency mismatch, it holds that
\begin{equation}
\Pr (\rv{M}_{\rm ph} > \rv{M}_{\rm ph,\delta}^{\rm U}) \leq 3\epsproK + \epsproB + \epsilon_{\rm dep-1}^2 + \epsilon_{\rm dep-2}^2.
\label{eq:M_ph_prob}
\end{equation}

To compute $\delta_1$ and $\delta_2$, we consider the canonical model for Bob's detectors, as in \cite{tupkary2024phase}. In this scenario, we have that $\delta_1$ and $\delta_2$ are bounded by
\begin{equation}
\begin{aligned}
&\delta_1 \leq \max \Bigg\{\Bigg(1- \frac{1-(1-d_{\rm min})^2}{1-(1-d_{\rm max})^2}\Bigg) \frac{d_{\rm max}(2-d_{\rm min})}{1-(1-d_{\rm min})^2},4 \Bigg|1 - \sqrt{1-(1-d_{\rm min})^2(1-r_{\eta})}\Bigg|\Bigg\}, \\ 
&\delta_2 \leq \max \Bigg\{ 1 - \frac{1-(1-d_{\rm min})^2}{1-(1-d_{\rm max})^2}, (1-d_{\rm min})^2(1-r_\eta) \Bigg\},
\end{aligned}
\end{equation}
where
\begin{equation}
\begin{aligned}
&d_{\rm max} = \max_{\beta\in \{X,Z\}} \{d_{\beta_0},d_{\beta_1}\} \leq d_{\rm det} (1+\Delta_{\rm dc}), \\
&d_{\rm min} = \min_{\beta\in \{X,Z\}} \{d_{\beta_0},d_{\beta_1}\} \geq d_{\rm det} (1-\Delta_{\rm dc}),\\
& r_\eta = (1-\Delta_\eta)/(1+\Delta_\eta).
\end{aligned}
\end{equation}
We remark that our analysis is not restricted to the canonical model; it can be applied to any detector model, as long as a bound on $\delta_1$ and $\delta_2$ can be obtained.

\subsection{Upper bound on the phase-error rate of the VP}
We now derive an upper bound on the number of phase errors $\rv{N}_{\rm ph}$ of the VP---a protocol in which no \LC{} rounds are defined---starting from the bound on $\rv{M}_{\rm ph}$ in the VEP given by~\cref{eq:M_ph_prob}.

To do so, we rely on two observations. First, the upper bound in~\cref{eq:M_ph_prob} depends solely on the RVs $\rv{M}_{{\rm Q},l}$ and $\rv{M}_{\rm key}$, i.e., it is independent of the outcomes of the \LC{} rounds. Second, the same bound applies in the VP. To see this, note that the VEP and the VP are indistinguishable from Eve's perspective during the quantum communication phase, since they differ only in the subsequent classical post-processing, where the VP requires public announcements that include information from rounds that are tagged as \LC{} in the VEP. Since the values of $\rv{M}_{{\rm Q},l}$, $\rv{M}_{\rm key}$, and $\rv{M}_{\rm ph}$ are already determined by the end of the quantum phase, and Eve cannot influence the protocol statistics after that point, the joint distribution of these RVs is identical in the two protocols. Consequently, any probabilistic statement about them that holds in the VEP, and in particular~\cref{eq:M_ph_prob}, also holds in the VP.

Note that this argument implicitly assumes that the RVs $\rv{M}_{{\rm Q},l}$, $\rv{M}_{\rm key}$, and $\rv{M}_{\rm ph}$ can be defined in the VP in the first place. They can, since we can view them as a random subsample---with sampling probability $(1-\ptrash)$---of the observable quantities $\rv{N}_{{\rm Q},l}$, $\rv{N}_{\rm key}$, and $\rv{N}_{\rm ph}$ measured in the VP.  More concretely, we can consider that Alice generates a binary random sequence $\vos{\rv{r}}_{\rm \LC}:=\rv{r}_{\rm \LC}^{1},\dots,\rv{r}_{\rm \LC}^{N}$, such that $\rv{r}_{\rm \LC}^{u}$ indicates if the round $u$ is tagged ($\rv{r}^u_{\rm \LC}=1$) or not ($\rv{r}^u_{\rm \LC}=0$) as a \LC{} round in the VEP, and uses it to compute $\rv{M}_{{\rm Q},l}$ ($\rv{M}_{\rm key}$) from the quantities $\rv{N}_{\rm test}^{i,\beta,b}$ ($\rv{N}_{\rm key}$)---in practice, though, it suffices to generate $\alpha$-biased bits only for the key rounds, and those test rounds that have been assigned to some $\rv{\chi}_{{\rm Q},l}^u$ in the VEP. 

Similarly, the set of phase-error events in the VEP can also be regarded as a random sample from the set of phase-error events in the VP. Thus, from Serfling's inequality (see~\cref{thm:serfling})---we note that tighter bounds can be found in~\cite{mannalath2024sharp}, though the improvement is expected to be negligible for the parameter regimes considered here---the number of phase errors in the subset of key rounds tagged as \LC{}, namely $\rv{S}_{\rm ph}:=\sum_{u=1}^N\rv{\chi}_{\rm ph}^u \rv{r}_{\rm \LC}^u=\rv{N}_{\rm ph}-\rv{M}_{\rm ph}$, can be estimated via a probabilistic bound of the form
\begin{equation}\label{eq:application_of_serfling}
\begin{split}
    &\Pr[\rv{S}_{\rm ph} 
    \geq 
    \frac{\rv{S}_{\rm key}}{\rv{M}_{\rm key}}\rv{M}_{\rm ph} 
    + \rv{S}_{\rm key}\Delta_{\rm S}(\rv{S}_{\rm key},\rv{M}_{\rm key},\epsproS)]
    \\
    &=\sum_{s,m}\Pr[\rv{S}_{\rm key}=s,\rv{M}_{\rm key}=m]\Pr[\rv{S}_{\rm ph} 
    \geq 
    \frac{s}{m}\rv{M}_{\rm ph} 
    + s\Delta_{\rm S}(s,m,\epsproS)\mid\rv{S}_{\rm key}=s,\rv{M}_{\rm key}=m]
    \leq \epsproS
\end{split}
\end{equation}
where the last inequality follows from~\cref{thm:serfling}, $\rv{S}_{\rm key}:=\sum_{u=1}^N\rv{\chi}_{\rm key}^u \rv{r}_{\rm \LC}^u$ denotes the number of sifted-key rounds tagged as \LC, and $\Delta_{\rm S}(r,s,\epsilon)$ is given in~\cref{eq:Delta_S}.

Therefore, from~\cref{eq:M_ph_prob,eq:application_of_serfling} it follows that
\begin{equation}\label{eq:final_bound_Nph}
\begin{split}
    \rv{N}_{\rm ph} 
    & = \rv{M}_{\rm ph} + \rv{S}_{\rm ph}
    \\
    &\leq 
    \rv{M}_{\rm ph}\left(1 + \frac{\rv{S}_{\rm key}}{\rv{M}_{\rm key}}\right) + \rv{S}_{\rm key}\Delta_{\rm S}(\rv{S}_{\rm key},\rv{M}_{\rm key},\epsproS)
    \\
    &\leq\rv{M}_{\rm ph,\delta}^{\rm U} \left(1 + \frac{\rv{S}_{\rm key}}{\rv{M}_{\rm key}}\right) + \rv{S}_{\rm key}\Delta_{\rm S}(\rv{S}_{\rm key},\rv{M}_{\rm key},\epsproS)
    \\
    & =:  N_{\rm ph}^{\rm U}\left(\rv{N}_{\rm key},\rv{M}_{\rm key},\set{\rv{M}_{{\rm Q},l}}_l\right)
\end{split}
\end{equation}
holds except with probability at most $\epsilon_{\rm pro}=3\epsproK+\epsproB+\epsproS+ \epsilon_{\rm dep-1}^2 + \epsilon_{\rm dep-2}^2$, where in the last line we have used that $\rv{S}_{\rm key}=\rv{N}_{\rm key}-\rv{M}_{\rm key}$. Note that if there is no detection-efficiency mismatch, the quantity $\rv{M}_{\rm ph,\delta}^{\rm U}$ in~\cref{eq:final_bound_Nph} can be substituted by $\rv{M}_{\rm ph}^{\rm U}$, in which case the failure probability becomes $\epsilon_{\rm pro}=3\epsproK+\epsproB+\epsproS$. 
In doing so, we obtain the following bound
\begin{equation}
    \Pr[\rv{e}_{\rm ph}> e_{\rm ph}^{\rm U}(\rv{N}_{\rm key},\rv{M}_{\rm key},\set{\rv{M}_{{\rm Q},l}}_l)]\leq \epsilon_{\rm pro},
\end{equation}
where
\begin{equation}
    e_{\rm ph}^{\rm U}(\rv{N}_{\rm key},\rv{M}_{\rm key},\set{\rv{M}_{{\rm Q},l}}_l)
    = \frac{N_{\rm ph}^{\rm U}\left(\rv{N}_{\rm key},\rv{M}_{\rm key},\set{\rv{M}_{{\rm Q},l}}_l\right)}{\rv{N}_{\rm key}}.
\end{equation}

We remark that, alternatively, one could straightforwardly obtain bounds on $\rv{M}_{\rm key}$ and the $\rv{M}_{{\rm Q},l}$ from the overall observed counts $\rv{N}_{\rm key}$ and the $\rv{N}_{\rm test}^{i,\beta,b}$ by applying standard concentration bounds for sums of independent random variables.  By doing so, one obtains a simpler bound of the form
\begin{equation}
    \Pr[\rv{e}_{\rm ph}> \tilde e_{\rm ph}^{\rm U}(\rv{N}_{\rm key},\set{\rv{N}_{\rm test}^{i,\beta,b}}_{i,\beta,b})]\leq \epsilon_{\rm pro},
\end{equation}
that is independent of the binary random sequence $\vos{\rv{r}}_{\rm \LC}$. This avoids having to actually generate these random bits in the actual protocol. However, it comes at the cost of a small reduction in the finite-key rate, since one needs to apply some additional concentration bounds.
\section{Technical results}\label{app:correlated_sources}
Here we prove \cref{thm:correlations,thm:correlations_decoy}, which are used in~\cref{sec:SecurityAnalysisNew} of the main text and \cref{app:general_protocols_decoy_state}, respectively. For this, we start by introducing~\cref{lemma:density_op_correlations}, which is an extension of the results derived in Appendix A of~\cite{pereira2025optimal} to correlated sources.

\begin{lemma}\label{lemma:density_op_correlations}
    Consider a QKD protocol with a (possibly) correlated source. For a sequence $i_1^N = i_1 \dots i_N$ of setting choices, with $i_u\in\mathcal{I}$ for some alphabet $\mathcal{I}$, the state emitted by Alice over $N$ rounds can be written as
	\begin{equation} \label{eq:global_state_corr}
	\sigma_{C_1^N}^{i_1^N} = \bigotimes_{u=1}^N \sigma_{C_u}^{i_{u-L}^u},
	\end{equation}
    where $\sigma_{C_u}^{i_{u-L}^u}$ denotes the state emitted in the $u$-th round, which depends on the settings selected in that round, as well as in the previous $L$ rounds. Suppose that
    \begin{equation} \label{eqapp:partial_characterization_corr}
        \ev{\sigma_{C_u}^{i_{u-L}^u}}{\phi_{i_u}}_{C_u} \geq 1 - \epsleak^{i_u}, \quad\forall i_{u-L}^u,
    \end{equation}
    for a fixed family of pure states $\{\ket{\phi_i}\}_{i\in\mathcal{I}}$. Besides, define an extended system $D_u=C_uC_u'$, and let $\ket{0}_{C_u'}$ be an arbitrary fixed state of system $C_u'$. Then, there exists a CPTP map $\Phi$ and pure states $\{\ket*{\psi_{i_{u-L}^u}}_{D_u}\}$ of the form
    \begin{equation}\label{eqapp:partial_characterization_pure_corr}
        \ket*{\psi_{i_{u-L}^u}}_{D_u} = \sqrt{1-\epsleak^{i_u}} \ket{\phi_{i_u}}_{D_u} + \sqrt{\epsleak^{i_u}} \ket*{\phi_{i_{u-L}^u}^{\perp}}_{D_u},
    \end{equation}
    where $\ket{\phi_{i_u}}_{D_u}=\ket{\phi_{i_u}}_{C_u}\ket{0}_{C_u'}$ and $\braket*{\phi_{i_{u-L}^u}^{\perp}}{\phi_{i_u}}_{D_u} = 0$, such that
    \begin{equation}
    \label{eq:map_corr}
        \sigma^{i_1^N}_{C_1^N}  = \Phi\left(\tau^{i_1^N}_{D_1^N} \right), \quad \forall i_1^N,
    \end{equation}
	with
	\begin{equation}
    \label{eq:sigma_corr}
        \tau^{i_1^N}_{D_1^N} = \bigotimes_{u=1}^N \ketbra*{\psi_{i_{u-L}^u}}_{D_u}.
    \end{equation}
    As a consequence, if security can be established for \emph{all} global pure state families $\{\tau^{i_1^N}_{D_1^N}\}_{i_1^N}$ of the form in \cref{eq:sigma_corr} with each factor satisfying \cref{eqapp:partial_characterization_pure_corr}, then security is guaranteed for \emph{all} global (possibly mixed) state families $\{\sigma^{i_1^N}_{C_1^N}\}_{i_1^N}$ of the form in \cref{eq:global_state_corr} with each factor satisfying \cref{eqapp:partial_characterization_corr}.

  \end{lemma}
\begin{proof}
Fix a sequence $i_1^N$. Since the global state factorizes as $\sigma_{C_1^N}^{i_1^N} = \bigotimes_{u=1}^N \sigma_{C_u}^{i_{u-L}^u}$, we can apply a purification construction to each factor independently.

Consider a fixed round $u$ and a fixed history $i_{u-L}^u$. By assumption, the state $\sigma_{C_u}^{i_{u-L}^u}$ satisfies \cref{eqapp:partial_characterization_corr}. If $\epsleak^{i_u}=0$, this implies that $\sigma_{C_u}^{i_{u-L}^u}=\dyad{\phi_{i_u}}_{C_u}$, and \cref{eqapp:partial_characterization_pure_corr} holds trivially with $\ket*{\psi_{i_{u-L}^u}}_{D_u}=\ket{\phi_{i_u}}_{D_u}$; thus, in what follows we consider that $\epsleak^{i_u}>0$. By Uhlmann's theorem~\cite{uhlmann1976transition}, there exists a purification $\ket*{\psi_{i_{u-L}^u}''}_{C_uS_u}$ of $\sigma_{C_u}^{i_{u-L}^u}$ such that
\begin{equation}
\label{eqapp:purification_assumption_corr}
    \abs{\braket*{\phi_{i_u}}{\psi_{i_{u-L}^u}''}_{C_uS_u}}^2 = \ev{\sigma_{C_u}^{i_{u-L}^u}}{\phi_{i_u}}_{C_u} \geq 1 - \epsleak^{i_u},
\end{equation}
where we have defined $\ket{\phi_{i_u}}_{C_uS_u} \coloneqq \ket{\phi_{i_u}}_{C_u}\ket{0}_{S_u}$ and $S_u$ denotes an ancillary system for round $u$. Without loss of generality, \cref{eqapp:purification_assumption_corr} implies that $\ket*{\psi_{i_{u-L}^u}''}_{C_uS_u}$ can be expressed as
\begin{equation}
    \ket*{\psi_{i_{u-L}^u}''}_{C_uS_u} = e^{i \varphi_{i_{u-L}^u}}\left(\sqrt{1-\tilde{\epsleak}^{i_{u-L}^u}} \ket{\phi_{i_u}}_{C_uS_u} + \sqrt{\tilde{\epsleak}^{i_{u-L}^u}} \ket*{\phi_{i_{u-L}^u}^{\perp\prime}}_{C_uS_u}\right),
\end{equation}
where $0\leq \tilde{\epsleak}^{i_{u-L}^u} \leq \epsleak^{i_u}$, $\varphi_{i_{u-L}^u} \in [0,2\pi)$, and $\braket*{\phi_{i_{u-L}^u}^{\perp\prime}}{\phi_{i_u}}_{C_uS_u} = 0$. Removing the global phase, we define
\begin{equation}
    \ket*{\psi_{i_{u-L}^u}'}_{C_uS_u}  
    = e^{-i \varphi_{i_{u-L}^u}} \ket*{\psi_{i_{u-L}^u}''}_{C_uS_u} 
    = \sqrt{1-\tilde{\epsleak}^{i_{u-L}^u}} \ket{\phi_{i_u}}_{C_uS_u} + \sqrt{\tilde{\epsleak}^{i_{u-L}^u}} \ket*{\phi_{i_{u-L}^u}^{\perp\prime}}_{C_uS_u},
\end{equation}
which is also a purification of $\sigma_{C_u}^{i_{u-L}^u}$.

Next, we introduce a fictitious system $F_u$ for round $u$ to replace $\tilde{\epsleak}^{i_{u-L}^u}$ with $\epsleak^{i_u}$. Define
\begin{equation}\label{eq:lemma_corr_1}
    \ket*{\psi_{i_{u-L}^u}}_{C_uS_uF_u} = \ket*{\psi_{i_{u-L}^u}'}_{C_uS_u} \otimes \left[\sqrt{\frac{1-\epsleak^{i_u}}{1-\tilde{\epsleak}^{i_{u-L}^u}}} \ket{0}_{F_u} + \sqrt{1-\frac{1-\epsleak^{i_u}}{1-\tilde{\epsleak}^{i_{u-L}^u}}} \ket{1}_{F_u} \right],
\end{equation}
where $\{\ket{0}_{F_u},\ket{1}_{F_u}\}$ is an orthonormal basis for system $F_u$. \cref{eq:lemma_corr_1} can be rewritten as
\begin{equation}
    \ket*{\psi_{i_{u-L}^u}}_{C_uS_uF_u} = \sqrt{1-\epsleak^{i_u}} \ket{\phi_{i_u}}_{C_uS_uF_u} + \sqrt{\epsleak^{i_u}} \ket*{\phi_{i_{u-L}^u}^{\perp}}_{C_uS_uF_u},
\end{equation}
where $\ket{\phi_{i_u}}_{C_uS_uF_u} \coloneqq \ket{\phi_{i_u}}_{C_uS_u} \ket{0}_{F_u} = \ket{\phi_{i_u}}_{C_u}  \ket{0}_{S_u} \ket{0}_{F_u} $ and 
\begin{equation}
    \begin{split}
        \ket*{\phi_{i_{u-L}^u}^{\perp}}_{C_uS_uF_u} &= \sqrt{\frac{\epsleak^{i_u} - \tilde{\epsleak}^{i_{u-L}^u}}{\epsleak^{i_u}}} \ket{\phi_{i_u}}_{C_uS_u}\ket{1}_{F_u} 
        \\
        &+ \sqrt{\frac{\tilde{\epsleak}^{i_{u-L}^u}}{\epsleak^{i_u}}} \ket{\phi_{i_{u-L}^u}^{\perp\prime}}_{C_uS_u} \otimes \left[\sqrt{\frac{1-\epsleak^{i_u}}{1-\tilde{\epsleak}^{i_{u-L}^u}}} \ket{0}_{F_u} + \sqrt{\frac{\epsleak^{i_u} - \tilde{\epsleak}^{i_{u-L}^u}}{1-\tilde{\epsleak}^{i_{u-L}^u}}} \ket{1}_{F_u}\right]
    \end{split}
\end{equation}
is a normalized state orthogonal to $\ket{\phi_{i_u}}_{C_uS_uF_u}$. 

The previous construction holds for every round $u$ and every possible history $i_{u-L}^u$. Define the global pure state
\begin{equation}
    \tau^{i_1^N}_{C_1^NS_1^NF_1^N} = \bigotimes_{u=1}^N \ketbra*{\psi_{i_{u-L}^u}}_{C_uS_uF_u}.
\end{equation}
Now define the CPTP map $\Phi(\cdot) = \Tr_{S_1^N F_1^N}(\cdot)$, which traces out all ancillary and fictitious systems. By construction, for each factor we have $\Tr_{S_uF_u}\left(\ketbra*{\psi_{i_{u-L}^u}}_{C_uS_uF_u}\right) = \sigma_{C_u}^{i_{u-L}^u}$. Therefore,
\begin{equation}
  \sigma^{i_1^N}_{C_1^N} = \bigotimes_{u=1}^N \sigma_{C_u}^{i_{u-L}^u} = \Phi(\tau^{i_1^N}_{C_1^NS_1^NF_1^N}),
\end{equation}
which is identical to \cref{eq:map_corr} after the renaming $C_uS_uF_u\equiv C_uC_u'\equiv D_u$. Since this argument holds for every sequence $i_1^N$, the lemma is proven.

\textbf{Security implication.} Consider a fixed pair of families $\{\sigma^{i_1^N}_{C_1^N}\}_{i_1^N}$ and $\{\tau^{i_1^N}_{D_1^N}\}_{i_1^N}$ related by \cref{eq:map_corr}. One can obtain the former family by applying the global map $\Phi$ to the latter family. Crucially, $\Phi$ is independent of the sequence $i_1^N$, since it is defined as a partial trace over fixed systems. Therefore, $\Phi$ can be absorbed into Eve's attack and, consequently, any protocol secure against general coherent attacks when Alice prepares $\{\tau^{i_1^N}_{D_1^N}\}_{i_1^N}$ remains secure when she prepares $\{\sigma^{i_1^N}_{C_1^N}\}_{i_1^N}$ (see, e.g., \cite[Lemma 8]{nahar2024postselection} for a formal proof of this statement). Since we have shown that every family $\{\sigma^{i_1^N}_{C_1^N}\}_{i_1^N}$ satisfying \cref{eq:global_state_corr,eqapp:partial_characterization_corr} can be obtained via \cref{eq:map_corr} from some family $\{\tau^{i_1^N}_{D_1^N}\}_{i_1^N}$ of the form in \cref{eq:sigma_corr,eqapp:partial_characterization_pure_corr}, it follows that proving security for all such pure state families $\{\tau^{i_1^N}_{D_1^N}\}_{i_1^N}$ guarantees security for all such mixed state families $\{\sigma^{i_1^N}_{C_1^N}\}_{i_1^N}$.

\end{proof}

Next, we prove \cref{thm:correlations}, which is a crucial piece of the overall security analysis. Note that \cref{thm:correlations} applies to both correlated sources ($L>0$) and uncorrelated sources ($L=0$), and thus the overall analysis is valid in both situations. Alternatively, one could prove security for uncorrelated sources ($L=0$), and then extend the overall proof to correlated sources using the modular lifting theorems introduced in Ref.~\cite{curras-lorenzoRigorousPhaseerrorestimation2026}. However, the latter approach requires dividing the protocol rounds into $(L+1)$ groups and applying the overall phase-error estimation procedure separately for each group, which reduces the finite-key efficiency. On the other hand, here, we are able to directly obtain a bound on the number of phase errors as a function of the overall counts observed in the protocol, without any explicit round division in the post-processing step (see \cref{eq:Kfinal_main}). The round division is only introduced as a fictitious step below, in \cref{thm:correlations}, to obtain an upper bound on the fictitious quantity $\rv{M}_{\Tsind}$. The way that correlated sources are treated in \cref{thm:correlations} reuses key ideas from \cite{curras-lorenzoRigorousPhaseerrorestimation2026,curras2026securitydecoy}.

\begin{theorem}[Probabilistic upper bound on $\rv{M}_{\Tsind}$]
\label{thm:correlations}
Consider the virtual estimation protocol (VEP) introduced in Box~\ref{box:VEP}, and let $\ket{\Psi^{\epsleak}}_{A_1^NC_1^N}$, $N$, $\chiT{u}$, $\rv{M}_{\Tsind}$, $\alpha$, and $\Tobs{A}$ be as defined there, where $\ket{\Psi^{\epsleak}}_{A_1^NC_1^N}$ is given in~\cref{eq:global_correlated_state_main} and depends on a certain set of reference states $\set{\ket{\phi_i}}_{i\in\set{0,1,\dots,\nstates-1}}$, setting probabilities $\set{\piA{i}}_{i\in\set{0,1,\dots,\nstates-1}}$, correlation length $L$, and on the parameter $\epsleak\geq0$.
Let the quantities $E_{\chi_{\Tsind}}^{\rm L}$, $E_{\chi_{\Tsind}}^{\rm U}$ and $E_{\chi_{\Tsind}^2}^{\rm U}$ satisfy
\begin{equation}
\begin{split}
E_{\chi_{\Tsind}}^{\rm L}
&\leq \ptrash\min_{G\in\mathcal{S}_{\rm const}}\left[\Tr\set{\Tobs{A} \rho_{A}^{\varepsilon'}}\right],
\\
E_{\chi_{\Tsind}}^{\rm U}
&\geq\ptrash\max_{G\in\mathcal{S}_{\rm const}}\left[\Tr\set{\Tobs{A} \rho_{A}^{\varepsilon'}}\right],
\\
E_{\chi_{\Tsind}^2}^{\rm U}
&\geq\ptrash\max_{G\in\mathcal{S}_{\rm const}}\left[\Tr\set{(\Tobs{A})^2 \rho_{A}^{\varepsilon'}}\right],
\end{split}
\end{equation}
where
\begin{equation}\label{eq:G_matrix_Thm1}
G:=
\begin{bmatrix}
G_{\phi} & G_{\rm c}     \\
G_{\rm c}^{\dagger}   & G_{\perp}   \\
\end{bmatrix}
\end{equation}
is a $2\nstates \times 2\nstates$ Hermitian block matrix built from the $\nstates\times\nstates$ matrices $G_{\phi}$, $G_{\rm c}$, and $G_{\perp}$; $\rho_{A}^{\varepsilon'}$ is a density matrix whose entries depend linearly on those of $G$, with $[\rho_{A}^{\varepsilon'}]_{ij} := \sqrt{\piA{i}\piA{j}}\big[(1-\varepsilon')[G_{\phi}]_{ji} +\varepsilon'[G_{\perp}]_{ji} + \sqrt{(1-\varepsilon')\varepsilon'} ([G_{\rm c}]_{ij}^{*} + [G_{\rm c}]_{ji})\big]$ and $\epsleak':=1-(1-\epsleak)^{L+1}$; and $\mathcal{S}_{\rm const}$ is a convex set defined by the constraints
\begin{equation}\label{eq:SDP_gram}
    \begin{split}
        \quad & G \succeq 0,\\
        \quad & [G_{\phi}]_{ij}=\braket{\phi_{i}}{\phi_{j}}, \quad i,j\in\set{0,1,\dots,\nstates-1},\\
        \quad & [G_{\rm c}]_{ii}=0, \quad i\in\set{0,1,\dots,\nstates-1},\\
        \quad & [G_{\perp}]_{ii}=1,\quad i\in\set{0,1,\dots,\nstates-1}.
    \end{split}
\end{equation}
Then, for any $\epsproB>0$, we have that
\begin{align}
\Pr[\rv{M}_{\Tsind}\geq M_{\Tsind}^{\rm U}] &\leq \epsproB
\end{align}
holds for
\begin{align}
M_{\Tsind}^{\rm U}:=& N E_{\chi_{\Tsind}}^{\rm U}+(L+1)\bar{N}_L\tilde{\Delta}_{\rm B},
\end{align}
where

\begin{align}
\tilde{\Delta}_{\rm B} := &\frac{c}{3\bar{N}_L} \ln\frac{L+1}{\epsproB}
+ \sqrt{ \left(\frac{c}{3\bar{N}_L} \ln\frac{L+1}{\epsproB}\right)^2  + 2 \frac{V_{\chi_{\Tsind}}^{\rm U}}{\bar{N}_L} \ln\frac{L+1}{\epsproB}},
\\
c := & \max\set{\omega_{\rm max}-E_{\chi_{\Tsind}}^{\rm L},E_{\chi_{\Tsind}}^{\rm U}-\omega_{\rm min}},
\\
V_{\chi_{\Tsind}}^{\rm U} := 
& 
E_{\chi_{\Tsind}^2}^{\rm U} - \left( \max(0, E_{\chi_{\Tsind}}^{\rm L}) + \min(0, E_{\chi_{\Tsind}}^{\rm U}) \right)^2,
\end{align}
$\bar{N}_{L}:=\lceil N/(L+1) \rceil$, and $\omega_{\min}$ ($\omega_{\max}$) denotes the minimum (maximum) possible value of the RVs $\chiT{u}$.
\end{theorem}

\begin{proof}
Let us define the sets of rounds 
\begin{equation}\label{eq:set_Sm_Thm1}
 \mathcal{S}_m:=\set{u\in\mathbb{N}:u\leq N, u \bmod (L+1)=m},   
\end{equation}
with $m\in\set{0,\dots,L}$, such that $\rv{M}_{\Tsind}=\sum_{u=1}^{N}\chiT{u}=\sum_{m}\sum_{u\in\mathcal{S}_m}\chiT{u}$. We shall bound $\rv{M}_{\Tsind}$ by bounding each of the sums $\rv{S}_m:=\sum_{u\in\mathcal{S}_m}\chiT{u}$ individually, using a similar approach as in \cite{curras2026securitydecoy}. 

In particular, let us focus on the sum $\rv{S}_0$.
We aim to obtain a probabilistic bound of the form 
\begin{equation}\label{eq:PS00}
    \Pr[\rv{S}_0\geq\delta^{\rm U}_0]\leq \epsilon,
\end{equation}
for certain $\delta^{\rm U}_0\in\mathbb{R}$ and $\epsilon>0$. For this, we note that the probability distribution of $\rv{S}_0$ does not change if we consider a variant of the VEP in which, in the rounds not belonging to $\mathcal{S}_0$ (including the \LC{} rounds), Alice measures her ancilla systems in the computational basis $\set{\ket{i}}_i$. This is possible because measurements in $\mathcal{S}_0$ and $\mathcal{S}_k$ commute if $k\neq 0$. Thus, we can consider this alternative scenario to obtain a probabilistic bound that must also hold for the VEP. 
In particular, let $\vec{j}_{\bar{\mathcal{S}_0}}=\bigotimes_{u\notin \mathcal{S}_0}\ket{j_u}_{A_u}$ denote a particular outcome of Alice's measurements in the rounds not belonging to $\mathcal{S}_0$, and let us assume that we know a certain $\epsilon$ such that $\Pr[\rv{S}_0\geq\delta^{\rm U}_0 | \vec{j}_{\bar{\mathcal{S}_0}}]\leq \epsilon$ holds for any possible outcome $\vec{j}_{\bar{\mathcal{S}_0}}$. Then, since $\delta^{\rm U}_0$ is independent of $\vec{j}_{\bar{\mathcal{S}_0}}$, it is clear that
\begin{equation}\label{eq:PS01}
\begin{split}
\Pr[\rv{S}_0\geq\delta^{\rm U}_0]&=\sum_{\vec{j}_{\bar{\mathcal{S}_0}}}\Pr[\vec{j}_{\bar{\mathcal{S}_0}}]\Pr[\rv{S}_0\geq\delta^{\rm U}_0 | \vec{j}_{\bar{\mathcal{S}_0}}]\leq \epsilon. 
\end{split}
\end{equation}

The \textit{a priori} state of all systems $A_1^N$ and $C_1^N$ considered in the VEP is given in~\cref{eq:global_correlated_state_main}. To obtain a bound on $\Pr[\rv{S}_0\geq\delta^{\rm U}_0 | \vec{j}_{\bar{\mathcal{S}_0}}]$, we shall consider the entangled state of the systems conditioned on the outcome $\vec{j}_{\bar{\mathcal{S}_0}}$. This state is given by 
\begin{equation}\label{eq:cond_state}
\begin{split}
\ket*{\Psi_{|\vec{j}_{\bar{\mathcal{S}_0}}}}
&=\bigotimes_{u\in \mathcal{S}_0} \sum_{i_u} \sqrt{\piA{i_u}}\ket{i_u}_{A_u}
\bigotimes_{k=0}^{L}\ket{\psi_{j_{u+1}^{u+k}i_u j_{u-L+k}^{u-1}}^{\epsleak}}_{C_{u+k}}
\\
&=\bigotimes_{u\in \mathcal{S}_0} \ket{\Psi_{|j_{u+1}^{u+L}j^{u-1}_{u-L}}^u}_{A_uC_u^{u+L}},
\end{split}
\end{equation}
where $\ket*{\psi_{i_{u-L}^{u}}^{\epsleak}}_{C_{u}}$ is defined in~\cref{eq:assumption1extra_main} of the main text, and
\begin{equation}\label{eq:eq:cond_state_u}
\begin{split}
\ket{\Psi_{|j_{u+1}^{u+L}j^{u-1}_{u-L}}^u}_{A_uC_u^{u+L}}
:=
&\sum_{i_u} \sqrt{\piA{i_u}}\ket{i_u}_{A_u}
\bigotimes_{k=0}^{L}\ket*{\psi_{j_{u+1}^{u+k}i_uj^{u-1}_{u-L+k}}^{\epsleak}}_{C_{u+k}}.
\end{split}
\end{equation}
The product form of~\cref{eq:cond_state} follows from the fact that consecutive rounds in $\mathcal{S}_0$ are separated by $L+1$ rounds, and hence the setting $i_u$ of any round $u\in\mathcal{S}_0$ only affects the systems $C_u^{u+L}$, which are not affected by the setting of any other round in $\mathcal{S}_0$. Moreover, we remark that the RVs $\chiT{u}$ are the outcomes of local measurements on Alice's ancilla systems, and thus their joint statistics are fully determined by the reduced state of $A_1^N$, which is independent of the CPTP map applied by Eve on the systems $C_1^N$. This means that the bound derived below holds for any coherent attack. Finally, note that for the first and last $L$ rounds of the protocol, the setting histories and the systems $C_u^{u+L}$ are truncated, as the sequences are understood to be restricted to the rounds within $[1,N]$. To treat these rounds on the same footing as the others, one can formally append $L$ fictitious rounds before round $1$ and after round $N$, with fixed settings and with systems prepared by Alice in the form of~\cref{eq:assumption1extra_main} but never sent. Since~\cref{eq:assumption1thm_main} holds for any history, this does not alter the actual protocol, and all the states $\ket*{\Vec{\psi}_{i_u}}_{C_u^{u+L}}$ introduced below then contain exactly $L+1$ factors.
From~\cref{eq:cond_state} it is clear that 
$\Pr[\rv{\chi}^{u}_{\Tsind}=\omega,\rv{\chi}^{v}_{\Tsind}=\omega'|\vec{j}_{\bar{\mathcal{S}_0}}]=\Pr[\rv{\chi}^{u}_{\Tsind}=\omega|\vec{j}_{\bar{\mathcal{S}_0}}]\Pr[\rv{\chi}^{v}_{\Tsind}=\omega'|\vec{j}_{\bar{\mathcal{S}_0}}]$ for any $u,v\in\mathcal{S}_0$ and for any outcomes $\omega$ and $\omega'$. Therefore, all the random variables in the set $\set{\rv{\chi}^{u}_{\Tsind}}_{u\in\mathcal{S}_0}$ are independent conditioned on the outcome $\vec{j}_{\bar{\mathcal{S}_0}}$, and one can use concentration bounds for independent random variables to obtain a bound on $\Pr[\rv{S}_0\geq\delta^{\rm U}_{0,\vec{j}_{\bar{\mathcal{S}_0}}} | \vec{j}_{\bar{\mathcal{S}_0}}]$. 
In particular, \cref{thm:bernstein} (see Appendix~\ref{app:finite-key}) guarantees that 
\begin{equation}\label{eq:bound_dependent_j}
    \Pr[\rv{S}_0\geq \mathbb{E}[\rv{S}_0| \vec{j}_{\bar{\mathcal{S}_0}}] + N_0\Delta_{\rm B}(N_0,v_0,c,\epsilon) | \vec{j}_{\bar{\mathcal{S}_0}}]\leq \epsilon,
\end{equation}
where $\Delta_{\rm B}(N_0,v_0,c,\epsilon)$ is given by~\cref{eq:DeltaB}, with
$N_0:=\abs{\mathcal{S}_0}$, $\epsilon>0$, $v_0:=\frac{1}{N_0}\sum_{u\in\mathcal{S}_0}{\rm Var}\left[\rv{\chi}^{u}_{\Tsind}| \vec{j}_{\bar{\mathcal{S}_0}}\right]$, and $c\geq \max\abs{\rv{\chi}^{u}_{\Tsind}-\mathbb{E}[\rv{\chi}^{u}_{\Tsind}| \vec{j}_{\bar{\mathcal{S}_0}}]}$ for all $u\in\mathcal{S}_0$. 

Now, given $\delta^{\rm U}_{0,\vec{j}_{\bar{\mathcal{S}_0}}}:=\mathbb{E}[\rv{S}_0| \vec{j}_{\bar{\mathcal{S}_0}}] + N_0\Delta_{\rm B}(N_0,v_0,c,\epsilon)$, we want to make this upper bound independent of $\vec{j}_{\bar{\mathcal{S}_0}}$ by considering a worst-case scenario. That is, we aim to find a bound $\delta^{\rm U}_0$ such that $\delta^{\rm U}_{0,\vec{j}_{\bar{\mathcal{S}_0}}}\leq \delta^{\rm U}_0$ for any $\vec{j}_{\bar{\mathcal{S}_0}}$. This can be achieved by employing SDP-based techniques~\cite{curras2023security2} to compute worst-case bounds satisfying $c\geq \max\abs{\rv{\chi}^{u}_{\Tsind}-\mathbb{E}[\rv{\chi}^{u}_{\Tsind}| \vec{j}_{\bar{\mathcal{S}_0}}]}$,  $E_{\chi_{\Tsind}}^{\rm U}\geq\mathbb{E}[\rv{\chi}^{u}_{\Tsind}|\vec{j}_{\bar{\mathcal{S}_0}}]$ and $V_{\chi_{\Tsind}}^{\rm U}\geq {\rm Var}[\rv{\chi}^{u}_{\Tsind}|\vec{j}_{\bar{\mathcal{S}_0}}]$, that are valid for any $\vec{j}_{\bar{\mathcal{S}_0}}$ and $u\in\mathcal{S}_0$, where
\begin{equation}
\begin{split}
\mathbb{E}\left[\rv{\chi}^{u}_{\Tsind}|\vec{j}_{\bar{\mathcal{S}_0}}\right]
&=
\ptrash\bra*{\Psi_{|j_{u+1}^{u+L}j^{u-1}_{u-L}}^u}\Tobs{A}\ket*{\Psi_{|j_{u+1}^{u+L}j^{u-1}_{u-L}}^u},
\\
{\rm Var}\left[\rv{\chi}^{u}_{\Tsind}|\vec{j}_{\bar{\mathcal{S}_0}}\right]
&=
\ptrash\bra*{\Psi_{|j_{u+1}^{u+L}j^{u-1}_{u-L}}^u}\left(\Tobs{A}\right)^2\ket*{\Psi_{|j_{u+1}^{u+L}j^{u-1}_{u-L}}^u} - \left(\ptrash\bra*{\Psi_{|j_{u+1}^{u+L}j^{u-1}_{u-L}}^u}\Tobs{A}\ket*{\Psi_{|j_{u+1}^{u+L}j^{u-1}_{u-L}}^u}\right)^2.
\end{split}
\end{equation}

For this, first note that, due to~\cref{eq:assumption1extra_main}, the states $\ket*{\Vec{\psi}_{i_u}}_{C_u^{u+L}}
:=\bigotimes_{k=0}^{L}\ket*{\psi_{j_{u+1}^{u+k}i_uj^{u-1}_{u-L+k}}^{\epsleak}}_{C_{u+k}}$ that appear in~\cref{eq:eq:cond_state_u} admit the decomposition \cite{pereira2020quantum}
\begin{equation}\label{eq:decomposition_corr}
\begin{split}
\ket*{\Vec{\psi}_{i_u}}_{C_u^{u+L}}
= &
\sqrt{1-\epsleak'}\ket*{\Vec{\phi}_{i_u}}_{C_u^{u+L}}
+\sqrt{\epsleak'}\ket*{\Vec{\phi}_{i_u}^{\perp}}_{C_u^{u+L}},
\end{split}
\end{equation}
where $\epsleak'=1-(1-\epsleak)^{L+1}\approx (L+1)\epsleak$; the states $\ket*{\Vec{\phi}_{i_u}}_{C_u^{u+L}}=\ket*{\phi_{i_{u}}}_{C_{u}}\bigotimes_{k=1}^{L}\ket*{\phi_{j_{u+k}}}_{C_{u+k}}$ satisfy $\braket*{\Vec{\phi}_{i_u}}{\Vec{\phi}_{i_u'}}=\braket*{\phi_{i_u}}{\phi_{i_u'}}$ for any pair $(i_u,i_u')$; and the state $\ket*{\Vec{\phi}_{i_u}^{\perp}}_{C_u^{u+L}}$ is orthogonal to $\ket*{\Vec{\phi}_{i_u}}_{C_u^{u+L}}$ for any $i_u$. Here, for simplicity of notation, we drop the explicit conditioning on the fixed setting sequence $j_{u+1}^{u+L}j^{u-1}_{u-L}$.

Then, we can construct a $2n_A\times2n_A$ Gram matrix $G$ containing the inner products of the states $\ket*{\Vec{\phi}_{0}},\dots,\ket*{\Vec{\phi}_{n_A-1}},\ket*{\Vec{\phi}_{0}^{\perp}},\dots,\ket*{\Vec{\phi}_{n_A-1}^{\perp}}$ for a given $j_{u+1}^{u+L}j^{u-1}_{u-L}$. Precisely, we define $G$ as given in~\cref{eq:G_matrix_Thm1}, where each of the $n_A\times n_A$ sub-blocks is defined by its entries as $[G_{\phi}]_{ij}=\braket{\phi_{i}}{\phi_{j}}$, $[G_{\rm c}]_{ij}=\braket*{\Vec{\phi}_{i}}{\Vec{\phi}_{j}^{\perp}}$, and $[G_{\perp}]_{ij}=\braket*{\Vec{\phi}_{i}^{\perp}}{\Vec{\phi}_{j}^{\perp}}$. 
Note that, due to~\cref{eq:decomposition_corr}, the marginal state $\rho_{A}=\Tr_{C_u^{u+L}}\set{\dyad*{\Psi_{|j_{u+1}^{u+L}j^{u-1}_{u-L}}^u}_{A_uC_u^{u+L}}}$ can be written as a function of the elements of $G$ as
\begin{equation}
\begin{split}
[\rho_{A}]_{ij} =& \sqrt{\piA{i}\piA{j}}\braket*{\Vec{\psi}_{j}}{\Vec{\psi}_{i}}
\\
=& \sqrt{\piA{i}\piA{j}}
\big[
(1-\epsleak')[G_{\phi}]_{ji} + \epsleak'[G_{\perp}]_{ji}
+ \sqrt{(1-\epsleak')\epsleak'} ([G_{\rm c}]_{ij}^{*} + [G_{\rm c}]_{ji})
\big].
\end{split}
\end{equation}
Thus, $E_{\chi_{\Tsind}}^{\rm U}$ can be taken to be any quantity satisfying
\begin{equation}
\begin{split}
E_{\chi_{\Tsind}}^{\rm U}
&\geq\ptrash\max_{G\in\mathcal{S}_{\rm const}}\left[\Tr\set{\Tobs{A} \rho_A}\right]\geq\mathbb{E}[\rv{\chi}^{u}_{\Tsind}|\vec{j}_{\bar{\mathcal{S}_0}}],
\end{split}
\end{equation}
where $\mathcal{S}_{\rm const}$ is the set of constraints given in~\cref{eq:SDP_gram}.
An upper bound on the variance $V_{\chi_{\Tsind}}^{\rm U}\geq {\rm Var}[\chiT{u}|\vec{j}_{\bar{\mathcal{S}_0}}]$ is given by
\begin{align}
V_{\chi_{\Tsind}}^{\rm U} &:= 
\begin{cases}
E_{\chi_{\Tsind}^2}^{\rm U}, & \text{if } 0\in [E_{\chi_{\Tsind}}^{\rm L},E_{\chi_{\Tsind}}^{\rm U}],
\\[2mm]
E_{\chi_{\Tsind}^2}^{\rm U} -\min\set{(E_{\chi_{\Tsind}}^{\rm L})^2,(E_{\chi_{\Tsind}}^{\rm U})^2}, & \text{if } 0\notin [E_{\chi_{\Tsind}}^{\rm L},E_{\chi_{\Tsind}}^{\rm U}],
\end{cases}
\\
&=
E_{\chi_{\Tsind}^2}^{\rm U} - \left[ \max(0, E_{\chi_{\Tsind}}^{\rm L}) + \min(0, E_{\chi_{\Tsind}}^{\rm U}) \right]^2,
\end{align}
where
\begin{equation}
\begin{split}
E_{\chi_{\Tsind}}^{\rm L}
&\leq\ptrash\min_{G\in\mathcal{S}_{\rm const}}\left[\Tr\set{\Tobs{A} \rho_A}\right],
\\
E_{\chi_{\Tsind}^2}^{\rm U}
&\geq\ptrash\max_{G\in\mathcal{S}_{\rm const}}\left[\Tr\set{(\Tobs{A})^2 \rho_A}\right].
\end{split}
\end{equation}
As for $c$, we can fix it to 
$c^{\rm U} = \max\set{\omega_{\rm max}-E_{\chi_{\Tsind}}^{\rm L},E_{\chi_{\Tsind}}^{\rm U}-\omega_{\rm min}} \geq \max\abs{\rv{\chi}^{u}_{\Tsind}-\mathbb{E}[\rv{\chi}^{u}_{\Tsind}| \vec{j}_{\bar{\mathcal{S}_0}}]}$.

Given $E_{\chi_{\Tsind}}^{\rm U}$ and $V_{\chi_{\Tsind}}^{\rm U}$, it follows that
\begin{equation}
\mathbb{E}[\rv{S}_0| \vec{j}_{\bar{\mathcal{S}_0}}] + N_0\Delta_{\rm B}(N_0,v_0,c,\epsilon) \leq N_0 E_{\chi_{\Tsind}}^{\rm U} + \bar{N}_L\Delta_{\rm B}(\bar{N}_L,V_{\chi_{\Tsind}}^{\rm U},c^{\rm U},\epsilon)=:\delta^{\rm U}_0,
\end{equation}
where $\bar{N}_L=\lceil N/(L+1) \rceil$, and we have used that $N_0\leq\bar{N}_L$ and that $n\Delta_{\rm B}(n,v,c,\epsilon)$ is non-decreasing in $n$. Therefore,
\begin{equation}\label{eq:PS02}
    \Pr[\rv{S}_0\geq \delta^{\rm U}_0 | \vec{j}_{\bar{\mathcal{S}_0}}]\leq \Pr[\rv{S}_0\geq \delta^{\rm U}_{0,\vec{j}_{\bar{\mathcal{S}_0}}} | \vec{j}_{\bar{\mathcal{S}_0}}] \leq \epsilon,
\end{equation}
and then~\cref{eq:PS00} can be obtained via~\cref{eq:PS01,eq:PS02}.

Finally, one can obtain analogous probabilistic bounds $\Pr[\rv{S}_m\geq\delta^{\rm U}_m]\leq \epsilon$ for all the other sets $\mathcal{S}_{m}$, and by virtue of the union bound, we have that
\begin{equation}
 \Pr[\rv{M}_{\Tsind}\geq \sum_{m=0}^L\delta^{\rm U}_m=:M_{\Tsind}^{\rm U}] \leq \epsproB,
\end{equation}
where $\epsproB :=(L+1)\epsilon$ and $M_{\Tsind}^{\rm U}:= N E_{\chi_{\Tsind}}^{\rm U}+(L+1)\bar{N}_L\Delta_{\rm B}(\bar{N}_L,V_{\chi_{\Tsind}}^{\rm U},c^{\rm U},\frac{\epsproB}{L+1})$, since $\sum_{m=0}^{L}N_m=N$ with $N_m:=\abs{\mathcal{S}_m}$.
This ends the proof.
\end{proof}

Finally, below we include~\cref{thm:correlations_decoy}, which is a generalization of~\cref{thm:correlations} applicable to the decoy-state setting with bit/basis and intensity correlations:

\begin{theorem}
\label{thm:correlations_decoy}
Consider the virtual estimation protocol (VEP) introduced in Box~\ref{box:VEPdec}, and let $\ket{\Psi^{\epsleak}}_{A_1^NR_1^NC_1^N}$, $N$, $\chiT{u}$, $\rv{M}_{\Tsind}$, $\alpha$, and $\Tobsn{A_u}$ be as defined there, where $\ket{\Psi^{\epsleak}}_{A_1^NR_1^NC_1^N}$ is given in~\cref{eq:global_correlated_state_decoy_main} and depends on a certain $\epsleak\geq0$, a set of reference states $\set{\ket{\phi_{n,i}}}_{n,i}$ (with $n\in\set{0,1,\dots,n_{\rm cut}}$ and $i\in\set{0,1,\dots,\nstates-1}$), setting probabilities $\set{\piA{i}}_{i}$, photon-number probabilities $\set{p_{n|i,i^{u-1}_{u-L}}}_{n,i,i^{u-1}_{u-L}}$, and the correlation length $L$.
Let $E_{\chi_{\Tsind}\rvert \mathcal{L}_{u}}^{\rm L}$, $E_{\chi_{\Tsind}\rvert \mathcal{L}_{u}}^{\rm U}$ and $E_{\chi_{\Tsind}^2\rvert \mathcal{L}_{u}}^{\rm U}$ satisfy
\begin{equation}\label{eq:worst_case_bounds_Theorem2}
\begin{split}
E_{\chi_{\Tsind}\rvert \mathcal{L}_{u}}^{\rm L}
&\leq\ptrash\sum_{n=0}^{n_{\rm cut}}
p_{n|\mathcal{L}_{u}^{-}}\min_{G^n\in\mathcal{S}_{\rm const}^{n}}\left[\Tr\set{\Tobsn{A_u} \rho_{A_u}^{n,\varepsilon',\mathcal{L}_u}}\right]
+\ptrash\sum_i\lambda_{i,\infty}^{*}t_{i,\infty}^{\mathcal{L}_{u}^{-}},
\\
E_{\chi_{\Tsind}\rvert \mathcal{L}_{u}}^{\rm U}
&\geq\ptrash\sum_{n=0}^{n_{\rm cut}}
p_{n|\mathcal{L}_{u}^{-}}\max_{G^n\in\mathcal{S}_{\rm const}^{n}}\left[\Tr\set{\Tobsn{A_u} \rho_{A_u}^{n,\varepsilon',\mathcal{L}_u}}\right]
+\ptrash\sum_i\lambda_{i,\infty}^{*}t_{i,\infty}^{\mathcal{L}_{u}^{-}},
\\
E_{\chi_{\Tsind}^2\rvert \mathcal{L}_{u}}^{\rm U}
&\geq\ptrash\sum_{n=0}^{n_{\rm cut}}
p_{n|\mathcal{L}_{u}^{-}}\max_{G^n\in\mathcal{S}_{\rm const}^{n}}\left[\Tr\set{(\Tobsn{A_u})^2 \rho_{A_u}^{n,\varepsilon',\mathcal{L}_u}}\right]
+\ptrash\sum_i(\lambda_{i,\infty}^{*})^2t_{i,\infty}^{\mathcal{L}_{u}^{-}},
\end{split}
\end{equation}
where $\mathcal{L}_u=j^{u+L}_{u+1}j^{u-1}_{u-L}$ and $\mathcal{L}_u^{-}=j^{u-1}_{u-L}$ represent particular sequences of setting choices; $t_{i,\infty}^{\mathcal{L}_{u}^{-}}:=\piA{i}\left(1-\sum_{n=0}^{n_{\rm cut}}p_{n|i,\mathcal{L}_{u}^{-}}\right)$, 
with $p_{i|n,\mathcal{L}_{u}^{-}}=\piA{i} p_{n|i,\mathcal{L}_u^{-}}/p_{n|\mathcal{L}_u^-}$;
$p_{n|\mathcal{L}_{u}^{-}}=\sum_i \piA{i}p_{n|i,\mathcal{L}_{u}^{-}}$ is the probability that Alice emits an $n$-photon state given the sequence of previous settings $\mathcal{L}_u^{-}$; the matrices
\begin{equation}\label{eq:Gmatrix_decoy}
G^n:=
\begin{bmatrix}
G_{\phi}^n & G_{\rm c}^n     \\
(G_{\rm c}^{n})^{\dagger}   & G_{\perp}^n   \\
\end{bmatrix}
\end{equation}
are $2\nstates \times 2\nstates$ Hermitian block matrices built from the $\nstates\times\nstates$ matrices $G_{\phi}^n$, $G_{\rm c}^n$, and $G_{\perp}^n$; 
each $\rho_{A_u}^{n,\varepsilon',\mathcal{L}_u}$ is a density matrix whose entries depend linearly on those of $G^n$ as $[\rho_{A_u}^{n,\varepsilon',\mathcal{L}_u}]_{ij} := \sqrt{p_{i|n,\mathcal{L}_{u}^{-}}p_{j|n,\mathcal{L}_{u}^{-}}}\set{(1-\varepsilon')[G_{\phi}^n]_{ji} +\varepsilon'[G_{\perp}^n]_{ji} + \sqrt{(1-\varepsilon')\varepsilon'} ([G_{\rm c}^n]_{ij}^{*} + [G_{\rm c}^n]_{ji})}$ with $\epsleak':=1-(1-\epsleak)^{L+1}$; and  
$\mathcal{S}_{\rm const}^{n}$ is a convex set defined by the constraints
\begin{equation}\label{eq:SDP_gram_DS2}
    \begin{split}
        \quad & G^n \succeq 0, \qquad (n\in\set{0,1,\dots,n_{\rm cut}})
        \\
        \quad & [G_{\phi}^n]_{ij}=\braket{\phi_{n,i}}{\phi_{n,j}}\xi_{i,j,\mathcal{L}_u} \qquad (n\in\set{0,1,\dots,n_{\rm cut}}, \forall i,j),
        \\
        \quad & [G_{\rm c}^n]_{ii}=0 \qquad (n\in\set{0,1,\dots,n_{\rm cut}}, \forall i),
        \\
        \quad & [G_{\perp}^n]_{ii}=1\qquad (n\in\set{0,1,\dots,n_{\rm cut}}, \forall i),
    \end{split}
\end{equation}
where $\xi_{i,j,\mathcal{L}_u}:=\prod_{k=1}^L\sum_{n_{u+k}}\sqrt{p_{n_{u+k}|j^{u+k}_{u+1} i j^{u-1}_{u-L+k}}p_{n_{u+k}|j^{u+k}_{u+1} j j^{u-1}_{u-L+k}}}$.

Then, for any $\epsproB>0$,
\begin{align} \label{eq:MT_bounds_theorem2}
\Pr[\rv{M}_{\Tsind}\geq M_{\Tsind}^{\rm U}] &\leq \epsproB,
\end{align}
where
\begin{align}
M_{\Tsind}^{\rm U}
:=& 
(L+1)\bar{N}_L\max_{\mathcal{L}_u}\left[E_{\chi_{\Tsind}\rvert\mathcal{L}_u}^{\rm U}+\Delta_{\rm B}(\bar{N}_L,V_{\chi_{\Tsind}\rvert\mathcal{L}_u}^{\rm U},c^{\rm U}_{\rvert\mathcal{L}_u},\frac{\epsproB}{L+1})\right],
\\
c^{\rm U}_{\rvert\mathcal{L}_u}  \label{eq:Thm2c}
:= &
\max\set{\omega_{\rm max}-E_{\chi_{\Tsind}\rvert\mathcal{L}_u}^{\rm L},E_{\chi_{\Tsind}\rvert\mathcal{L}_u}^{\rm U}-\omega_{\rm min}},
\\
V_{\chi_{\Tsind}\rvert\mathcal{L}_u}^{\rm U} := \label{eq:Thm2Variance}
& 
\begin{cases}
E_{\chi_{\Tsind}^2\rvert\mathcal{L}_u}^{\rm U}, &  \text{if } 0\in [E_{\chi_{\Tsind}\rvert\mathcal{L}_u}^{\rm L},E_{\chi_{\Tsind}\rvert\mathcal{L}_u}^{\rm U}],
\\[2mm]
E_{\chi_{\Tsind}^2\rvert\mathcal{L}_u}^{\rm U} -\min\set{(E_{\chi_{\Tsind}\rvert\mathcal{L}_u}^{\rm L})^2,(E_{\chi_{\Tsind}\rvert\mathcal{L}_u}^{\rm U})^2}, &  \text{if } 0\notin [E_{\chi_{\Tsind}\rvert\mathcal{L}_u}^{\rm L},E_{\chi_{\Tsind}\rvert\mathcal{L}_u}^{\rm U}],
\end{cases}
\end{align}
and $\Delta_{\rm B}(n,v,c,\epsilon)$ is given by~\cref{eq:DeltaB}, $\omega_{\min}$ ($\omega_{\max}$) denotes the minimum (maximum) possible value of the RVs $\chiT{u}$, and $\bar{N}_{L}:=\lceil N/(L+1) \rceil$. Here, for simplicity, we assume that $N$ is a multiple of $L+1$.
\end{theorem}

\begin{proof}

The proof is similar to that of~\cref{thm:correlations}. 
Consider the sets defined in~\cref{eq:set_Sm_Thm1}, such that $\rv{M}_{\Tsind}=\sum_{u=1}^{N}\chiT{u}=\sum_{m}\sum_{u\in\mathcal{S}_m}\chiT{u}$. We shall bound $\rv{M}_{\Tsind}$ by bounding each of the sums $\rv{S}_m:=\sum_{u\in\mathcal{S}_m}\chiT{u}$ individually. As in the proof of~\cref{thm:correlations}, the first and last $L$ rounds can be treated on the same footing as the others by appending fictitious rounds with fixed settings, in which case the photon-number probabilities and states conditioned on a history that includes fictitious rounds are simply those conditioned on the corresponding truncated history.

For this, let us focus first on the sum $\rv{S}_0$.
By using arguments analogous to those in the proof of~\cref{thm:correlations}, we consider a variant of the VEP in which, in the rounds not belonging to $\mathcal{S}_0$---including those tagged as \LC{}---Alice measures her systems $A_1^N$ in the basis $\set{\ket{i}}_i$, observing the outcome $\vec{j}_{\bar{\mathcal{S}_0}}:=\bigotimes_{u\notin \mathcal{S}_0}\ket{j_u}_{A_u}$. In this scenario, the entangled state of the remaining systems conditioned on the outcome $\vec{j}_{\bar{\mathcal{S}_0}}$ can be written as
\begin{equation}\label{eq:cond_state_dec2}
\begin{split}
\ket*{\Psi_{|\vec{j}_{\bar{\mathcal{S}_0}}}}
&=\bigotimes_{u\in \mathcal S_0} 
\left(
\sum_{n_u,i_u} 
\sqrt{\piA{i_u}p_{n_u|i_u,\mathcal{L}_{u}^{-}}}
\ket{n_u}_{R_u}\ket{i_u}_{A_u}
\ket*{\psi_{n_{u},i_u j_{u-L}^{u-1}}^{\epsleak}}_{C_{u}}
\bigotimes_{k=1}^{L}\ket*{\varphi_{j_{u+1}^{u+k}i_u j_{u-L+k}^{u-1}}}_{R_{u+k}C_{u+k}}
\right)
\\
&=\bigotimes_{u\in \mathcal S_0} \ket*{\Psi_{|\mathcal{L}_{u}}^u}_{A_uR_{u}^{u+L}C_u^{u+L}},
\end{split}
\end{equation}
where the states $\ket*{\psi_{n_{u},i_{u-L}^{u}}^{\epsleak}}_{C_{u}}$ are defined in~\cref{eq:states_decoy}, and
\begin{equation}
\ket*{\varphi_{i_{u-L}^{u}}}_{R_uC_{u}} = \sum_{n_u}\sqrt{p_{n_u|i_{u-L}^{u}}}
\ket{n_u}_{R_u}
\ket*{\psi_{n_{u},i_{u-L}^{u}}^{\epsleak}}_{C_{u}},
\end{equation}
$\mathcal{L}_{u}:=j^{u+L}_{u+1}j^{u-1}_{u-L}$ and $\mathcal{L}_{u}^{-}:=j^{u-1}_{u-L}$. From~\cref{eq:cond_state_dec2} it is clear that 
$\Pr[\rv{\chi}^{u}_{\Tsind}=\omega,\rv{\chi}^{v}_{\Tsind}=\omega'|\vec{j}_{\bar{\mathcal{S}_0}}] =\Pr[\rv{\chi}^{u}_{\Tsind}=\omega|\vec{j}_{\bar{\mathcal{S}_0}}]\Pr[\rv{\chi}^{v}_{\Tsind}=\omega'|\vec{j}_{\bar{\mathcal{S}_0}}]$ for any $u,v\in\mathcal{S}_0$ and for any outcomes $\omega$ and $\omega'$. That is, all the random variables in the set $\set{\rv{\chi}^{u}_{\Tsind}}_{u\in\mathcal{S}_0}$ are independent conditioned on the outcomes $\vec{j}_{\bar{\mathcal{S}_0}}$, and one can follow similar arguments as those used in the proof of~\cref{thm:correlations} to obtain a probabilistic bound on $\rv{S}_0$.

Specifically, we aim to prove that
\begin{equation}\label{eq:PS01_dec}
\begin{split}
\Pr[\rv{S}_0\geq\delta^{\rm U}_0]&=\sum_{\vec{j}_{\bar{\mathcal{S}_0}}}\Pr[\vec{j}_{\bar{\mathcal{S}_0}}]\Pr[\rv{S}_0\geq\delta^{\rm U}_0 | \vec{j}_{\bar{\mathcal{S}_0}}]
\leq 
\epsilon
\end{split}
\end{equation}
holds for certain $\delta_0^{\rm U}$ and $\epsilon$. 
As the random variables $\set{\rv{\chi}^{u}_{\Tsind}}_{u\in\mathcal{S}_0}$ are independent, one can invoke \cref{thm:bernstein} to find a bound on $\Pr[\rv{S}_0\geq\delta^{\rm U}_0 | \vec{j}_{\bar{\mathcal{S}_0}}]$ (see Appendix~\ref{app:finite-key}). In particular, we have that 
\begin{equation}\label{eq:bound_dependent_j_dec}
    \Pr[\rv{S}_0\geq \mathbb{E}[\rv{S}_0| \vec{j}_{\bar{\mathcal{S}_0}}] + N_0\Delta_{\rm B}(N_0,v_0,c,\epsilon) \rvert \vec{j}_{\bar{\mathcal{S}_0}}]\leq \epsilon,
\end{equation}
where the function $\Delta_{\rm B}(n,v,c,\epsilon)$ is given in~\cref{eq:DeltaB}, 
$N_0:=\abs{\mathcal{S}_0}$, $\epsilon>0$, $v_0:=\frac{1}{N_0}\sum_{u\in\mathcal{S}_0}{\rm Var}\left[\rv{\chi}^{u}_{\Tsind} | \vec{j}_{\bar{\mathcal{S}_0}}\right]$, and $c\geq \max\abs{\rv{\chi}^{u}_{\Tsind}-\mathbb{E}[\rv{\chi}^{u}_{\Tsind} | \vec{j}_{\bar{\mathcal{S}_0}}]}$ for all $u\in\mathcal{S}_0$. 
Given $\delta^{\rm U}_{0}\big\rvert_{\vec{j}_{\bar{\mathcal{S}_0}}}:=\mathbb{E}[\rv{S}_0 | \vec{j}_{\bar{\mathcal{S}_0}}] + N_0\Delta_{\rm B}(N_0,v_0,c,\epsilon)$, the next step is to obtain an upper bound that is independent of $\vec{j}_{\bar{\mathcal{S}_0}}$ by considering a worst-case scenario---i.e., we aim to find a $\delta^{\rm U}_0$ such that $\delta^{\rm U}_{0}\big\rvert_{\vec{j}_{\bar{\mathcal{S}_0}}}\leq \delta^{\rm U}_0$ for any $\vec{j}_{\bar{\mathcal{S}_0}}$. 
For this, first we employ SDP techniques~\cite{curras2023security2} to compute worst-case bounds satisfying 
$c^{\rm U}_{\rvert\mathcal{L}_u}\geq \max\abs{\rv{\chi}^{u}_{\Tsind}-\mathbb{E}[\rv{\chi}^{u}_{\Tsind}| \mathcal{L}_u]}$,  
$E_{\chi_{\Tsind}\rvert\mathcal{L}_u}^{\rm U}\geq\mathbb{E}[\rv{\chi}^{u}_{\Tsind}|\mathcal{L}_u]$ 
and 
$V_{\chi_{\Tsind}\rvert\mathcal{L}_u}^{\rm U}\geq {\rm Var}[\rv{\chi}^{u}_{\Tsind}|\mathcal{L}_u]$ that are valid for any $\mathcal{L}_u$ and $u\in\mathcal{S}_0$, where
\begin{equation}
\begin{split}
\mathbb{E}\left[\rv{\chi}^{u}_{\Tsind}|\mathcal{L}_u\right]
&=
\ptrash
\bra{\Psi_{|\mathcal{L}_u}^u}
\Tobs{R_uA_u}
\ket{\Psi_{|\mathcal{L}_u}^u}_{A_uR_u^{u+L}C_u^{u+L}},
\\
{\rm Var}\left[\rv{\chi}^{u}_{\Tsind}|\mathcal{L}_u\right]
&=
\ptrash
\bra{\Psi_{|\mathcal{L}_u}^u}
\left(\Tobs{R_uA_u}\right)^2
\ket{\Psi_{|\mathcal{L}_u}^u}_{A_uR_u^{u+L}C_u^{u+L}}
-
\left(\mathbb{E}\left[\rv{\chi}^{u}_{\Tsind}|\mathcal{L}_u\right]\right)^2,
\end{split}
\end{equation}
and $\ket*{\Psi_{|\mathcal{L}_u}^u}$ is given in~\cref{eq:cond_state_dec2}. Note that, due to~\cref{eq:states_decoy}, the states
\begin{equation}
\ket*{\Vec{\psi}_{n_u,i_u,\mathcal{L}_u}}_{R^{u+L}_{u+1}C_u^{u+L}}
:=
\ket*{\psi_{n_{u},i_u j_{u-L}^{u-1}}^{\epsleak}}_{C_{u}}
\bigotimes_{k=1}^{L}\ket*{\varphi_{j_{u+1}^{u+k}i_u j_{u-L+k}^{u-1}}}_{R_{u+k}C_{u+k}}
\end{equation}
that appear in~\cref{eq:cond_state_dec2} admit the decomposition
\begin{equation}\label{eq:decomposition_corr-decoy}
\begin{split}
\ket*{\Vec{\psi}_{n_u,i_u,\mathcal{L}_u}}_{R^{u+L}_{u+1}C_u^{u+L}}
= &
\sqrt{1-\epsleak'}\ket*{\Vec{\phi}_{n_u,i_u,\mathcal{L}_u}}_{R^{u+L}_{u+1}C_u^{u+L}}
+\sqrt{\epsleak'}\ket*{\Vec{\phi}_{n_u,i_u,\mathcal{L}_u}^{\perp}}_{R^{u+L}_{u+1}C_u^{u+L}},
\end{split}
\end{equation}
where $\epsleak'=1-(1-\epsleak)^{L+1}\approx (L+1)\epsleak$; the states 
$$\ket*{\Vec{\phi}_{n_u,i_u,\mathcal{L}_u}}_{R^{u+L}_{u+1}C_u^{u+L}}=
\ket*{\phi_{n_u,i_{u}}}_{C_{u}}\bigotimes_{k=1}^{L}\sum_{n_{u+k}}
\sqrt{p_{n_{u+k}|j^{u+k}_{u+1} i_u j^{u-1}_{u-L+k}}}
\ket*{n_{u+k}}_{R_{u+k}}
\ket*{\phi_{n_{u+k},j_{u+k}}}_{C_{u+k}}$$
satisfy $\braket*{\Vec{\phi}_{n_u,i_u,\mathcal{L}_u}}{\Vec{\phi}_{n_u',i_u',\mathcal{L}_u}}=\braket*{\phi_{n_u,i_u}}{\phi_{n_u',i_u'}}\xi_{i_u,i_u',\mathcal{L}_u}$ for any $n_u$, $i_u$, $n_u'$, and $i_u'$, with $\xi_{i_u,i_u',\mathcal{L}_u}:=\prod_{k=1}^L\sum_{n_{u+k}}\sqrt{p_{n_{u+k}|j^{u+k}_{u+1} i_u j^{u-1}_{u-L+k}}p_{n_{u+k}|j^{u+k}_{u+1} i_u' j^{u-1}_{u-L+k}}}$; and the state $\ket*{\Vec{\phi}_{n_u,i_u,\mathcal{L}_u}^{\perp}}_{R^{u+L}_{u+1}C_u^{u+L}}$ is orthogonal to $\ket*{\Vec{\phi}_{n_u,i_u,\mathcal{L}_u}}_{R^{u+L}_{u+1}C_u^{u+L}}$ for any $n_u$ and $i_u$.

Then, for each value $n_u\in\set{0,1,\dots,n_{\rm cut}}$ we can construct a $2n_A\times2n_A$ Gram matrix $G^n$ filled with the inner products of the states $\ket*{\Vec{\phi}_{n_u,0,\mathcal{L}_u}},\dots,\ket*{\Vec{\phi}_{n_u,n_A-1,\mathcal{L}_u}},\ket*{\Vec{\phi}_{n_u,0,\mathcal{L}_u}^{\perp}},\dots,\ket*{\Vec{\phi}_{n_u,n_A-1,\mathcal{L}_u}^{\perp}}$ for a given $\mathcal{L}_u$. Precisely, we define $G^n$ as in~\cref{eq:Gmatrix_decoy}, where each of the $n_A\times n_A$ sub-blocks is defined by its entries as $[G_{\phi}^n]_{ij}=\braket*{\Vec{\phi}_{n,i,\mathcal{L}_u}}{\Vec{\phi}_{n,j,\mathcal{L}_u}}$, 
$[G_{\rm c}^n]_{ij}=\braket*{\Vec{\phi}_{n,i,\mathcal{L}_u}}{\Vec{\phi}_{n,j,\mathcal{L}_u}^{\perp}}$, and 
$[G_{\perp}^n]_{ij}=\braket*{\Vec{\phi}_{n,i,\mathcal{L}_u}^{\perp}}{\Vec{\phi}_{n,j,\mathcal{L}_u}^{\perp}}$. 
Thus, we can define some bounds $E_{\chi_{\Tsind}\rvert \mathcal{L}_{u}}^{\rm L}$, $E_{\chi_{\Tsind}\rvert \mathcal{L}_{u}}^{\rm U}$ and $E_{\chi_{\Tsind}^2\rvert \mathcal{L}_{u}}^{\rm U}$ satisfying~\cref{eq:worst_case_bounds_Theorem2},
where $t_{i,\infty}^{\mathcal{L}_{u}^{-}}=\piA{i}\left(1-\sum_{n=0}^{n_{\rm cut}}p_{n|i,\mathcal{L}_{u}^{-}}\right)$. Due to~\cref{eq:decomposition_corr-decoy}, the marginal state $\rho_{A_u}^{n,\epsleak',\mathcal{L}_u}=\frac{1}{p_{n|\mathcal{L}_{u}^{-}}}\Tr_{R_u^{u+L}C_u^{u+L}}\set{\Pi^n_{R_u}\dyad*{\Psi_{|\mathcal{L}_u}^u}_{A_uR_u^{u+L}C_u^{u+L}}}$ of Alice's ancilla system $A_u$ associated with the $n$-photon transmitted state, with $\Pi^n_{R_u}=\dyad{n}_{R_u}$, can be written as a function of the elements of $G^n$ as
\begin{equation}
\begin{split}
[\rho_{A_u}^{n,\epsleak',\mathcal{L}_u}]_{ij} =& \sqrt{p_{i|n,\mathcal{L}_{u}^{-}}p_{j|n,\mathcal{L}_{u}^{-}}}\braket*{\Vec{\psi}_{n,j,\mathcal{L}_u}}{\Vec{\psi}_{n,i,\mathcal{L}_u}}
\\
=& 
\sqrt{p_{i|n,\mathcal{L}_{u}^{-}}p_{j|n,\mathcal{L}_{u}^{-}}}\big[
(1-\epsleak')[G_{\phi}^n]_{ji} + \epsleak'[G_{\perp}^n]_{ji}
+ \sqrt{(1-\epsleak')\epsleak'} ([G_{\rm c}^n]_{ij}^{*} + [G_{\rm c}^n]_{ji})
\big],
\end{split}
\end{equation}
and $\mathcal{S}_{\rm const}^{n}$ denotes the set of constraints provided in~\cref{eq:SDP_gram_DS2}.
An upper bound on the variance $V_{\chi_{\Tsind}\rvert \mathcal{L}_{u}}^{\rm U}\geq {\rm Var}[\chiT{u}|\mathcal{L}_{u}]$ is then given by~\cref{eq:Thm2Variance}. As for the quantity
$c^{\rm U}_{\rvert \mathcal{L}_{u}} \geq \max\abs{\rv{\chi}^{u}_{\Tsind}-\mathbb{E}[\rv{\chi}^{u}_{\Tsind}|\mathcal{L}_{u}]}$, we can fix it to~\cref{eq:Thm2c}.

Note that the bounds $E_{\chi_{\Tsind}\rvert \mathcal{L}_{u}}^{\rm U}$, $V_{\chi_{\Tsind}\rvert \mathcal{L}_{u}}^{\rm U}$ and $c^{\rm U}_{\rvert \mathcal{L}_{u}}$ still depend on the local context $\mathcal{L}_{u}$. To remove this dependence, we take again the worst-case bound by maximizing over all possible $\mathcal{L}_{u}$, such that
\begin{equation}\label{eq:thm_decoy_boundSDP1}
\mathbb{E}[\rv{S}_0| \vec{j}_{\bar{\mathcal{S}_0}}] + N_0\Delta_{\rm B}(N_0,v_0,c,\epsilon) 
\leq 
\max_{\mathcal{L}_{u}}\left[\bar{N}_L E_{\chi_{\Tsind}\rvert \mathcal{L}_{u}}^{\rm U} + \bar{N}_L\Delta_{\rm B}(\bar{N}_L,V_{\chi_{\Tsind}\rvert \mathcal{L}_{u}}^{\rm U},c^{\rm U}_{\rvert \mathcal{L}_{u}},\epsilon)\right]=:\delta^{\rm U}_0,
\end{equation}
holds for  any $\vec{j}_{\bar{\mathcal{S}_0}}$, where $\bar{N}_L=\lceil N/(L+1) \rceil$. Therefore, it follows from~\cref{eq:bound_dependent_j_dec,eq:thm_decoy_boundSDP1} that
\begin{equation}\label{eq:PS02_decoy}
    \Pr[\rv{S}_0\geq \delta^{\rm U}_0 | \vec{j}_{\bar{\mathcal{S}_0}}]
    \leq \epsilon,
\end{equation}
for any $\vec{j}_{\bar{\mathcal{S}_0}}$, and so~\cref{eq:PS01_dec} follows straightforwardly.

Finally, one can obtain analogous probabilistic bounds $\Pr[\rv{S}_m\geq\delta^{\rm U}_m]\leq \epsilon$ for all the other sets $\mathcal{S}_{m}$, and by virtue of the union bound, we have that
\begin{equation}
 \Pr[\rv{M}_{\Tsind}\geq \sum_{m=0}^L\delta^{\rm U}_m=:M_{\Tsind}^{\rm U}] \leq \epsproB,
\end{equation}
where we set $\epsproB :=(L+1)\epsilon$ and $M_{\Tsind}^{\rm U}:= (L+1)\max_{\mathcal{L}_{u}}\left[\bar{N}_LE_{\chi_{\Tsind}\rvert \mathcal{L}_{u}}^{\rm U}+\bar{N}_L\Delta_{\rm B}(\bar{N}_L,V_{\chi_{\Tsind}\rvert \mathcal{L}_{u}}^{\rm U},c^{\rm U}_{\rvert \mathcal{L}_{u}},\frac{\epsproB}{L+1})\right]$.  
This ends the proof.

\end{proof}

\begin{remark}\label{rem:remark_extra_constraints}The constraints in~\cref{eq:SDP_gram_DS2} are the minimal ones implied by the assumption in~\cref{eq:assumption1thm_decoy}. Importantly, however, additional knowledge about the source can be readily incorporated as extra constraints on the matrices $G^n$, resulting in tighter bounds in~\cref{eq:worst_case_bounds_Theorem2} and thus in a higher secret-key rate. For instance, if the $n$-photon states are known to be independent of the current and previous intensity settings---that is, if side-channel leakage and correlations only affect the bit/basis information---then the states $\ket*{\phi_{n_u,i^{u}_{u-L}}^{\perp}}$ coincide for all values of the intensity settings, which translates into linear equality constraints on the entries of $G^n$. In particular, when $L=0$, this amounts to imposing $[G_{\perp}^n]_{ij}=1$ and $[G_{\rm c}^n]_{ki}=[G_{\rm c}^n]_{kj}$ for all $k$ and all pairs $i=(a,\mu)$ and $j=(a,\mu')$.
\end{remark}

\section{Model for intensity correlations}\label{app:model_correlations}
To illustrate our security analysis for decoy-state QKD schemes, we consider a particular model of intensity correlations. In particular, we focus on a typical scenario in which the actual intensity $I_u$ of the $u$-th transmitted signal satisfies $I_u(i^{u}_{u-L}) = I_u(\mu^{u}_{u-L})$. That is, in our simulations, we assume that only the previous intensity settings---and not the bit/basis settings---may affect the actual intensity of the current emitted signal. This is merely a particular modeling choice; more general models, in which previous bit and basis settings also influence the intensity of the current signal, can be considered in a straightforward manner.

Specifically, we define the average intensity as
\begin{equation}\label{eq:Iu_av_definition}
    \Iav{\mu}:=\sum_{\mu^{u-1}_{u-L}} p_{\mu^{u-1}_{u-L}}\, I_u(\mu,\mu^{u-1}_{u-L}).
\end{equation}
Then, we model the actual intensity as a history-dependent perturbation of this average, namely,
\begin{equation}
    I_u(\mu,\mu^{u-1}_{u-L}) = \Iav{\mu}\bigl(1+\delta(\mu^{u-1}_{u-L})\bigr),
\end{equation}
where $\delta(\mu^{u-1}_{u-L})$ is a scalar correction term accounting for the correlations that must satisfy $\delta(\mu^{u-1}_{u-L})\geq -1$ to ensure that the actual intensity is positive. To compute this latter quantity, we assume that the highest intensity setting ($\mu=0$) slightly increases the intensity of subsequent rounds relative to the average. Conversely, every other intensity setting slightly reduces the intensity of subsequent rounds relative to the average; for simplicity, we assume this reduction to be identical across those settings. Thus, the dependence on each previous choice is modeled only through whether that choice was the highest intensity or not. Precisely, we write
\begin{equation}
    \delta(\mu^{u-1}_{u-L}) = \sum_{l=1}^{L} \epsilon_{\rm ic}^{\,l}\, z(\mu_{u-l}),
\end{equation}
where
\begin{equation}
    \epsilon_{\rm ic}^{\,l} = \epsilon_{\rm ic}^{\,1} e^{-\zeta(l-1)}
\end{equation}
for some $\zeta>0$, where $\epsilon_{\rm ic}^{\,1} > 0$ quantifies the impact of the immediately previous round, $\zeta$ determines the exponential decay of the correlations with the lag $l$, and the function $z(\mu)$ is defined as
\begin{equation}
    z(\mu)=
    \begin{cases}
        1-\sum_a\piA{a,0}, & \text{if } \mu=0,\\[0.5em]
        -\sum_a\piA{a,0}, & \text{if } \mu\neq 0.
    \end{cases}
\end{equation}
Note that, by construction, the random variable $z(\mu)$ has zero mean. Therefore, $\mathbb{E}\!\left[\delta(\mu^{u-1}_{u-L})\right]=0$, which implies $\mathbb{E}_{\mu^{u-1}_{u-L}}\!\left[I_u(\mu,\mu^{u-1}_{u-L})\right] = \Iav{\mu}$. Hence, the model preserves the prescribed average intensity vector while introducing exponentially decaying correlations over the previous $L$ rounds.

We remark that the above model is introduced solely for simulation purposes. The security proof developed here remains valid for other arbitrary forms of correlations.

\section{Normalizing the SDP}\label{app:normalizationSDP}

The SDP in~\cref{eq:primal1} can be re-normalized to make the input statistics of the SDP independent of Alice's setting probabilities $\piA{i}$~\cite{zhou2022numerical}. For this, consider the map $\mathcal{R}_A:\mathcal{L}(\mathcal{H}_A)\to\mathcal{L}(\mathcal{H}_A)$ defined by its action over an operator $\hat{L}_A\in\mathcal{L}(\mathcal{H}_A)$ as
\begin{equation}
\begin{split}
\bra{i}\mathcal{R}_A(\hat{L}_A)\ket{j} &= \frac{1}{\sqrt{\piA{i}\piA{j}}}\bra{i}\hat{L}_A\ket{j},
\\
\bra{i}\mathcal{R}^{-1}_A(\hat{L}_A)\ket{j} &= \sqrt{\piA{i}\piA{j}}\bra{i}\hat{L}_A\ket{j}.
\end{split}
\end{equation}
This allows us to define the re-normalized density matrix $\rho^{\rm R}_{AB}:=\mathcal{R}_A\otimes \mathcal{I}_B\left(\rho_{AB}\right)$ and the re-normalized phase-error operator $\hat{E}_{AB}^{\rm ph, R}:=\mathcal{R}_A^{-1}\otimes \mathcal{I}_B(\hat{E}^{\rm ph}_{AB})$, which satisfies $\hat{E}_{AB}^{\rm ph,R}=\frac{\piA{Z}}{2}\hat{E}_{AB}^{\rm ph}$, with $\piA{Z}:=\piA{0}+\piA{1}=2\piA{0}$. Note that, since $\mathcal{R}_A^{\dagger}=\mathcal{R}_A$, these re-normalized operators satisfy $\Tr\set{\hat{E}^{\rm ph}_{AB}\rho_{AB}}=\Tr\set{\hat{E}_{AB}^{\rm ph,R}\rho^{\rm R}_{AB}}=\frac{\piA{Z}}{2}\Tr\set{\hat{E}_{AB}^{\rm ph}\rho^{\rm R}_{AB}}$, and therefore the SDP introduced in~\cref{eq:primal1} can be reformulated as
\begin{equation}
    \begin{split}
        \max_{\rho_{AB}^{\rm R}} \quad &  \Tr\set{\hat{E}_{AB}^{\rm ph}\rho_{AB}^{\rm R}}\\
        \text{s.t.} \quad & \rho_{AB}^{\rm R}\succeq 0, \\
        \quad & \Tr\set{\hat{T}^k_A \rho_{AB}^{\rm R}} = t_k^{\rm R},\qquad \forall k\\
        \quad & \Tr\set{\hat{Q}^l_{AB} \rho_{AB}^{\rm R}} = q_l^{\rm R},\qquad \forall l
    \end{split}
\end{equation}
where now the statistics $t_k^{\rm R}$ and $q_l^{\rm R}:= \sum_{i,\beta,b}c_{i,\beta,b}^l p_{b|i,\beta}$ no longer depend on the selection probabilities $\piA{i}$.
Therefore, we can run the dual SDP
\begin{equation}
\begin{split}
    \min_{\vos{\Lambda}} \quad &  \sum_l\tilde{\eta}_l q^{\rm gs, R}_l + \sum_k\tilde{\lambda}_k t^{\rm gs, R}_{k}\\
    \text{s.t.} \quad & \hat{E}^{\rm ph}_{AB}  \preceq \sum_l\tilde{\eta}_l\hat{Q}_{AB}^{l}  +\sum_k\tilde{\lambda}_k\hat{T}_{A}^{k}\otimes\identity_B,
\end{split}
\end{equation}
and obtain the solutions $\set{\tilde{\eta}_l^{*}}_l$ and  $\set{\tilde{\lambda}_k^{*}}_k$, such that the operator inequality 
\begin{equation}
    \hat{E}^{\rm ph}_{AB}  \preceq \sum_l\tilde{\eta}_l^{*}\hat{Q}_{AB}^{l}  +\sum_k\tilde{\lambda}_k^{*}\hat{T}_{A}^{k}\otimes\identity_B,
\end{equation}
is tight for a target operator $\rho^{\rm gs, R}_{AB}$. Note that
\begin{equation}
\begin{split}
    \Tr\set{\hat{E}^{\rm ph}_{AB}\rho_{AB}} =\frac{\piA{Z}}{2} 
    \Tr\set{\hat{E}^{\rm ph}_{AB}\rho^{\rm R}_{AB}}  
    & \leq 
    \frac{\piA{Z}}{2}\sum_l\tilde{\eta}_l^{*}\Tr\set{\hat{Q}_{AB}^{l}\rho^{\rm R}_{AB}}  +
    \frac{\piA{Z}}{2}\sum_k\tilde{\lambda}_k^{*}\Tr\set{(\hat{T}_{A}^{k}\otimes\identity_B)\rho^{\rm R}_{AB}}
    \\
    & = \frac{\piA{Z}}{2}\sum_l\tilde{\eta}_l^{*}\Tr\set{\mathcal{R}_A(\hat{Q}_{AB}^{l})\rho_{AB}}  +
    \frac{\piA{Z}}{2}\sum_k\tilde{\lambda}_k^{*}\Tr\set{(\mathcal{R}_A(\hat{T}_{A}^{k})\otimes\identity_B)\rho_{AB}}.
\end{split}
\end{equation}

\bibliographystyle{apsrev4-2}
\bibliography{refs}

@article{Xoel,
  title = {Security of Decoy-State Quantum Key Distribution with Correlated Intensity Fluctuations},
  author = {Sixto, Xoel and Zapatero, V\'{\i}ctor and Curty, Marcos},
  journal = {Physical Review Applied},
  volume = {18},
  issue = {4},
  pages = {044069},
  numpages = {21},
  year = {2022},
  publisher = {American Physical Society},
  doi = {10.1103/PhysRevApplied.18.044069}
}

@article{bernstein1924modification,
  title={On a modification of {C}hebyshev’s inequality and of the error formula of {L}aplace},
  author={Bernstein, Sergei},
  journal={Ann. Sci. Inst. Sav. Ukraine, Sect. Math},
  volume={1},
  number={4},
  pages={38--49},
  year={1924}
}

@book{boucheron2013concentration,
  author    = {Boucheron, St{\'e}phane and Lugosi, G{\'a}bor and Massart, Pascal},
  title     = {Concentration Inequalities: A Nonasymptotic Theory of Independence},
  publisher = {Oxford University Press},
  year      = {2013},
  doi       = {10.1093/acprof:oso/9780199535255.001.0001}
}

@article{tupkary2024phase,
  title={Phase error rate estimation in {QKD} with imperfect detectors},
  author={Tupkary, Devashish and Nahar, Shlok and Sinha, Pulkit and L{\"u}tkenhaus, Norbert},
  journal={Quantum},
  volume={9},
  pages={1937},
  year={2025},
  publisher={Verein zur F{\"o}rderung des Open Access Publizierens in den Quantenwissenschaften}
}

@article{serfling1974probability,
  title     = {Probability Inequalities for the Sum in Sampling without Replacement},
  author    = {Serfling, R. J.},
  journal   = {The Annals of Statistics},
  volume    = {2},
  number    = {1},
  pages     = {39--48},
  year      = {1974},
  publisher = {Institute of Mathematical Statistics},
  doi       = {10.1214/aos/1176342611}
}

@article{tamaki2014loss,
  title={Loss-tolerant quantum cryptography with imperfect sources},
  author={Tamaki, Kiyoshi and Curty, Marcos and Kato, Go and Lo, Hoi-Kwong and Azuma, Koji},
  journal={Physical Review A},
  volume={90},
  number={5},
  pages={052314},
  year={2014},
  publisher={APS}
}

@article{pereira,
	Author = {Pereira, Margarida and Curty, Marcos and Tamaki, Kiyoshi},
	Journal = {npj Quantum Information},
	Number = {1},
	Pages = {62},
	Title = {Quantum key distribution with flawed and leaky sources},
	Volume = {5},
	Year = {2019}}

@article{pereira2020quantum,
          title={Quantum key distribution with correlated sources},
          author={Pereira, Margarida and Kato, Go and Mizutani, Akihiro and Curty, Marcos and Tamaki, Kiyoshi},
          journal={Science Advances},
          volume={6},
          number={37},
          pages={eaaz4487},
          year={2020},
          publisher={American Association for the Advancement of Science}
        }

@article{azuma,
	Author = {K. Azuma},
	Journal = {Tohoku Mathematical Journal},
	Pages = {357-367},
	Title = {Weighted sums of certain dependent random variables},
	Volume = {19},
	Year = {1967}}

@article{kato,
	Author = {Go Kato},
	Journal = {preprint arXiv:2002.04357},
	Title = {Concentration inequality using unconfirmed knowledge},
	Year = {2020}}

@article{lo2005decoy,
      title={Decoy state quantum key distribution},
      author={Lo, Hoi-Kwong and Ma, Xiongfeng and Chen, Kai},
      journal={Physical Review Letters},
      volume={94},
      number={23},
      pages={230504},
      year={2005},
      publisher={APS}
}

@article{koashi2009simple,
  title={Simple security proof of quantum key distribution based on complementarity},
  author={Koashi, Masato},
  journal={New Journal of Physics},
  volume={11},
  number={4},
  pages={045018},
  year={2009},
  publisher={IOP Publishing}
}

@article{zapatero2021security,
  title={Security of quantum key distribution with intensity correlations},
  author={Zapatero, V{\'\i}ctor and Navarrete, {\'A}lvaro and Tamaki, Kiyoshi and Curty, Marcos},
  journal={Quantum},
  volume={5},
  pages={602},
  year={2021},
  publisher={Verein zur F{\"o}rderung des Open Access Publizierens in den Quantenwissenschaften}
}

@article{curras2023security,
  title={Security of quantum key distribution with imperfect phase randomisation},
  author={Curr{\'a}s-Lorenzo, Guillermo and Nahar, Shlok and L{\"u}tkenhaus, Norbert and Tamaki, Kiyoshi and Curty, Marcos},
  journal={Quantum Science and Technology},
  volume={9},
  number={1},
  pages={015025},
  year={2023},
  publisher={IOP Publishing}
}

@article{xu2020secure,
  title={Secure quantum key distribution with realistic devices},
  author={Xu, Feihu and Ma, Xiongfeng and Zhang, Qiang and Lo, Hoi-Kwong and Pan, Jian-Wei},
  journal={Reviews of Modern Physics},
  volume={92},
  number={2},
  pages={025002},
  year={2020},
  publisher={APS}
}

@article{pirandola2020advances,
  title={Advances in quantum cryptography},
  author={Pirandola, Stefano and Andersen, Ulrik L and Banchi, Leonardo and Berta, Mario and Bunandar, Darius and Colbeck, Roger and Englund, Dirk and Gehring, Tobias and Lupo, Cosmo and Ottaviani, Carlo and others},
  journal={Advances in Optics and Photonics},
  volume={12},
  number={4},
  pages={1012--1236},
  year={2020},
  publisher={Optica Publishing Group}
}

@article{lo2014secure,
  title={Secure quantum key distribution},
  author={Lo, Hoi-Kwong and Curty, Marcos and Tamaki, Kiyoshi},
  journal={Nature Photonics},
  volume={8},
  number={8},
  pages={595--604},
  year={2014},
  publisher={Nature Publishing Group UK London}
}

@article{navarrete2022improved,
  title={Improved finite-key security analysis of quantum key distribution against {T}rojan-horse attacks},
  author={Navarrete, {\'A}lvaro and Curty, Marcos},
  journal={Quantum Science and Technology},
  volume={7},
  number={3},
  pages={035021},
  year={2022},
  publisher={IOP Publishing}
}

@article{bunandar2020numerical,
  title={Numerical finite-key analysis of quantum key distribution},
  author={Bunandar, Darius and Govia, Luke CG and Krovi, Hari and Englund, Dirk},
  journal={npj Quantum Information},
  volume={6},
  number={1},
  pages={104},
  year={2020},
  publisher={Nature Publishing Group UK London}
}

@article{coles2016numerical,
  title={Numerical approach for unstructured quantum key distribution},
  author={Coles, Patrick J and Metodiev, Eric M and L{\"u}tkenhaus, Norbert},
  journal={Nature Communications},
  volume={7},
  number={1},
  pages={11712},
  year={2016},
  publisher={Nature Publishing Group UK London}
}

@article{wang2019characterising,
  title={Characterising the correlations of prepare-and-measure quantum networks},
  author={Wang, Yukun and Primaatmaja, Ignatius William and Lavie, Emilien and Varvitsiotis, Antonios and Lim, Charles Ci Wen},
  journal={npj Quantum Information},
  volume={5},
  number={1},
  pages={17},
  year={2019},
  publisher={Nature Publishing Group UK London}
}

@article{lorente2024quantum,
  title={Quantum key distribution rates from non-symmetric conic optimization},
  author={Lorente, Andr{\'e}s Gonz{\'a}lez and Parellada, Pablo V and Castillo-Celeita, Miguel and Ara{\'u}jo, Mateus},
  journal={Quantum},
  volume={9},
  pages={1657},
  year={2025},
  publisher={Verein zur F{\"o}rderung des Open Access Publizierens in den Quantenwissenschaften}
}

@article{zhou2022numerical,
  title={Numerical method for finite-size security analysis of quantum key distribution},
  author={Zhou, Hongyi and Sasaki, Toshihiko and Koashi, Masato},
  journal={Physical Review Research},
  volume={4},
  number={3},
  pages={033126},
  year={2022},
  publisher={APS}
}

@article{george2021numerical,
  title={Numerical calculations of the finite key rate for general quantum key distribution protocols},
  author={George, Ian and Lin, Jie and L{\"u}tkenhaus, Norbert},
  journal={Physical Review Research},
  volume={3},
  number={1},
  pages={013274},
  year={2021},
  publisher={APS}
}

@article{winick2018reliable,
  title={Reliable numerical key rates for quantum key distribution},
  author={Winick, Adam and L{\"u}tkenhaus, Norbert and Coles, Patrick J},
  journal={Quantum},
  volume={2},
  pages={77},
  year={2018},
  publisher={Verein zur F{\"o}rderung des Open Access Publizierens in den Quantenwissenschaften}
}

@article{kamin2025improved,
  title={Improved finite-size effects in {QKD} protocols with applications to decoy-state {QKD}},
  author={Kamin, Lars and Tupkary, Devashish and L{\"u}tkenhaus, Norbert},
  journal={preprint arXiv:2502.05382},
  year={2025}
}

@article{kamin2024finite,
  title={Finite-size analysis of prepare-and-measure and decoy-state quantum key distribution via entropy accumulation},
  author={Kamin, Lars and Arqand, Amir and George, Ian and L{\"u}tkenhaus, Norbert and Tan, Ernest Y-Z},
  journal={PRX Quantum},
  volume={6},
  number={2},
  pages={020342},
  year={2025},
  publisher={APS}
}

@article{curras2025security,
  title={Security of quantum key distribution with source and detector imperfections through phase-error estimation},
  author={Curr{\'a}s-Lorenzo, Guillermo and Pereira, Margarida and Nahar, Shlok and Tupkary, Devashish},
  journal={preprint arXiv:2507.03549},
  year={2025}
}

@article{curras2023security2,
    title={Security framework for quantum key distribution with imperfect sources},
      author={Curr{\'a}s-Lorenzo, Guillermo and Pereira, Margarida and Kato, Go and Curty, Marcos and Tamaki, Kiyoshi},
      journal={Optica Quantum},
      volume={3},
      number={6},
      pages={525--534},
      year={2025},
      publisher={Optica Publishing Group}
}

@article{kamin2025r,
  title={R{\'e}nyi security framework against coherent attacks applied to decoy-state {QKD}},
  author={Kamin, Lars and Burniston, John and Tan, Ernest Y-Z},
  journal={preprint arXiv:2504.12248},
  year={2025}
}

@article{mannalath2024sharp,
  title   = {Sharp Finite Statistics for Quantum Key Distribution},
  author  = {Mannalath, Vaisakh and Zapatero, V{\'\i}ctor and Curty, Marcos},
  journal = {Physical Review Letters},
  volume  = {135},
  pages   = {020803},
  year    = {2025},
  doi     = {10.1103/l735-x48g}
}

@article{curras2025numerical,
  title={Numerical security analysis for quantum key distribution with partial state characterization},
  author={Curr{\'a}s-Lorenzo, Guillermo and Navarrete, {\'A}lvaro and N{\'u}{\~n}ez-Bon, Javier and Pereira, Margarida and Curty, Marcos},
  journal={Quantum Science and Technology},
  volume={10},
  number={3},
  pages={035031},
  year={2025},
  publisher={IOP Publishing}
}

@article{pereira2024quantum,
  title={Quantum key distribution with unbounded pulse correlations},
  author={Pereira, Margarida and Curr{\'a}s-Lorenzo, Guillermo and Mizutani, Akihiro and Rusca, Davide and Curty, Marcos and Tamaki, Kiyoshi},
  journal={Quantum Science and Technology},
  volume={10},
  number={1},
  pages={015001},
  year={2024},
  publisher={IOP Publishing}
}

@article{curras2021finite,
  title={Finite-key analysis of loss-tolerant quantum key distribution based on random sampling theory},
  author={Curr{\'a}s-Lorenzo, Guillermo and Navarrete, {\'A}lvaro and Pereira, Margarida and Tamaki, Kiyoshi},
  journal={Physical Review A},
  volume={104},
  number={1},
  pages={012406},
  year={2021},
  publisher={APS}
}

@article{navarrete2021practical,
  title={Practical quantum key distribution that is secure against side channels},
  author={Navarrete, {\'A}lvaro and Pereira, Margarida and Curty, Marcos and Tamaki, Kiyoshi},
  journal={Physical Review Applied},
  volume={15},
  number={3},
  pages={034072},
  year={2021},
  publisher={APS}
}

@article{nahar2024postselection,
  title={Postselection technique for optical Quantum Key Distribution with improved de {F}inetti reductions},
  author={Nahar, Shlok and Tupkary, Devashish and Zhao, Yuming and L{\"u}tkenhaus, Norbert and Tan, Ernest Y-Z},
  journal={PRX Quantum},
  volume={5},
  number={4},
  pages={040315},
  year={2024},
  publisher={APS}
}

@article{metger2023security,
  title={Security of quantum key distribution from generalised entropy accumulation},
  author={Metger, Tony and Renner, Renato},
  journal={Nature Communications},
  volume={14},
  number={1},
  pages={5272},
  year={2023},
  publisher={Nature Publishing Group UK London}
}

@article{pereira2025optimal,
  title={Optimal key rates for quantum key distribution with partial source characterization},
  author={Pereira, Margarida and Curr{\'a}s-Lorenzo, Guillermo and Ara{\'u}jo, Mateus},
  journal={preprint arXiv:2510.13085},
  year={2025}
}

@article{tomamichel2017largely,
  title={A largely self-contained and complete security proof for quantum key distribution},
  author={Tomamichel, Marco and Leverrier, Anthony},
  journal={Quantum},
  volume={1},
  pages={14},
  year={2017},
  publisher={Verein zur F{\"o}rderung des Open Access Publizierens in den Quantenwissenschaften}
}

@article{pittaluga2021600,
  title={600-km repeater-like quantum communications with dual-band stabilization},
  author={Pittaluga, Mirko and Minder, Mariella and Lucamarini, Marco and Sanzaro, Mirko and Woodward, Robert I and Li, Ming-Jun and Yuan, Zhiliang and Shields, Andrew J},
  journal={Nature Photonics},
  volume={15},
  number={7},
  pages={530--535},
  year={2021},
  publisher={Nature Publishing Group UK London}
}

@article{uhlmann1976transition,
  title   = {The ``transition probability'' in the state space of a {$*$}-algebra},
  author  = {Uhlmann, Armin},
  journal = {Reports on Mathematical Physics},
  volume  = {9},
  number  = {2},
  pages   = {273--279},
  year    = {1976},
  publisher = {Elsevier}
}

@article{tupkary2026rigorous,
  title={A rigorous and complete security proof of decoy-state {BB84} quantum key distribution},
  author={Tupkary, Devashish and Nahar, Shlok and Arqand, Amir and Tan, Ernest Y-Z and L{\"u}tkenhaus, Norbert},
  journal={preprint arXiv:2601.18035},
  year={2026}
}

@article{agulleiro2025modeling,
  title={Modeling and characterization of arbitrary order pulse correlations for quantum key distribution},
  author={Agulleiro, Ainhoa and Gr{\"u}nenfelder, Fadri and Pereira, Margarida and Curr{\'a}s-Lorenzo, Guillermo and Zbinden, Hugo and Curty, Marcos and Rusca, Davide},
  journal={preprint arXiv:2506.18684},
  year={2025}
}

@article{sixto2025quantum,
  title={Quantum key distribution with imperfectly isolated devices},
  author={Sixto, Xoel and Navarrete, {\'A}lvaro and Pereira, Margarida and Curr{\'a}s-Lorenzo, Guillermo and Tamaki, Kiyoshi and Curty, Marcos},
  journal={Quantum Science and Technology},
  volume={10},
  number={3},
  pages={035034},
  year={2025},
  publisher={IOP Publishing}
}

@article{navarro2025finite,
  title={Finite-size quantum key distribution rates from R{\'e}nyi entropies using conic optimization},
  author={Navarro, Mariana and Lorente, Andr{\'e}s Gonz{\'a}lez and Parellada, Pablo V and Pascual-Garc{\'\i}a, Carlos and Ara{\'u}jo, Mateus},
  journal={preprint arXiv:2511.10584},
  year={2025}
}

@article{hwang2003quantum,
  title={Quantum key distribution with high loss: toward global secure communication},
  author={Hwang, Won-Young},
  journal={Physical Review Letters},
  volume={91},
  number={5},
  pages={057901},
  year={2003},
  publisher={APS}
}

@article{wang2005beating,
  title={Beating the photon-number-splitting attack in practical quantum cryptography},
  author={Wang, Xiang-Bin},
  journal={Physical Review Letters},
  volume={94},
  number={23},
  pages={230503},
  year={2005},
  publisher={APS}
}

@article{curras-lorenzoRigorousPhaseerrorestimation2026,
  title={Rigorous phase-error-estimation security framework for QKD with correlated sources},
  author={Curr{\'a}s-Lorenzo, Guillermo and Pereira, Margarida and Tamaki, Kiyoshi and Curty, Marcos},
  journal={preprint arXiv:2601.08417},
  year={2026}
}

@article{xu2015experimental,
  title={Experimental quantum key distribution with source flaws},
  author={Xu, Feihu and Wei, Kejin and Sajeed, Shihan and Kaiser, Sarah and Sun, Shihai and Tang, Zhiyuan and Qian, Li and Makarov, Vadim and Lo, Hoi-Kwong},
  journal={Physical Review A},
  volume={92},
  number={3},
  pages={032305},
  year={2015},
  publisher={APS}
}

@article{honjo2004differential,
  title={Differential-phase-shift quantum key distribution experiment with a planar light-wave circuit {M}ach--{Z}ehnder interferometer},
  author={Honjo, T and Inoue, K and Takahashi, H},
  journal={Optics Letters},
  volume={29},
  number={23},
  pages={2797--2799},
  year={2004},
  publisher={Optical Society of America}
}

@misc{NSA_QKD,
  author       = {{US National Security Agency}},
  title        = {Quantum Key Distribution ({QKD}) and Quantum Cryptography ({QC})},
  year         = {2020},
  howpublished = {\url{https://www.nsa.gov/Cybersecurity/Quantum-Key-Distribution-QKD-and-Quantum-Cryptography-QC/}},
  note         = {Accessed: 2026-04-21},
}

@misc{NCSC_QKD,
  author       = {{UK National Cyber Security Centre}},
  title        = {Quantum Security Technologies},
  year         = {2020},
  howpublished = {\url{https://www.ncsc.gov.uk/whitepaper/quantum-security-technologies}},
  note         = {Accessed: 2026-04-21},
}

@misc{ANSSI_BSI_NLNCSA_QKD,
  author       = {{ANSSI} and {BSI} and {NLNCSA} and {Swedish NCSA}},
  title        = {Position Paper on Quantum Key Distribution},
  year         = {2024},
  howpublished = {\url{https://www.bsi.bund.de/SharedDocs/Downloads/EN/BSI/Crypto/Quantum_Positionspapier.pdf}},
  note         = {Accessed: 2026-04-21},
}

@misc{mathworks_global_optimization_toolbox_2025,
  author       = {{The MathWorks Inc.}},
  title        = {{Global Optimization Toolbox, version 25.1 (R2025a)}},
  year         = {2025},
}

@article{naharImperfectDetectors2026,
  title = {Imperfect Detectors for Adversarial Tasks with Applications to Quantum Key Distribution},
  author = {Nahar, Shlok and Tupkary, Devashish and L{\"u}tkenhaus, Norbert},
  year = 2026,
  month = mar,
  journal = {Quantum},
  volume = {10},
  pages = {2044},
  publisher = {Verein zur F\"orderung des Open Access Publizierens in den Quantenwissenschaften},
  doi = {10.22331/q-2026-03-24-2044},
  urldate = {2026-04-08},
  langid = {british}
}

@article{wangPhaseError2025,
  title = {Phase Error Estimation for Passive Detection Setups with Imperfections and Memory Effects},
  author = {Wang, Zhiyao and Tupkary, Devashish and Nahar, Shlok},
  year = 2025,
  journal = {preprint arXiv:2508.21486},
}

@article{grunenfelder2020performance,
  author  = {Gr\"{u}nenfelder, Fadri and Boaron, Alberto and Rusca, Davide and Martin, Anthony and Zbinden, Hugo},
  title   = {Performance and security of {5} {GHz} repetition rate polarization-based quantum key distribution},
  journal = {Appl. Phys. Lett.},
  volume  = {117},
  number  = {14},
  pages   = {144003},
  year    = {2020},
  doi     = {10.1063/5.0021468},
}

@article{roberts2018patterning,
  author  = {Roberts, George L. and Pittaluga, Mirko and Minder, Mariella and Lucamarini, Marco and Dynes, James F. and Yuan, Zhiliang L. and Shields, Andrew J.},
  title   = {Patterning-effect mitigating intensity modulator for secure decoy-state quantum key distribution},
  journal = {Opt. Lett.},
  volume  = {43},
  number  = {20},
  pages   = {5110--5113},
  year    = {2018},
  doi     = {10.1364/OL.43.005110},
}

@article{yoshino2018quantum,
  author  = {Yoshino, Ken-ichiro and Fujiwara, Mikio and Nakata, Kensuke and Sumiya, Tatsuya and Sasaki, Toshihiko and Takeoka, Masahiro and Sasaki, Masahide and Tajima, Akio and Koashi, Masato and Tomita, Akihisa},
  title   = {Quantum key distribution with an efficient countermeasure against correlated intensity fluctuations in optical pulses},
  journal = {npj Quantum Information},
  volume  = {4},
  number  = {1},
  pages   = {8},
  year    = {2018},
  doi     = {10.1038/s41534-017-0057-8},
}

@article{trefilov2025intensity,
  author  = {Trefilov, Daniil and Sixto, Xoel and Zapatero, V\'{i}ctor and Huang, Anqi and Curty, Marcos and Makarov, Vadim},
  title   = {Intensity correlations in decoy-state {BB84} quantum key distribution systems},
  journal = {Optica Quantum},
  volume  = {3},
  pages   = {417},
  year    = {2025},
  doi     = {10.1364/OPTICAQ.561759},
}

@article{curras2026securitydecoy,
  author  = {Curr{\'a}s-Lorenzo, Guillermo and Pereira, Margarida and Marcomini, Alessandro and Tamaki, Kiyoshi and Curty, Marcos},
  title   = {Security of decoy-state quantum key distribution with correlated bit-and-basis encoders},
  journal= {preprint arXiv:2605.11767},
  year    = {2026}
}

@article{naharImperfectPhase2023,
	title = {Imperfect Phase Randomization and Generalized Decoy-State Quantum Key Distribution},
	author = {Nahar, Shlok and Upadhyaya, Twesh and L{\"u}tkenhaus, Norbert},
	year = 2023,
	month = dec,
	journal = {Physical Review Applied},
	volume = {20},
	number = {6},
	pages = {064031},
	publisher = {American Physical Society},
	doi = {10.1103/PhysRevApplied.20.064031}
}

@article{caoDiscretephaserandomizedCoherent2015,
	title = {Discrete-Phase-Randomized Coherent State Source and Its Application in Quantum Key Distribution},
	author = {Cao, Zhu and Zhang, Zhen and Lo, Hoi-Kwong and Ma, Xiongfeng},
	year = 2015,
	month = may,
	journal = {New Journal of Physics},
	volume = {17},
	number = {5},
	pages = {053014},
	issn = {1367-2630},
	doi = {10.1088/1367-2630/17/5/053014}
}
\end{document}